\documentclass[11pt, oneside]{mnthesis}
\usepackage{epsfig} 
\usepackage{epic} 
\usepackage{eepic} 
\usepackage{units} 
\usepackage{url} 
\usepackage{longtable} 
\usepackage{mathrsfs} 
\usepackage{multirow} 
\usepackage{bigstrut} 
\usepackage{amssymb} 
\usepackage{graphicx} 
\usepackage{setspace} 
\usepackage{xspace} 
\usepackage{amsmath} 
\usepackage{siunitx} 
\usepackage{booktabs} 
\usepackage{hyperref} 
\usepackage[noabbrev]{cleveref} 

\DeclareSIUnit\electronvolt{e\hspace{-0.08em}V}

\usepackage{amsthm}
\usepackage{leftidx}
\usepackage{tikz-cd}

\newtheorem{lem}{Lemma}[section]
\newtheorem{thm}{Theorem}[section]
\newtheorem{prop}{Proposition}[section]
\newtheorem{cor}{Corollary}[section]
\newtheorem{defx}{Definition}[section]
\newtheorem{ex}{Example}[section]

\newcommand{\w}{\omega}
\newcommand{\J}{\mathcal{J}}
\newcommand{\R}{\mathbb{R}}
\newcommand{\Cx}{\mathbb{C}} 
\newcommand{\wre}{\omega_{\text{red}}}

\newcommand{\nr}{n_R}

\newcommand{\oo}{\mathcal{O}}
\newcommand{\tth}{\theta}
\newcommand{\ber}{\text{Ber}\,}

\newcommand{\tred}{\text{red}}

\newcommand{\msp}{\mathfrak{M}_{g;\nr}}

\newcommand{\stsh}{\mathcal{O}}
\newcommand{\parx}{\frac{\partial}{\partial x}}

\newcommand{\parxi}{\frac{\partial}{\partial \xi}}

\newcommand{\dis}{\mathcal{D}}

\newcommand{\F}{\mathcal{F}}
\newcommand{\Ll}{\mathcal{L}}
\newcommand{\Kk}{\mathcal{K}}
\newcommand{\ee}{\mathcal{E}}
\newcommand{\bz}{\bar{z}}

\newcommand{\bpk}{[\partial_{z_k} \, | \, \partial_{\tth_k}]}

\newcommand{\M}{\mathcal{M}}
\newcommand{\nn}{\text{Norm}}
\newcommand{\Pp}{\mathbb{P}}

\begin{document}
\bibliographystyle{hunsrt} 

\makeatletter\@addtoreset{chapter}{part}\makeatother




\phd 

%
\title{\bf The Super Mumford Form in the Presence of Ramond and Neveu-Schwarz Punctures}
\author{Daniel Joseph Diroff}
\director{Alexander A. Voronov}

\submissionmonth{July}
\submissionyear{2019}

\abstract{

We generalize the result of \cite{vormum} to give an expression for the super Mumford form $\mu$ on the moduli spaces of super Riemann surfaces with Ramond and Neveu-Schwarz punctures. In the Ramond case we take the number of punctures to be large compared to the genus. We consider for the case of Neveu-Schwarz punctures the super Mumford form over the component of the moduli space corresponding to an odd spin structure. The super Mumford form $\mu$ can be used to create a measure whose integral computes scattering amplitudes of superstring theory. We express $\mu$ in terms of local bases of $H^0(X, \w^j)$ for $\w$ the Berezinian line bundle of a family of super Riemann surfaces.}

\copyrightpage       

\acknowledgements{

I'd like to thank the organizers of the 2015 Supermoduli workshop at the Simons Center for Geometry and Physics and those who put in the effort to put those excellent lectures online. I am grateful to the speakers of the workshop R. Donagi, P. Deligne, E. Witten, E. D'Hoker and D. H. Phong whose efforts illuminated many interesting concepts. I'd especially like to thank E. Witten for his insight and several valuable comments. Most of all, I'd like to thank A. Voronov for introducing me to the subject, offering guidance and for the frequent helpful discussions. 

}


\beforepreface


\afterpreface


\include{chapters/intro_true}





\chapter{Introduction} \label{intro}
\label{intro}

Due to relatively recent computations done by E. D'Hoker and D. H. Phong \cite{DPhong} and new ideas pushed forward by E. Witten \cite{wit4}, the role of supergeometry in superstring perturbation theory has been revived from what it once was in the 1980s. However, the task of computing superstring scattering amplitudes have proved difficult due to many complications boiling down to the fact that the underlying supergeometry was not completely understood. 

Scattering amplitudes in superstring theory are expressed as Berezin integrals over various moduli spaces of super Riemann surfaces. One might hope that such integrals would be computable via expressing supermoduli space as a fiber bundle over a bosonic reduced space, allowing one to integrate in the odd directions fiberwise. In fact, this is exactly the technique utilized in the D'Hoker and Phong results. However, this assumption was only valid for low genus, as it was shown in a recent paper by R. Donagi and E. Witten \cite{donwit} that in general supermoduli space \emph{is not} a fiber bundle over its reduced space. This notion is significant in supergeometry and is known as \emph{splitness}.

Essentially, one says a supermanifold is split if it can be expressed as such a fiber bundle over a bosonic base. It is known that every $C^{\infty}$ supermanifold is indeed split \cite{man1}. Thus in principal the theory of smooth supermanifolds is contained in the theory of exterior algebra vector bundles over a smooth manifold. However, holomorphic methods have proved to be very useful in studying super Riemann surfaces and their moduli as holomorphic or complex supermanifolds need not be split. Thus holomorphic supergeometry is central in understanding computations of superstring scattering amplitudes.

In bosonic string theory, the $g$ loop contribution to the partition function can be written as the integral
$$ Z_g = \int_{\mathcal{M}_g} d\pi_g, $$
where $\mathcal{M}_g$ is the usual moduli stack of Riemann surfaces of genus $g$ and $d\pi_g$ is the so-called Polyakov measure. Suppose we have a universal family $\mathcal{C}_g$ over $\mathcal{M}_g$ and let $\pi: \mathcal{C}_g \to \mathcal{M}_g$ denote the projection. 

In a famous theorem due to Belavin and Knizhnik, the Polyakov measure was shown to be the modulus squared of a trivializing section of a holomorphic line bundle on $\mathcal{M}_g$,
$$ d\pi_g = \mu_g \wedge \overline{\mu}_g. $$
The form $\mu_g$ is called a Mumford form and it is a section exhibiting the Mumford isomorphism
$$ \left( \text{det }  \pi_*\Omega \right)^{13} \otimes \left( \text{det }  R^1\pi_*\Omega \right)^{-13} \cong \det \pi_*\Omega^2 \otimes \left( \det R^1\pi_*\Omega^2 \right)^{-1},$$
where $\Omega$ is the sheaf of relative differentials on $\mathcal{C}_g$. Here and henceforth, powers of vector bundles, sheaves and vector spaces stand for tensor powers.

In the super case, the object one integrates over in computations of superstring scattering amplitudes is slightly more complicated than simply $\mathfrak{M}_g$, see \cite{wit1}. Nevertheless there still is a relevant canonical super Mumford isomorphism,
$$ (\ber \pi_*\w )^5 \otimes (\ber R^1\pi_*\w )^{-5} \cong \ber \pi_*\w^3 \otimes \left( \ber R^1\pi_*\w^3 \right )^{-1}  $$
for $\w$ the relative Berezinian sheaf of a family of super Riemann surfaces of genus $g$. The trivializing section that exhibits the above isomorphism is called the super Mumford form. Such a form is useful in the super case in very much the same way as that of the bosonic Mumford form, as sections of $\ber \pi_*\w^3$ are super volume forms on $\mathfrak{M}_g$. In a paper by A. Voronov \cite{vormum}, an explicit formula of the super Mumford form was computed over the odd-spin component of $\mathfrak{M}_g$. 

In this paper we expand on those ideas and produce explicit formulas for the analogous super Mumford forms over the moduli spaces $\mathfrak{M}_{g;n_R}$ and $\mathfrak{M}_{g;n_{NS}}$ of genus $g \geq 2$ super Riemann surfaces with Ramond or Neveu-Schwarz punctures. In both cases we work under some assumptions regarding the local freeness of the sheaves $R^i\pi_*\w^j$. The specifics are given at the end of Section \ref{susy_prelim}. In the Ramond case we furthermore impose the condition that the number of Ramond punctures $\nr$ be strictly greater than $6g-6$. 

We then discuss how these formulae give rise to a physically relevant measure. By explicit formulas, we mean those written in terms of chosen sections of natural sheaves defined on the moduli spaces.

The main results (Theorem \ref{mainthm1} and Corollary \ref{maincor}) are found in Chapters \ref{r_case} and \ref{ns_case} where the explicit formulas of the relevant super Mumford forms are presented. A review of the basic theories of super mathematics are presented in Chapters \ref{salg} and \ref{sgeo}. Chapter \ref{susy_prelim} is devoted to presenting the neccessary theory of super Riemann surfaces needed for the main results of the paper.  Appendices appear after in Chapter \ref{appendices} containing a few technical lemmas used in the main arguments as well as a general proof of the super Mumford isomorphism. This work has been published and will appear in the Journal of Geometry and Physics \cite{dan_paper}. 

\chapter{Superalgebra} \label{salg}
\label{superalg}


\section{Super Linear Algebra}

We begin by defining the basic algebraic objects one works with in super geometry. For the remainder of the section we let $k$ be a field of characteristic not equal to $2$. We closely follow the Chapter 1 of \cite{qfst}.

\bigskip

\subsection{Super Vector Spaces}~

A \emph{super $k$-vector space} $V$ is a $\mathbb{Z}/2\mathbb{Z}=\mathbb{Z}_2$ graded vector space over $k$,
$$
V = V_0 \oplus V_1.
$$
Elements $v \in V_i$ are called \emph{homogeneous}. If $v \in V_0$ then it is called \emph{even} and if $v \in V_1$ it is called \emph{odd}. For any homogeneous $v$ we denote by $|v|$ its degree, also called its parity. Frequently we will use the notation $v = v_0 + v_1$ to denote the decomposition of an arbitrary element $v \in V$ into its even and odd parts.

A map between two super vector spaces $T: V \to W$ is a linear map that preserves the grading, $T(V_0) \subset W_0$ and $T(V_1) \subset W_1$. We can then see that super vector spaces over $k$ form an abelian category. We define the dimension $\text{dim}\, V$ to be the pair of integers
$$
\text{dim}\,  V = \text{dim} \, V_0 \, | \, \text{dim}\,  V_1
$$
and the superdimension $\text{sdim} V$ by the single integer
$$
\text{sdim}\, V = \text{dim} \, V_0 - \text{dim}\, V_1.
$$

We have a \emph{parity reversing functor} $\Pi$ taking a supervector space $V$ to $\Pi V$ defined by
$$
(\Pi V)_0 = V_1, \hspace{.5cm} (\Pi V)_1 = V_0.
$$
Super vector spaces admit tensor products defined in the obvious way
$$
(V \otimes W)_0 = (V_0 \otimes W_0) \oplus (V_1 \otimes W_1)
$$
$$
(V \otimes W)_1 = (V_0 \otimes W_1) \oplus (V_0 \otimes W_1).
$$
Important in super algebra is the \emph{sign rule} which is a specific choice of commutivity isomorphism different from the classical one
$$
c_{V,W}: V \otimes W \to W \otimes V
$$
$$
v \otimes w \to (-1)^{|v| |w| } w \otimes v.
$$

As is common in the subject, when giving definitions one frequently works with homogeneous elements and it is understood to extend by linearity.

We denote by Hom$(V,W)$ the set of all linear maps $T: V \to W$ that preserve the grading. $\text{Hom}(V,W)$ is best understood as a usual vector space, or a super vector space with trivial odd part. The category of super vector spaces admits an internal hom, denoted \underline{Hom}$(V,W)$ defined by simply considering all $k$-linear maps graded via
$$
\underline{\text{Hom}}(V,W)_0 = \text{Hom}(V,W)
$$
$$
\underline{\text{Hom}}(V,W)_1 = \text{Hom}(V,\Pi W) = \text{Hom}(\Pi V, W).
$$
In other words \emph{even} elements of $\underline{\text{Hom}}(V,W)$ (also called even maps) are linear maps that preserve the $\mathbb{Z}_2$-grading, while \emph{odd} elements (called odd maps) reverse the grading. The dual $V^*$ of a super vector space $V$ is defined then by $V^* := \underline{\text{Hom}}(V,k)$.

A linear map $T: V \to W$ after choosing homogeneous bases for $V$ and $W$, corresponds to a block matrix
$$
T \sim
\begin{pmatrix}
T_{00} & T_{01} \\
T_{10} & T_{11}
\end{pmatrix}
$$
The decomposition of $T$ into its even and odd parts $T = T_0 + T_1$ then corresponds to 
$$
\begin{pmatrix}
T_{00} & T_{01} \\
T_{10} & T_{11}
\end{pmatrix}
=
\begin{pmatrix}
T_{00} & 0 \\
0 & T_{11}
\end{pmatrix}
+
\begin{pmatrix}
0 & T_{01} \\
T_{10} & 0
\end{pmatrix}.
$$

\begin{ex} $k^{m|n}$
\end{ex}

The most basic example of a super vector space is simply 
$$
k^{m|n} := \bigoplus^m k \oplus \left ( \bigoplus^n \Pi k \right ).
$$
We caution that the notation $k^{m|n}$ is used in two distinct ways in this context and one must take care to avoid confusion. In what follows the use will be clear from context.

\bigskip

\section{Superalgebras and Modules over Them}

A \emph{superalgebra} $A$ is a super vector space together with a super vector space morphism $A \otimes A \to A$ called the product. For our purposes we shall always assume that the superalgebras we consider are associative and possess a unit. We say that $A$ is \emph{super commutative} (or simply commutative) if the product morphism commutes with the commutivity isomorphism $c_{A,A}$. Specifically this is the requirement
\begin{equation}
ab = (-1)^{|a| |b|} ba
\end{equation}
for homogeneous $a,b$.

The tensor product $A \otimes B$ of super superalgebras is again a superalgebra with product 
$$
(a \otimes b)(a' \otimes b') = (-1)^{|a'| |b|} aa' \otimes bb'.
$$
This is an example of the general rule of thumb in super mathematics that when two quantities are swapped a power of $-1$ to the product of their parities appears. This general philosophy can help one in keeping signs straight later on.

In any superalgebra $A$ we have the \emph{supercommutator} (or simply commutator/bracket) $[\cdot, \cdot]$
$$
[a,b] := ab - (-1)^{|a| |b|}ba
$$
so that $A$ is supercommutative if and only if the supercommutator is trivial. The \emph{anti-supercommutator} (or anti-commutator/anti-bracket) is
$$
\{a , b\} := ab + (-1)^{|a| |b|}ba.
$$
We denote by $[A,A]$ the sub-superalgebra of $A$ generated by all expressions of the form $[a,a]$ and similarly for $\{A,A\}$.

Suppose now $V$ is a super vector space. Let $T(V)$ denote the tensor algebra of $V$, $T(V) = \oplus_n (V^{\otimes n})$ with the usual algebra structure given by concatenation of tensors, then define the \emph{symmetric} $S(V)$ and \emph{exterior} $\bigwedge (V)$ algebras of $V$ as follows
$$
S(V) := T(V)/(v \otimes w - (-1)^{|v| |w|} w \otimes v)
$$
$$
\bigwedge (V) := T(V)/(v \otimes w + (-1)^{|v| |w|} w \otimes v)
$$
where the denominators in the above quotients are the ideals generated by all expressions of those forms for $v,w \in V$. Notice that because of the sign rule we get the peculiar formula $S(k^{0|n}) = \Pi \bigwedge(k^{n|0})$, or in more generality, $S(\Pi V) = \Pi \bigwedge (V)$.

The algebras $T(V), S(V)$ and $\bigwedge(V)$ are all naturally $\mathbb{Z}$-graded and hence can also be thought of as superalgebras. Sometimes the notation $T^{\bullet}(V) = T(V)$ is used to emphasise the grading.

\begin{ex} The Sheaf of Differential Forms $\Omega_M^{\bullet}$
\end{ex} On a manifold $M$ (real or complex) one can consider the sheaf of differential forms. This is a sheaf of $\mathbb{Z}$-graded algebras that are graded commutative. Reduction modulo two of the grading allows one to view $\Omega_M$ as a sheaf of supercommutative algebras.

\begin{ex} The Polynomial Superalgebra $A = k[x_1,\dots,x_m \, | \, \xi_1, \dots, \xi_n ]$.
\end{ex} The superalgebra $k[x_1, \dots, x_m \, | \, \xi_1, \dots, \xi_m]$ is the free superalgebra generated by the even quantities $x_i$ and odd quantities $\xi_i$. Concretely every element $f \in A$ can be written
$$
f = \sum_{I} f_I(x_1,\dots,x_m) \xi^I
$$
where $I \subset \{1, \dots, n \}$ is a multi-index and $\xi^I = \xi_{i_1}\dots\xi_{i_k}$ if $I = \{i_1 < \dots < i_k \}$. In fact a more "coordinate" free viewpoint would be to consider $S(V^{m|n})$ where $V^{m|n}$ is any $m| n$ dimensional super vector space. After choosing a basis of $V^{m|n}$, one can construct an isomorphism $S(V^{m|n}) \cong A$. This example will be ubiquitous is what follows.

\bigskip

\subsection{Modules over Supercommutative Algebras}~

Let $A$ be a superalgebra. $M$ is a \emph{supermodule} over $A$ if it is a module in the usual sense over $A$, is $\mathbb{Z}_2$-graded $M = M_0 \oplus M_1$ and if the module structure respects parity in the sense that $|am| = |a||m|$ for all $a\in A, m \in M$ homogeneous. We will typically refer to $M$ as simply a module over $A$ and omit the phrase "super". All of the usual notions/operations of modules have super analogues and one can write them down explicitly if they follow the philosophy of the sign rule. For instance the tensor product $M \otimes N$ of two modules over $A$ makes sense as a supermodule with $\mathbb{Z}_2$-grading identical to that of super vector spaces. 

For any positive integers $m,n$, we let $A^{m | n}$ denote the $A$ module
$$
A^{m | n} = \bigoplus_{i=1}^n A \oplus \bigoplus_{i=1}^m \left( \Pi A \right ). 
$$
We say an $A$ module $M$ is free if it is isomorphic to $A^{m|n}$ for some $m,n$. We will not be concerned with free modules of infinite rank. The pair $m |n$ is the \emph{rank} of the module $M$.

The dual of an $A$ module $M$ is $M^* = \underline{\text{Hom}}(M,A)$. Care must be taken when discussing the canonical evaluation maps
$$
M^* \otimes M \to A, \hspace{.5cm} M \otimes M^* \to A.
$$
According to the sign rule the first map is the usual one $f \otimes m \to f(m)$ while the second is $m \otimes f \to (-1)^{|f| |m| } f(m)$ is potentially different by a sign. Of course these are compatible with the commutivity isomorphism $c_{M^*, M}$. If $M$ is free of finite rank then one has the usual isomorphism
$$
M \otimes M^* \cong \underline{\text{End}}(M)
$$
$$
m \otimes f \mapsto (m' \mapsto mf(m')),
$$
however in view of the two possible orderings above this is arguable not the most natural isomorphism. In this light we will in this paper work systematically with evaluation map $M^* \otimes M \to A$ in this order, which does produce a sign in the isomorphism
\begin{equation} \label{endiso}
M^* \otimes M \cong \underline{\text{End}}(M)
\end{equation}
$$
(f \otimes m) \mapsto \left ( m' \mapsto (-1)^{|m||f|}mf(m') \right ).
$$

If $M$ is free with homogeneous basis $\{ e_1, \dots, e_m \, | \, \tth_1, \dots, \tth_n \}$ then we define the \emph{left dual} basis or simply the \emph{dual} basis $\{ e_1^*,\dots,e_m^* \, | \, \tth_1^*,\dots,\tth_m^* \}$ for $M^*$ by the usual relations
$$
e^*_j(e_i) = \delta_{ij}, \hspace{.5cm} \tth_j^*(\tth_j) = \delta_{ij}
$$
$$
e^*_j(\tth_i) = 0, \hspace{.5cm} \tth_j^*(e_j) = 0.
$$
We emphasize we are choosing the dual basis to be natural with respect to $M^* \otimes M \to A$ rather than $M \otimes M^* \to A$ which would lead to a notion of a \emph{right} dual basis and essentially amount to changing the $\tth_j^*$ by a sign.

Given a morphism $T: M \to N$ between two $A$ modules in classical algebra, one has a canonical dual map $T^*: N^* \to M^*$, which differs from the usual formula by following the sign rule,
\begin{equation}
T^*(f)(m) = (-1)^{|T| |f|}f(T(m)),
\end{equation}
if $T$ is an arbitrary (not necessarily even nor odd) map from $M \to N$. This will produce a slightly different form for the "transpose" of a matrix in superalgebra.

\bigskip

\section{The Berezinian and Related Constructions}

\subsection{The Supertranspose}~

Suppose $M = M^{p|q}$ and $N = N^{r|s}$ are free $A$-modules of ranks $p|q$ and $r|s$ respectively. Suppose we have an $A$-module map $T: M \to N$. Let $e_1 \dots e_{p+q}$ and $f_1, \dots, f_{r+s}$ be homogeneous bases of $M$ and $N$ respectively such that the first $p$ (resp. $r$) elements are even and the last $q$ (resp. $s$) are odd. To each such morphism $T$ and choices of bases we associate the $(r+s) \times (p+q)$ matrix $[T] = (T^i_j)$ defined by the equations
$$
Te_j = \sum_{i=1}^{r+s} f_i \, T_{j}^i.
$$

Note how the entries $T^i_j$ appear \emph{to the right} of the basis elements $f_i$. We aim now to relate the entries of the matrix $[T]$ with those of the matrix of its dual map $[T^*]$ with respect to the dual of the chosen bases. The relationship derived below will then be regarded as the \emph{supertranspose} and then will be thought off as an operation on super matrices. Indeed by definition the matrix $[T^*] = (T^{*i}_j)$ is defined by the equations
$$
T^*f^*_j = \sum_{i=1}^{p+q} e^*_i \, T^{*i}_j.
$$
This allows us to compute the expression $T^*f_j^*(e_i)$ in two ways. First see that
\begin{equation}
\begin{split}
T^*f_j^*(e_i) & =  \left( \sum_{k=1}^{p+q} e^*_k \, T^{*k}_j \right ) (e_i) \\
& = (-1)^{|T^{*i}_j| |e_i^*|} T^{*i}_j.
\end{split}
\end{equation}
On the other hand,
\begin{equation}
\begin{split}
T^*f_j^*(e_i) & =  (-1)^{|T| |f_j^*|}f_j^*(T(e_i)) \\
& = (-1)^{|T| |f_j^*|}f_j^* \left( \sum_{k=1}^{r+s} f_k \, T_{i}^k \right ) \\
& = (-1)^{|T| |f_j^*|} T^{j}_i.
\end{split}
\end{equation}
In the end, we conclude
\begin{equation} \label{stranspose}
T^{*i}_j = (-1)^{|T^{*i}_j| |e_i^*| + |T| |f_j^*|} T^j_i.
\end{equation}
Thus, if we write $[T]$ in block matrix form
$$
[T] =  
\begin{pmatrix}
A & B \\
C & D
\end{pmatrix}.
$$
then the matrix of $[T^*]$ is easy to identify according to (\ref{stranspose}). If $T$ is an even $|T| = 0$, then
$$
[T^*] =  
\begin{pmatrix}
A^t & C^t \\
-B^t & D^t
\end{pmatrix},
$$
with $A^t$ denoting the usual transpose of $A$. If $T$ is odd $|T| = 1$,
$$
[T^*] =  
\begin{pmatrix}
A^t & -C^t \\
B^t & D^t
\end{pmatrix}.
$$
This motivates the general definition of the \emph{supertranspose} of a general supermatrix $X$ denoted by $X^{st}$, writing $X$ is block form as above,
$$
X^{st} :=  
\begin{pmatrix}
A & B \\
C & D
\end{pmatrix}^{st}
=
\begin{pmatrix}
A_0^t & C_1^t \\
-B_1^t & D_0^t
\end{pmatrix}
+
\begin{pmatrix}
A_1^t & -C_0^t \\
B_0^t & D_1^t
\end{pmatrix}.
$$
Where the block matrices above were decomposed into their even and odd parts. We then have $[T^*] = [T]^{st}$.

\medskip

\subsection{The Supertrace}~

The evaluation map $M^* \otimes M \to A$ gives rise to an $A$-linear map $\underline{\text{Hom}}(M,M) = \underline{\text{End}}(M) \to A$ via the identification (\ref{endiso}) above. This map is called the \emph{supertrace}. In terms of the matrix
$$
X = 
\begin{pmatrix}
A & B \\
C & D
\end{pmatrix},
$$
this is simply
$$
\text{str}\, X := \text{tr}\,A - \text{tr}\,D.
$$
The supertrace posseses several nice properties such as 

\begin{enumerate}
\item $\text{str}\, X^{st} = \text{str}\,X$ \\ 
\item $\text{str} \, X X' = (-1)^{|X| |X'|}\text{str} \, X' X $.
\end{enumerate}

Others can be found, with proof, in \cite{man1}, however most should not concern us here.

\medskip

\subsection{The Berezinian of an Automorphism}~

\medskip

In supermathematics, the Berezinian is an analog of the classical determinant, it plays a vital role in what follows. We let $GL(p|q, A)$ denote the group of all automorphisms of the free $A$-module $A^{p|q}$, then $GL(p|q, A)$ is naturally identified with all invertible $(p+q) \times (p+q)$ matrices $X$ which we write in the standard block form,
$$
X =
\begin{pmatrix}
A & B \\
C & D
\end{pmatrix}.
$$
We then define the Berezinian of $X$ to be
\begin{equation}
\ber X = 
\ber 
\begin{pmatrix}
A & B \\
C & D
\end{pmatrix}
:= \text{det}\, (A - BD^{-1}C) (\text{det}\,D)^{-1}.
\end{equation}
Note that this definition makes sense as the matrix $X$ is invertible if and only if both $A$ and $D$ are. Clearly $\ber T \in A^{\times}_0$ is an even invertible element of $A$ and in fact the Berezinian gives a group homomorphism 
$$ 
GL(p|q, A) \to GL(1|0, A) = A_0^{\times}
$$ 
$$
\ber XX' = \ber X \, \ber X'.
$$
The proof of this fact can be found in many places including \cite{man1}, the argument is somewhat tedious and we will not show it here. However taking for granted the multiplicative property of the Berezinian, one can motivate the definition (7) by defining for $X$ strictly upper or lower triangluar, $\ber X = \det A (\det D)^{-1}$, and notice the trivial factorization
\begin{equation}
\begin{pmatrix}
A & B \\
C & D
\end{pmatrix}
=
\begin{pmatrix}
I & BD^{-1} \\
0 & I
\end{pmatrix}
\begin{pmatrix}
A - BD^{-1}C & 0 \\
0 & D
\end{pmatrix}
\begin{pmatrix}
I & 0 \\
D^{-1}C & I
\end{pmatrix},
\end{equation}
we remark again that $X$ invertible implies both $A$ and $D$ are invertible. In fact, an alternative factorization to $(8)$ exists
\begin{equation}
\begin{pmatrix}
A & B \\
C & D
\end{pmatrix}
=
\begin{pmatrix}
I & 0 \\
CA^{-1} & I
\end{pmatrix}
\begin{pmatrix}
A & 0 \\
0 & D-CA^{-1}B
\end{pmatrix}
\begin{pmatrix}
I & A^{-1}B \\
0 & I
\end{pmatrix}.
\end{equation}
This then yields an alternative calculation for the berezininan
$$
\ber X
=
\ber 
\begin{pmatrix}
A & B \\
C & D
\end{pmatrix}
= \text{det}\, A (\text{det}\,(D - CA^{-1}B))^{-1}.
$$
The analogy with the classical determinant can be seen as one has
\begin{enumerate}
\item $\ber T^{st} = \ber T$ \\
\item $\ber \text{exp} \,(T) = \text{exp}\, (\text{str}\, T).$
\end{enumerate}

In general for any even invertible automorphism $T$ of a free $A$-module $M$, we define $\ber T$ to be $\ber [T]$ where the matrix $[T]$ is expressed in any basis. The fact that the Berezinian is a group homomorphism implies that this is well defined. 

\medskip

\subsection{The Berezinian of a Free Module}~

In classical algebra, given a free $A$-module $M$ of rank $n$ there is the notion of the \emph{determinant} of $M$, $\text{det}\, M$ given as the maximal exterior power $\bigwedge^n M$. In this context we have a useful interpretation of the determinant of an automorphism $T$, namely as the action of $T$ on $\text{det}\, M$. It is possible to find analogous statements regarding the Berezinian.

Returning to the super case, for $M$ free of rank $p|q$ with $q>0$, there \emph{is no} top exterior power of $M$, simply because odd generators commute (in the classical sense) in $\bigwedge ^{\bullet} M$ and thus, for example, given an odd generator $\theta$ of $M$, $\theta^k$ does not vanish for any $k$. Nevertheless a super analog of the determinant exists and we denote it $\ber M$.

We explain two approaches as explained in \cite{qfst}, the first is a concrete realization. To every homogeneous basis $e_1,\dots e_p \, | \, f_1, \dots f_q$ we obtain an element of $\ber M$ denoted
$$
[e_1, \dots, e_p \, | \, f_1, \dots, f_q],
$$
subject to the relations given by
$$
[Te_1, \dots, Te_p \, | \, Tf_1, \dots, Tf_q] = \ber T \, [e_1, \dots, e_p \, | \, f_1, \dots, f_q],
$$
for $T$ an automorphism of $M$. The $A$-module construted will be considered of rank $1|0$ if $q$ is even and $0|1$ if $q$ is odd. This realization is the one we will most frequently use. 

One should notice that applying the above procedure to a free $M$ of rank $p|0$ over an purely even algebra $A = A_0$ that this recovers the classical determinant.

Alternatively a basis independent definition was given in \cite{man1} and discussed in \cite{qfst}. The motivation for the following is that in the non-super case, for an ordinary free $A$-module $M$ or rank $n$, one can see by a Koszul complex a canonical isomorphism
$$
\text{Ext}^n_{S^{\bullet}(M^*)}(A, S^{\bullet}(M^*)) = \bigwedge^n M,
$$
where $A$ is given the structure of a $S^{\bullet}(M^*)$-module via augmentation. 

The above expression can be understood in the super setting. That is, for $M$ free of rank $p|q$ over a supercommutative algebra $A$ we set
$$
\ber M := \underline{\text{Ext}}^n_{S^{\bullet}(M^*)}(A, S^{\bullet}(M^*)),
$$
and for any even automorphism $T:M \to M$ we set
$$
\ber T = \text{ action of T on } \underline{\text{Ext}}^n_{S^{\bullet}(M^*)}(A, S^{\bullet}(M^*)).
$$

In greater generality, if $T: M \to N$ is an isomorphism between two $A$-modules then $T$ induces a map $\ber M \to \ber N$ also called $\ber T$. In the case that $M = N$, such a map $\ber T: \ber M \to \ber M$ can be naturally identified with an even invertible element of $A$, and this agrees with the definition of $\ber T$ above.


\chapter{Supergeometry} \label{sgeo}
\label{algsupergeo}




\section{Superspaces and Superschemes}

The various different super-geometric categories we will work in will concern objects which are specializations of the notion of a \emph{superspace}. We follow closely the notation and notions given in \cite{man1}, \cite{dan_r_beamer}.

\begin{defx}
A superspace is a locally ringed spaced $(X, \oo_X)$ where $\oo_X$ is a sheaf of super-commutative rings.
\end{defx}

A morphism of superspaces is the usual one; a morphism $f: X \to Y$ is a continuous map $|f|: X \to Y$ of underlying topological spaces along with a map of sheaves $\oo_Y \to f_*\oo_X$ such that for any point $x \in X$, the stalk morphsim $f_x: \oo_{Y, f(x)} \to \oo_{X,x}$ is local
$$
f_x(\mathfrak{m}_{f(x)}) \subset \mathfrak{m}_x.
$$

The structure sheaf of a superspace $X$ is $\mathbb{Z}/2$ graded $\oo_X = \oo_{X,0} \oplus \oo_{X,1}$ and contains a subsheaf $\mathcal{J}_X = \oo_{X,1} \oplus \oo_{X,1}^2$ generated by all odd elements. We sometimes omit the subscript $X$ and write $\mathcal{J} = \mathcal{J}_X$ when it is clear from context.

We have a natural morphism
$$
(X, \oo_X/\mathcal{J}) \to (X, \oo_X)
$$
corresponding to the projection $\oo_X \to \oo_X/\mathcal{J}$ and refer to the superspace $(X, \oo_X/\mathcal{J})$ as the \emph{odd-reduction} of $X$, $X_{\text{rd}}$. In most situations we will have equality of $X_{\text{rd}}$ and the usual reduction $X_{\tred}$ corresponding to the sheaf generated by all nilpotents $\mathcal{N}_X$.  In \cite{man1} the author discusses the distinctions between these two situations in more detail.

One also considers the \emph{split model of X} or the \emph{associated graded space of} $X$, namely the superspace
$$
\text{gr } X := (X, \text{gr } \oo_X )
$$
where
$$
\text{gr } \oo_X = \bigoplus_{n=0}^{\infty} \mathcal{J}^n/\mathcal{J}^{n+1} = \oo_X/\mathcal{J} \oplus \mathcal{J}/\mathcal{J}^2 \oplus \mathcal{J}^2/\mathcal{J}^3 \dots
$$

We have the following notions:

\begin{defx}
Let $X = (X, \oo_X)$ be a superspace and $\mathcal{J}$ denote the sheaf generated by all odd-elements of $\oo_X$. Then we say
\begin{enumerate}
    \item $X$ is locally split if $\ee := \mathcal{J}/\mathcal{J}^2$ is locally free as an $\oo_{X_{\text{rd}}} = \oo_X/\mathcal{J}$ module and $\bigwedge \ee \cong \text{gr}\,\, \oo_X$.
    \item $X$ is split if it is locally split and $X \cong \text{gr}\, X$ globally
    \item $X$ is projected if there exists a right inverse $p: X \to X_{\text{rd}}$ to the natural morphism $X_{\text{rd}} \to X$.
\end{enumerate}
\end{defx}

We will mostly be concerned with supermanifolds, but will sometimes find it useful to have a scheme-theoretic viewpoint in mind.

\begin{defx}
A superscheme $X$ is a superspace $(X, \oo_X)$ such that $(X, \oo_{X,0})$ is an ordinary scheme and $\oo_{X,1}$ is a coherent sheaf of $\oo_{X,0}$-modules.
\end{defx}

Superschemes can be covered by \emph{affine} superschemes $\text{Spec}(A)$ for supercommutative rings $A$. Much of classical algebraic geometry immediately generalizes to the supercase, therefore we do not pause here to elaborate. 

\begin{ex} (Affine Space $\mathbb{A}^{m|n}$) \label{aff_def}
\end{ex}
We discuss a natural assignment of an affine superscheme $\mathbb{A}_V$ given a super vector space $V = V_0 \oplus V_1$. Let $V^* = \underline{\text{Hom}}(V,\Cx)$ denote the internal Hom of all linear maps $V \to \Cx$. Let $S$ denote the graded ring $\underline{\text{Sym}}^{\bullet}(V^*)$, and consider this as a super algebra whose $\mathbb{Z}/2$-grading comes from reduction modulo $2$. One then has a canonical decomposition
$$
\underline{\text{Sym}}^{\bullet}(V^*) = \text{Sym}^{\bullet}(V^*_0) \otimes \bigwedge(V^*_1)
$$
where we take the exterior product in the usual classical sense. The affine superscheme associated to $V$ is then $\mathbb{A}_V = \text{Spec}(\underline{\text{Sym}}^{\bullet}(V^*))$ and in view of the decomposition above, it is immediately seen as split. In the special case $V = \Cx^{m|n} = \Cx^m \oplus (\Pi \Cx)^n$ we write $\mathbb{A}_V = \mathbb{A}_{\Cx}^{m|n} = \mathbb{C}^{m|n}$.

We remark that the seemingly harmless operation of parity change $\Pi$ on modules behaves non-trivially with respect to the functor $V \to \mathbb{A}_V$. Indeed topologically $\mathbb{A}_V$ is the same as the space associated to the classical vector space $V_0$, while for $\mathbb{A}_{\Pi V}$ it is that of $V_1$.

\section{Supermanifold Theory}

The most utilized notion for us is that of a \emph{supermanifold}.

\subsection{Basic Notions}~

Supergeometry can be thought of an extension of ordinary geometry where one adds extra "odd anti-commuting functions". For supermanifolds one will find essentially two different notions in the literature. We adopt the more algebro-geometric approach. In this paper we denote the sheaf of holomorphic (resp. smooth) functions on $\Cx^n$ ($\R^n$) by $\oo_{\Cx^n}$ ($\mathcal{C}^{\infty}_{\R^n}$). 
\begin{defx}
A complex (resp. smooth) supermanifold of dimension $m | n$ is a locally ringed space $(X , \stsh)$ with
\begin{enumerate}
    \item $X$ a second countable Hausdorff topological space,
    \item $\stsh$ a sheaf of supercommutative $\Cx$ (resp. $\R)$ algebras,
    \item $\stsh$ is locally isomorphic to $\mathcal{O}_{\mathbb{C}^m} \otimes \bigwedge(\xi_1, \dots, \xi_n)$ \\
    (resp. $\mathcal{C}^{\infty}_{\mathbb{R}^m} \otimes \bigwedge(\xi_1, \dots, \xi_n)$).
\end{enumerate}
\end{defx}

That is, as a superspace, a supermanifold $(X, \oo_X)$ is locally split. When it is clear from the context, $X$ will sometimes refer to the supermanifold $(X, \stsh)$ and $|X|$ will denote the underlying topological space. 

The majority of our analysis will concern complex supermanifolds. By $\mathbb{C}^{m|n}$ we mean the supermanifold $(\Cx^m,\stsh_{\Cx^{m|n}})$ whose structure sheaf is globally given by $\mathcal{O}_{\mathbb{C}^m} \otimes \bigwedge(\xi_1, \dots, \xi_n)$ and whose $\mathbb{Z}_2$ grading is determined by reduction modulo two of the standard $\mathbb{Z}$ grading of the exterior algebra. The $\xi_j$'s are referred to as the \emph{odd generators} or \emph{odd coordinates}. $\R^{m|n}$ has the analogous definition. If $U \subset \Cx^{m | n}$, then we call $U^{m | n} = (U, \stsh_{\Cx^{m|n}}{\big |}_U)$ a (an open) superdomain of $\Cx^{m|n}$.

Given any manifold (real or complex) $M$, and a vector bundle $F$ with sheaf of sections $\mathcal{F}$, the space $(M, \bigwedge^{\bullet}_{\stsh_M}\mathcal{F})$ is a supermanifold of dimension $\text{dim}\,  M \, | \, \text{rank}\,  \mathcal{F}$, where $\mathcal{F}$ is taken to be odd and the $\mathbb{Z}_2$ grading is the obvious one. Supermanifolds constructed in this way are split.

Let $X = (X, \oo_X)$ be a complex supermanifold and $x_1, \dots, x_m$ denote local coordinates in an open set $U$. Possibly shrinking this coordinate chart we find an isomorphism $\stsh|_U \cong \mathcal{O}_{\mathbb{C}^m}|_U \otimes_{\mathbb{C}} \bigwedge(\xi_1, \dots, \xi_n)$, then the collection $(x_1, \dots, x_m | \xi_1, \dots, \xi_n)$ are referred to as local coordinates on $(X,\stsh_X)$. Therefore, locally every super function $f$ is a Grassmann polynomial in the $\xi's$ with coefficients holomorphic functions of $x_1, \dots, x_m$, 
$$ f(x_1, \dots, x_m | \xi_1, \dots, \xi_n) = \sum_{ I \subset \{ 1,\dots,n \} } f_I(x_1, \dots, x_m)\xi_I
$$
where $\xi_I = \xi_{i_1} \cdots \xi_{i_k}$ if $I = \{i_1 < \dots < i_k \}$. The vertical bar in the argument of $f$ simply reminds the reader of the even and odd variables.

The subsheaf $\J$ of ideals generated by the odd part $\stsh_1$ is equivalently the sheaf generated by all nilpotent functions on a supermanifold. Thus $X_{\tred}$ is a classical manifold and $X_{\text{rd}} = X_{\tred}$.

A morphism $\phi : (X, \stsh_X) \to (Y, \stsh_Y)$ of supermanifolds is defined simply as above, as a  morphism of locally ringed spaces, hence it is given by a pair $(|\phi|, \phi^*)$ where $|\phi| : X \to Y$ is a continuous map and $\phi^* : \stsh_Y \to |\phi|_*\stsh_X$ is a map of sheaves of supercommutative algebras. Every such morphism $\phi$ induces a morphism of ordinary manifolds $\phi_{\tred}: X_{\tred} \to Y_{\tred}$.

\bigskip

\subsection{Construction of Supermanifolds by Gluing}~

As it is the case in classical geometry, sometimes it is useful to think of/construct geometric objects by gluing pieces of local models together. This approach is useful in the super setting as well and will allow us to construct more supermanifolds than simply those that are split. Locally supermanifolds can be thought of as a gluing of superdomains. We follow closely the ideas outlined in \cite{var}.

Indeed, let $\{ U_j = U_j^{m|n} \}$ be a collection of superdomains and $\{ W_j \}$ be a collection of open subdomains. Write $\stsh_j = \stsh_{U_j}$. Suppose we have then isomorphisms
$$ f_{ij}: (W_j , \stsh_j|_{W_j}) \to (W_i, \stsh_i|_{W_i}) $$
we can then construct a supermanifold $(X, \stsh)$ by setting the topological space to be
$$ X = \left ( \bigsqcup | U_j | \right ) / \sim $$
with the quotient topology given by the usual equivalence relation, namely $ p_j \sim q_i$ if $q_i = |f_{ij}|(p_j)$. The sheaf $\stsh$ is given as follows: for $V \subset X$ an open subset, we let $\widetilde{V}$ denote the pre-image of $V$ under the equivalence class projection, then 
\begin{equation}
    \begin{split}
        \stsh(V) & = \{ (s_j) \,\, | \,\, \text{for all}\,\, i,j, \,\, s_j = f_{ij}^*s_i  \} \\
        & \subset \prod_{j} \stsh_j(\widetilde{V} \cap |U_j|).
    \end{split}
\end{equation}
It is a straightforward tedious exercise to verify this does infact give a supermanifold $(X,\stsh)$. We then say that $(X, \stsh)$, or more generally any supermanifold $Y$ that is isomorphic to such a constructed $X$, is glued together by the data $\{ U_j, W_j,  \{f_{ij} \} \}$. We remark that if $X$ is glued together by the data $\{ U_j, W_j,  \{f_{ij} \} \}$, then $X_{\tred}$ is glued together by $\{ U_j, W_j,  \{(f_{ij})_{\tred} \} \}$.

Conversely given a supermanifold $(X,\stsh)$ we may cover it with open sets $\{ V_j \}$ such that for each $j$ we have an isomorphism of locally ringed spaces $\varphi_j: (V_j , \stsh|_{V_j}) \to (U_j, \stsh_{\Cx^{m|n}}|_{U_j})$. Then trivially $X$ is glued together by the data $\{ U_j, U_i \cap U_j, \{ (\varphi_i \circ \varphi^{-1}_j)|_{U_i \cap U_j} \} \} $.

\begin{ex} \label{proj_def}
Projective Superspace $\mathbb{P}^{m|n}$
\end{ex}

We present 3 different approaches to the analog of projective space is supergeometry. The first approach is to construct superprojective space by way of gluing as above. For each $i = 0,1 \dots m$ let $U_i = \Cx^{m|n}$ and $W_i = U_i - \{0\}$, with (global) coordinates denoted by $x_{0/j}, \dots, \widehat{x}_{j/j}, \dots, x_{m/j} \, | \, \xi_{1/j} \dots \xi_{n/j} $. Using the notation above we define the gluing functions $f_{ij}$ via
$$ f_{ij}: (W_j, \stsh_j|_{U_{ij}}) \to (W_i, \stsh_i|_{U_{ij}}) $$

$$ 
|f_{ij}|: (x_{0/j}, \dots, \widehat{x}_{j/j} \dots, x_{m/j}) \to \frac{1}{x_{i/j}} (x_{0/j}, \dots, \widehat{x}_{j/j} \dots, x_{m/j})
$$
$$
f_{ij}^*(x_{k/i}) = \frac{1}{x_{i/j}}x_{k/j}
$$
$$
f_{ij}^*(\xi_{k/i}) = \frac{1}{x_{i/j}}\xi_{k/j}.
$$
The resulting object we call \emph{projective superspace} $\mathbb{P}^{m|n}$. Clearly we have that $|\mathbb{P}^{m|n}| = \mathbb{P}^m$ as one should expect.

For the second approach we discuss the algebro-geometric description of the \emph{projectivization} $\mathbb{P}(V)$ of any
super vector space $V = V_0 \oplus V_1$. We let $V^* = \underline{\text{Hom}}(V,\Cx)$ denote the internal dual of $V$, i.e. all linear maps $V \to \Cx$ which need not preserve parity. Let $S^{\bullet} = \underline{\text{Sym}}^{\bullet}(V^*)$ be the internal symmetric algebra on the dual and consider the set $\text{Proj}(S)$ of all homogeneous prime ideal of $S$ which do not contain the irrelevent ideal $S^+ = S^{\geq 1}$. The graded algebra $S$ is also naturally $\mathbb{Z}/2$-graded, $S = S_0 \oplus S_1$. The decomposition $V^* = V^*_0 \oplus V^*_1$ then gives the decomposition for $S$
$$
S = \underline{\text{Sym}}^{\bullet}(V^*) = \text{Sym}^{\bullet}(V^*_0) \otimes \bigwedge(V^*_1)
$$
where we use the exterior algebra notation on the right hand side in the usual sense. Then it is immediate that in fact as a set 
$$
\text{Proj}(S) = \text{Proj}(S_0) = \text{Proj}(\text{Sym}^{\bullet}(V_0^*)).
$$
Let $S_{\tred} = \text{Sym}^{\bullet}(V_0^*)$, then $S$ is naturally an $S_{\tred}$ module and so we can consider the corresponding sheaf $\widetilde{S}$ on $\text{Proj}(S_{\tred})$. We take this to be the structure sheaf $\oo = \oo_{\mathbb{P}(V)}$ on superprojective space $\text{Proj}(S)$,
$$
\mathbb{P}(V) = \text{Proj}(S) := (\text{Proj}(S_{\tred}), \widetilde{S}) = (\text{Proj}(S_{\tred}), \widetilde{\text{Sym}^{\bullet}(V^*_0) \otimes \bigwedge(V^*_1)})
$$

It is then immediate that $\mathbb{P}(V)$ is split as it is visibly modelled on $\mathbb{P}(V_0)$ equipped with the locally free sheaf associated to the module $\text{Sym}^{\bullet}(V^*_0) \otimes V_1^*$.

Lastly we discuss a more functorial approach to $\mathbb{P}(V)$, as a set as all rank $1 | 0$ linear subspaces of $V$. One then defines the topology and sheaf in the obvious manner. We can describe $\mathbb{P}(V)$ in terms of its functor of points. It is the space representing the functor which associates to any $\Cx$-superscheme $S$, the set of locally free quotients
\begin{equation}
    \oo_S \otimes V \longrightarrow \mathcal{Q} \longrightarrow 0
\end{equation}
or co-rank $1|0$ (i.e. rank $m-1|n$). On $\mathbb{P}(V)$ one has the tautological line bundle $\oo_{\mathbb{P}(V)}(-1)$ defined by $\ker \left ( \oo_{\mathbb{P}(V)} \otimes V \to \mathcal{Q} \right )$, where the morphism $ \oo_{\mathbb{P}(V)} \otimes V \to \mathcal{Q}$ arises from the identity map $\text{id}:\mathbb{P}(V) \to \mathbb{P}(V)$.

Let us identify the reduction $\mathbb{P}(V)_{\tred}$ in this light and verify indeed it is $\mathbb{P}(V_0)$. For any supermanifold $X$, its reduction satisfies the following universal property: namely $X_{\tred}$ is the unique ordinary manifold with embedding and $X_{\tred} \hookrightarrow X$, such that for any other classical manifold $S = S_{\tred}$ with a morphism $\alpha: S \to X$ there exists a unique map $\alpha': S \to X_{\tred}$ making the following diagram commute
$$
\begin{tikzcd}
S \arrow[rd, "\alpha'"', dotted] \arrow[r, "\alpha"] & X                   \\
                                                     & X_{\tred} \arrow[u]
\end{tikzcd}.
$$
See that for $V = V_0 \oplus V_1$ decomposed into its even and odd part, that $\mathbb{P}(V)_{\tred} = \mathbb{P}(V_0)$. Indeed, define an embedding $\mathbb{P}(V_0) \hookrightarrow \mathbb{P}(V)$ as follows: an $S$-point of $\mathbb{P}(V_0)$ is a surjection
$$
\oo_{S} \otimes V_0 \longrightarrow \mathcal{Q} \longrightarrow 0
$$
where $\mathcal{Q}$ is locally free of rank $m-1|0$. This maps to the $S$-point of $\mathbb{P}(V)$
$$
\oo_{S} \otimes V \longrightarrow \mathcal{Q} \oplus \left( \oo_S \otimes V_1 \right) \longrightarrow 0.
$$
It is then immediately seen that the universal property holds after noting that for an ordinary manifold $S$, a sequence
$$
\oo_{S} \otimes V \longrightarrow \mathcal{Q} \longrightarrow 0
$$
with $\mathcal{Q}$ locally free of rank $m-1|n$ canonically decomposes into two sequences of locally free sheaves of either purely even or odd rank. The odd part gives a surjection
$$
\oo_{S} \otimes V_1 \longrightarrow \mathcal{Q}_1 \longrightarrow 0
$$
of locally free sheaves of the same rank $0|n$ which must be an isomorphism, hence $\mathbb{P}(V)_{\tred} = \mathbb{P}(V_0)$.

The projective space $\mathbb{P}(V)$ is canonically projected, i.e. there is a morphism $\mathbb{P}(V) \to \mathbb{P}(V)_{\tred} = \mathbb{P}(V_0)$. This is constructed as follows: given an $S$-point of $\mathbb{P}(V)$, we utilize the surjective map $\oo_S \otimes V \to \oo_S \otimes V_0$ and form the pushout
$$
\begin{tikzcd}
\oo_S \otimes V \arrow[d] \arrow[r] & \mathcal{Q} \arrow[r] \arrow[d] & 0 \\
\oo_S \otimes V_0 \arrow[r]         & \mathcal{Q}' \arrow[r]          & 0
\end{tikzcd}.
$$
The bottom row is then the corresponding $S$-point of the reduction.

\bigskip

\subsection{Super Vector Bundles}~

Vector bundles in supergeometry are most easily understood as generalizations of their algebro-geometric counterparts in the classical situation. In algebraic geometry a vector bundle on a space carries the same data as a locally free sheaf and this idea is what we adopt here. 

Let $X$ be a supermanifold. A locally free sheaf $\F$ of rank $r |  s$ on $X$ is a sheaf of $\stsh_X$-modules that is locally isomorphic to the trivial bundle $\stsh_X^{r|s} = \stsh_X^r \oplus (\Pi \stsh_X)^s$. That is, there exists an open cover of $\{ U_i \}$ of $X$ and sheaf isomorphisms
\begin{equation} \label{vectbundle}
    \begin{split}
         \F|_{U_i} & \overset{\sim}{\rightarrow} \left(\stsh_X|_{U_i}\right)^{r | s} \\
         & = \left(\stsh_X|_{U_i}\right)^{\oplus r} \oplus \left(\Pi \stsh_X|_{U_i}\right)^{\oplus s},
    \end{split}
\end{equation}
these are called the (a set of) \emph{local trivializations} of $\F$. We will use this notion as a complete replacement for that of a super vector bundle. We remark that it is dangerous to call $\F$ an \emph{even (odd) vector bundle} if it is locally free of rank $r | 0  \,\, (\text{resp.} \,\, 0|r )$ for some $r$. Indeed if $\stsh_X$ has nontrivial odd part then a locally free sheaf of rank $r|0$ \emph{does not} have the property that all of its sections are even. By an \emph{invertible sheaf} (or line bundle) we mean a sheaf locally free of rank $1 |0 $ or $0 | 1$.

Given a vector bundle $\F$ of rank $r | s$ with trivializations $\{ \psi_{i} \}$ for an open cover $\{ U_i \}$ establishing the isomorphisms above in (\ref{vectbundle}), the compositions $g_{ij} := \psi_{i} \circ \psi_{j}^{-1}$ give $\stsh_X$-linear automorphisms of $(\stsh_X|_{U_i \cap U_j})^{r|s}$. These are called (a set of) \emph{transition functions} for the vector bundle $\F$. Just as in the usual case the vector bundle $\F$ is equivalent to the data of its transition functions. This fact will prove useful in computations. Specializing to the case that $\F$ is an invertible sheaf of rank $1 | 0$ gives that the $g_{ij}$ form a 1-\v{C}ech cocycle of the sheaf $\stsh_{X,0}^*$ of even invertible superfunctions and we have the familiar result Pic($X) \cong H^1(X, \stsh_{X,0}^*)$, where $\text{Pic}(X)$ is the group of isomorphism classes of all rank $1 | 0$ line bundles on $X$.

Letting $\J \F$ denote the sub $\stsh_X$-module generated by $\J$, we always have the exact sequence
$$ 0 \rightarrow \J \F \rightarrow \F \rightarrow \F_{\tred} \rightarrow 0 $$
defining the quotient $\F_{\tred}$ as a \emph{super} vector bundle on $X_{\tred}$. In terms of transition functions, if the cocycle $g_{ij}$ describes $\F$ then reducing modulo $\J$, the cocycle $g_{ij} \,\, \text{mod} \,\, \J$ computes the transition functions of the vector bundle $\F_{\tred}$.

\bigskip

\subsection{The Tangent and Cotangent Sheaves}~

\smallskip

\subsubsection{The Tangent Sheaf}~

A super $\Cx$-derivation of $\stsh_X$ is a map $D : \stsh_X \to \stsh_X$ of sheaves of super $\Cx$-algebras such that (for homogeneous $D$)
$$ D(ab) = D(a)b + (-1)^{|a||D|}aD(b). $$
The tangent sheaf of $X$ is the sheaf of $\stsh_X$-modules $\text{Der}(\stsh_X)$ of all super $\Cx$-derivations. We denote it by $\mathcal{T}_X$ and call sections of it vector fields. If $x_1, \dots, x_m | \xi_1, \dots, \xi_n$ are coordinates on $\Cx^{m|n}$ then we have the \emph{odd coordinate vector fields} determined by the conditions
$$
\frac{\partial}{\partial \xi_j} \in (\mathcal{T}_{\Cx^{m|n}})_1
$$
$$
\frac{\partial}{\partial \xi_j}(x_k) = 0, \,\,\,\,\, \frac{\partial}{\partial \xi_j}(\xi_k) = \delta_{ij}.
$$
We have the following result from \cite{var}.

\begin{prop}
The tangent sheaf $\mathcal{T}_{\Cx^{m|n}}$ is free of rank $m|n$ with homogeneous generators the even and odd coordinate vector fields $\frac{\partial}{\partial x_1}, \dots, \frac{\partial}{\partial x_m} \, | \, \frac{\partial}{\partial \xi_1}, \dots, \frac{\partial}{\partial \xi_n}.$
\end{prop}

This implies that the tangent sheaf of any $m  | n$ dimensional supermanifold is locally free of rank $m | n$. 

Analogous to the case in classical differential geometry, given any morphism $\psi: X \to Y$ of supermanifolds one can define the differential or pushforward $\psi_*$ of $\psi$ to be the the $\stsh_X$-module map
$$
\psi_*: \mathcal{T}_X \to \psi^*\mathcal{T}_Y
$$
$$
\psi_*(V) = V \circ \psi^*
$$
thinking of $V$ as a $\Cx$-derivation of $\stsh_X$.

Let us now compute a \v{C}ech-cocycle that corresponds to the tangent sheaf. For simplicity let us assume a supermanifold $X$ is glued together by the data of two superdomains $U_1^{m|n}$ and $U_2^{m|n}$ with coordinates $x_1, \dots, x_m \, | \, \xi_1, \dots, \xi_n$ and $y_1, \dots, y_m \, | \, \zeta_1, \dots, \zeta_n$ respectively and an isomorphism $\varphi: W_1^{m|n} \overset{\sim}{\rightarrow} W_2^{m|n}$ between two sub superdomains. Computing the pullbacks of the generators of the sheaf on $W_2^{m|n}$ (abbreviating and writing simply $f(x | \xi)$ for $f(x_1, \dots, x_m \, | \, \xi_1, \dots, \xi_n)$),
$$
\varphi^*y_k = f_k(x | \xi)
$$
$$
\varphi^*\zeta_j = \eta_j(x | \xi)
$$
we produce $m$ even functions $f_k$ and $n$ odd functions $\eta_j$ which are commonly referred to as the \emph{coordinate transformations}. One can then form the \emph{Jacobian} matrix of this coordinate transformation

$$ 
\frac{\partial(y,\zeta)}{\partial(x,\xi)} =
\begin{pmatrix}
    \displaystyle{\frac{\partial f_i(x | \xi)}{\partial x_j}}      & \displaystyle{\frac{\partial f_i(x | \xi)}{\partial \xi_j}}   \\
  \displaystyle{\frac{\partial \eta_i(x | \xi)}{\partial x_j}}       & \displaystyle{\frac{\partial \eta_i(x | \xi)}{\partial \xi_j}}  \\
\end{pmatrix}.
$$

Now on $W_1$ the tangent sheaf, by the above proposition, is freely generated by $\frac{\partial}{\partial x_k}$ and $\frac{\partial}{\partial \xi_j}$ and similarly for $W_2$. 

As $\varphi: W_1 \to W_2$ is an isomorphism note two facts; the pushforward map $\varphi_*:\mathcal{T}_{W_1} \to \varphi^*\mathcal{T}_{W_2}$ is in fact an isomorphism of $\stsh_{W_1}$-modules and furthermore the pullback sheaf $\varphi^*\mathcal{T}_{W_2}$ has
$$
\frac{\partial}{\partial y_1}, \dots, \frac{\partial}{\partial y_n} \, \bigg | \, \frac{\partial}{\partial \zeta_1}, \dots, \frac{\partial}{\partial \zeta_n}
$$
as an $\stsh_X$-module basis.

By definition an arbitrary vector field $V$ on $X$ is given by vector fields $V_1, V_2$ with each $V_i$ a vector field on (a subset of) of $W_i$ that is compatible with the gluing isomorphism $\varphi$, in the sense that $V_2 = \varphi_*V_1$. We say that in this case $V_2$ is \emph{$V_1$ expressed in the coordinates $y  |  \zeta$.} 

In fact the map $\varphi_*$ expressed as a matrix $g$ with respect to the coordinate vector bases is exactly a cocycle that defines $\mathcal{T}_X$. One can easily then compute
\begin{equation}
    \begin{split}
        \varphi_*\frac{\partial}{\partial x_j} & = \sum_{k=1}^m \frac{\partial}{\partial x_j} ( \varphi^*y_k)\frac{\partial}{\partial y_j} + \sum_{k=1}^n \frac{\partial}{\partial x_j} ( \varphi^*\zeta_k)\frac{\partial}{\partial \zeta_k} \\
        & = \sum_{k=1}^m \frac{\partial f_k(x | \xi)}{\partial x_j} \frac{\partial}{\partial y_j} + \sum_{k=1}^n \frac{\partial \eta_k(x | \xi)}{\partial x_j} \frac{\partial}{\partial \zeta_j}
    \end{split}   
\end{equation}
and similarly
\begin{equation}
    \begin{split}
         \varphi_*\frac{\partial}{\partial \xi_j} & = \sum_{k=1}^m \frac{\partial f_k(x | \xi)}{\partial \xi_j} \frac{\partial}{\partial y_j} + \sum_{k=1}^n \frac{\partial \eta_k(x | \xi)}{\partial \xi_j} \frac{\partial}{\partial \zeta_j}
    \end{split}   
\end{equation}
which shows that the matrix of $\varphi^*$ is precisely the Jacobian $\frac{\partial(y,\zeta)}{\partial(x,\xi)}$. The result persists when one deals with an arbitrary supermanifold not necessarily built from two open domains.

On a supermanifold one would hope that the sheaf $\mathcal{T}_{X_{\tred}}$ of usual vector fields on the underlying manifold have some relationship with the sheaf of super vector fields $\mathcal{T}_{X}$. Indeed we have the following.

\begin{prop} \label{tangentsheaf}
For any supermanifold $X$ (real or complex) we have a canonical isomorphism of $\stsh_{\tred}$-modules
$$
(\mathcal{T}_X)_{\tred} = (\mathcal{T}_X)_{\tred,0} \oplus (\mathcal{T}_X)_{\tred, 1} \cong \mathcal{T}_{X_{\tred}} \oplus (\J/\J^2)^*
$$
where $(\J/\J^2)^* = \mathcal{H}om_{\stsh_{\tred}}(\J/\J^2, \stsh_{\tred}).$
\end{prop}
\begin{proof}
Let $k = \mathbb{C}$ or $\mathbb{R}$. We will construct two maps (of sheaves of $k$-vector spaces) $(\mathcal{T}_X)_0 \to \mathcal{T}_{X_{\tred}}$ and $(\mathcal{T}_X)_1 \to (\J/\J^2)^*$ to then obtain a map on the direct sum. The key idea is that given any derivation $V$ of $\stsh_X$, $V$ at worst "reduces degree", in the sense that 
$$
V(\J^k) \subset \J^{k-1}.
$$
for any $k$. This is a trivial consequence of the derivation property. Any vector field $V$ therefore, thought of as a derivation on $\stsh$ induces a map
$$
\bar{\bar{V}}: \J/\J^2 \to \stsh_X/\J
$$
$$
\bar{\bar{V}}(f + \J^2) = V(f) + \J.
$$
However for even vector fields $V = V_0$ we have that $V_0$ stabilizes $\J$ and hence $\bar{\bar{V_0}} = 0$. Lastly it is immediate that $\bar{\bar{V}}$ respects the $\stsh_{\tred}$ module structure on $\J/\J^2$ and thus we obtain a map
$$
(\mathcal{T}_X)_1 \to (\J/\J^2)^*.
$$
On the other hand, as every even vector field $V$ stabilizes $\J$, we get an induced derivation $\bar{V}: \stsh/\J \to \stsh/\J$ defining a map
$$
(\mathcal{T}_X)_0 \to \mathcal{T}_{X_{\tred}}.
$$
In total we get a well-defined map of sheaves of $k$-vector spaces
$$
\mathcal{T}_X \to \mathcal{T}_{X_{\tred}} \oplus (\J/\J^2)^*
$$
$$
V = V_0 + V_1 \to \bar{V_0} + \bar{\bar{V_1}}.
$$
The kernel of the above map can be identified as those derivations $V$ whose image lies completely in $\J$, and hence contains the submodule $\J\mathcal{T}_X$. This gives our desired canonical map of $\stsh_{\tred}$-modules
$$
(\mathcal{T}_X)_{\tred} \to \mathcal{T}_{X_{\tred}} \oplus (\J/\J^2)^*.
$$
It is trivial now to show that in local coordinates this map is an isomorphism.

\end{proof}

\subsubsection{The Lie Bracket}~

On a supermanifold one also has the \emph{super Lie bracket} $[V,W]$ of vector fields defined by (for homogeneous $V,W$)
$$
[V,W] := VW - (-1)^{|V||W|}WV.
$$
One can easily check that this defines a derivation of parity $|V||W|$. Furthermore we have the useful identity
$$
[fV,W] = f[V,W] - (-1)^{(|f|+|V|)|W|}W(f)V,
$$
for homogeneous $f \in \stsh$.

\bigskip

\subsubsection{The Cotangent Sheaf}~

The dual of the tangent sheaf is denoted $\Omega_X^1 := \mathcal{H}om(\mathcal{T}_X, \stsh)$ is called the \emph{cotangent sheaf}. This is also locally free of the same dimension with the (local) basis given by the coordinate one-forms $dx_1, \dots, dx_m \, | \, d\xi_1, \dots, d\xi_n$. For any $k$ let $\Omega^k_X = \bigwedge^k \Omega^1_X$ be the $k$th exterior power of the cotangent bundle and let $\bigwedge^{\bullet} \Omega^1_X = \oplus_k \Omega_X^k$. Sections of $\bigwedge^{\bullet} \Omega^1_X$ are called (super) differential forms, while homogeneous sections of degree $k$ are called (super) differential $k$-forms.

A similar calculation to the above will show that the transition functions corresponding to the cotangent bundle is exactly the inverses of the Jacobians of the coordinate transformations. This also is immediate when one recalls the transition functions of the dual of a vector bundle $\mathcal{F}^*$ are the inverses of the original those corresponding to $\mathcal{F}$.

Following the convention in \cite{qfst} we have the canonical pairing
$$
\langle \cdot , \cdot \rangle : \mathcal{T}_X \otimes \Omega^1_X \to \stsh_X,
$$
which implies the sign rule $\langle aV, b\omega \rangle = (-1)^{|V||b|} ab \langle V , \omega \rangle$. This pairing is used to define the super exterior derivative (or simply the exterior derivative) $d: \stsh_X \to \Omega^1_X$ from the equation
$$
\langle V , df \rangle = V(f).
$$
As in the classical case, $d$ extends uniquely to as square zero derivation of the sheaf of differential forms $\bigwedge^{\bullet} \Omega_X^1$,
\begin{enumerate}
    \item $d^2 = 0$
    \item $d(\alpha \wedge \beta) = d\alpha \wedge \beta + (-1)^{k} \alpha \wedge d\beta, \,\,\,\,\,\,\,$ for $\alpha \in \Omega^k_X,$
\end{enumerate}
and so for example DeRham Cohomology makes sense on a supermanifold.

Differential forms however play a very different role in supergeometry than that in classical geometry. In the super setting the exterior algebra of a superalgebra (or sheaf of superalgebras) is the quotient of the tensor algebra by the relations $x \otimes y \sim - y \otimes x$ if one of $x$ or $y$ are even and $\xi \otimes \zeta \sim \zeta \otimes \xi$ if both $\xi$ and $\zeta$ are odd. Hence in particular, if $X$ is a supermanifold of dimension $m  |  n$, $n > 0$ and $\xi$ a local odd coordiante, then the $k$-fold wedge product $(d\xi)^k = d\xi \wedge \dots \wedge d\xi$ is non-zero for each $k$. Thus, \emph{there is no top exterior power of} $\Omega^1_X.$ This shows that the identification of objects analogues to volume forms from classical geometry, i.e. objects which it makes sense to integrate over a manifold, is not the naive guess of super differential forms one might make.

\bigskip

\subsection{The Berezinian Sheaf and the Berezin Integral}~

For vector bundles on ordinary manifolds, one had several common constructions available that were essentially formal consequences of operations with vector spaces such as direct sum, tensor product, determinant etc. In supergeometry one has a new construction available known as the Berezinian. 

Suppose $\F$ is vector bundle of rank $r | s$ and let $\{ g_{ij} \}$ be a cocycle corresponding to $\F$. We then define the invertible sheaf $\ber \F$ to be the bundle corresponding to the transition functions $\{ \ber g_{ij} \}$. We enforce that $\ber \F$ is locally free of rank $1 | 0$ if $s$ is even and $0 | 1$ if $s$ is odd. If $e_1, \dots, e_r \, | \, \theta_1, \dots, \theta_s$ is a collection of local generators trivializing $\F$ then the symbol $[e_1 \dots e_r \, | \, \theta_1 \dots \theta_s]$ denotes a trivializing section of $\ber \F$. Note that in the classical situation, this definition specializes to the definition of the determinant of a vector bundle.

Similar to the classical situation, any exact sequence of super vector bundles on a supermanifold $X$
$$
\cdots \rightarrow \F_{i-1} \rightarrow \F_i \rightarrow \F_{i+1} \rightarrow \cdots
$$
induces a canonical isomorphism
$$
\otimes_i \, \left ( \ber \F_i \right )^{(-1)^i} = \stsh_X.
$$

If $X$ is a supermanifold then we define $\ber X := \ber \Omega_X^1$ and simply call this \emph{the Berezinian of} $X$. We emphasize that for two local coordinate systems $x | \xi$ and $y | \zeta$ the gluing law reads
$$ [dy_1,\dots dy_m | d\zeta_1, \dots d\zeta_n ] = \ber \frac{\partial(y,\zeta)}{\partial (x,\xi)}[dx_1,\dots dx_m | d\xi_1, \dots d\xi_n ].
$$

\begin{ex}
The Berezinian of Projective Superspace $\mathbb{P}^{m|n}$.
\end{ex}

Let us compute the Berezinian for any projectivization $\mathbb{P}(V)$ for $V$ a super vector space of dimension $m+1 |n$. The argument hinges on the identification of the cotangent sheaf $\Omega^1_{\mathbb{P}(V)}$. We can describe it as follows: there is a canonical exterior derivative $d: \oo_{\mathbb{P}(V)} \to \Omega^1_{\mathbb{P}(V)}$ which extends uniquely as a first order differential operator
\begin{equation} \label{d_ext}
    d: \oo_{\mathbb{P}(V)} \otimes V \to \Omega^1_{\mathbb{P}(V)} \otimes V
\end{equation}
by the simple formula $d(f \otimes v) = df \otimes v$ extended by linearity. The fact that the vector bundles involved are twisted by a constant sheaf will give that this is well defined. 

On $\mathbb{P}(V)$ we have the natural sequence (as discussed in the functorial approach to $\mathbb{P}(V)$ in Example (\ref{proj_def}))
\begin{equation} \label{proj_nat_seq}
    0 \rightarrow \stsh_{{\mathbb{P}(V)}}(-1) \rightarrow \stsh_{\mathbb{P}(V)} \otimes V \rightarrow \mathcal{Q} \rightarrow 0.
\end{equation}
Now pre and post composing $d$ in (\ref{d_ext}) with the above short exact sequence gives us a morphism
$$
\stsh_{{\mathbb{P}(V)}}(-1) \longrightarrow \stsh_{\mathbb{P}(V)} \otimes V \overset{d}{\longrightarrow} \Omega^1_{\mathbb{P}(V)} \otimes V \longrightarrow \Omega^1_{\mathbb{P}(V)} \otimes \mathcal{Q}.
$$
The morphsim we obtain $\oo_{{\mathbb{P}(V)}}(-1) \to \Omega^1_{\mathbb{P}(V)} \otimes \mathcal{Q}$ can be seen to be a morphism of $\oo_{{\mathbb{P}(V)}}$-modules and furthermore, in local coordinates, to be an isomoprhism. The details of this computation can be found in Manin's book \cite{man1}.

Using this fact, dualizing and twisting by $\oo_{\mathbb{P}(V)}(1)$ we arrive from (\ref{proj_nat_seq})
$$
0 \longrightarrow \Omega^1_{\mathbb{P}(V)} \longrightarrow \stsh_{\mathbb{P}(V)}(-1) \otimes V^* \longrightarrow \stsh_{\mathbb{P}(V)} \longrightarrow 0.
$$
This gives
$$
\ber \mathbb{P}(V) = \ber \Omega^1_{\mathbb{P}(V)} = \ber \left ( \stsh_{\mathbb{P}(V)}(-1) \otimes V^* \right ) \cong \oo_{\mathbb{P}(V)}(n-m-1).
$$

\bigskip

\subsubsection{The Berezin Integral}~

We now move on to the question of integration on a supermanifold. We define first the \emph{Berezin integral} on the smooth split supermanifold $\mathbb{R}^{m|n}$. As noted above any function on $\mathbb{R}^{m |n}$ is written (uniquely as) a Grassmann polynomial in odd coordinates $\xi_1, \dots, \xi_n$ with coefficients smooth functions of ordinary coordinates $x_1, \dots, x_n$ in the following way,
$$
f(x_1, \dots, x_m | \xi_1, \dots, \xi_n) = \sum_{ I \subset \{ 1,\dots,n \} } f_I(x_1, \dots, x_m)\xi_I
$$
where $\xi_I = \xi_{i_1} \cdots \xi_{i_k}$ if $I = \{i_1 < \dots < i_k \}$. We then define the \emph{Berezin integral of} $f$ \emph{over} $\mathbb{R}^{m|n}$ to be
$$
\int_{\mathbb{R}^{m|n}} f(x|\xi) \, [dx_1 \cdots dx_m \, | \, d\xi_1 \cdots d\xi_n ] := \int_{\mathbb{R}^{m}} \frac{\partial^n}{\partial \xi_n \cdots \partial \xi_1} f(x|\xi) \, dx_1\dots dx_m,
$$
that is, we simply integrate in the usual way the highest nonzero coefficient of $f$ in its expansion in the $\xi$'s. Of course one deals with the convergence of the integral in the usual ways, we will not pause here to comment. This definition can be extended in the obvious way to integrals over arbitrary sub superdomains $U^{m|n}$ of $\mathbb{R}^{m|n}$. As a shorthand, it is common to write $[dx_1 \cdots dx_m \, | \, d\xi_1 \cdots d\xi_n ] = [dx \, | \, d\xi ].$

This notation above suggests that the Berezin integral is an integral of not a superfunction but rather a section of $\ber \mathbb{R}^{m|n}$. This viewpoint is the correct one as one has the following super analog of the change of variables theorem.
\begin{prop}
If $\varphi: \mathbb{R}^{m|n} \to \mathbb{R}^{m|n}$ is an isomorphism of smooth supermanifolds with $(x | \xi)$, and $(y | \zeta)$ denoting coordinates on the source and target we have
\end{prop}
\vspace{-.35cm}
$$
\int_{\mathbb{R}^{m|n}} f(y | \zeta) \, [dy \, | \, d\zeta] = \int_{\mathbb{R}^{m|n}} \varphi^*f(y | \zeta) \, \text{\ber} \, \frac{\partial (y, \zeta)}{\partial (x, \xi)} [dx \, | \, d\xi].
$$

The proof of the above proposition is non-trivial and can be found in \cite{var}. One can think of the Berezinian sheaf as exactly those objects which gives a coordinate independent definition of the Berezin integral.

Now if $X$ is an oriented \emph{smooth} supermanifold, we define the integral of a global section $\sigma \in \Gamma(X, \ber X)$ in the usual way, by finding a partition of unity and reducing the the definition given above.

Let $(X,\stsh)$ be now a complex supermanifold of dimension $m | n$. In this context we will usually write $\ber X = \omega$. On $X$ one can construct the sheaf of \emph{smooth superfunctions}, denoted by $\ee$. Loosely $\mathcal{E}$ is defined by the condition that for each local trivialization of $\stsh$ as $\stsh_{\Cx^m} \otimes \bigwedge (\zeta_1, \dots, \zeta_n)$ we take $\ee$ to be trivialized on the same neighborhood as $\mathcal{C}^{\infty}_{\R^{2m}} \otimes \bigwedge (\alpha_1, \dots, \alpha_n, \beta_1, \dots, \beta_n) \otimes \Cx $ and then glued together via the transition functions for $\stsh$. Intuitively one thinks that the relationships
$$
x_k = \text{Re} \, z_k, \hspace{.5cm} y_k = \text{Im} \, z_k, \hspace{.5cm} \alpha_k = \text{Re} \, \zeta_k, \hspace{.5cm} \beta_k = \text{Im} \, \zeta_k
$$
are enforced in this construction. This has been made precise in a short paper of Haske and Wells \cite{smoothfromcomplex}, however we do not discuss this further. We will think of $z_k, \bar{z}_k, \zeta_k, \bar{\zeta}_k$ as generators of $\ee$ analogous as to what is common in complex analysis, so that locally every smooth superfunction $f$ is of the form (abbreviating the indices and using the usual multi-index notation)
$$
f(z, \bar{z} | \zeta, \bar{\zeta}) =\sum_{I,J} f_{IJ}(z, \bz)\zeta_I \bar{\zeta}_J
$$
for ordinary ($\Cx$-valued) smooth functions $f_{IJ}$. The sheaf $\ee$ then naturally has complex conjugation. 

Once the sheaf $\ee$ of smooth superfunctions on a complex supermanifold is established one has the notion of a smooth section of a complex super vector bundle $\mathcal{F}$, namely sections of the tensor product $\mathcal{F} \otimes_{\stsh} \mathcal{E}$ (note that $\mathcal{E}$ is an $\stsh$-module). Furthermore we denote by $\overline{\mathcal{F}}$, the complex conjugate vector bundle of $\mathcal{F}$. We often write for the sheaf of smooth sections of $\mathcal{F}$ as $\mathcal{F}_{\ee}$. Of particular interest is the \emph{smooth} Berezinian sheaf $\omega \otimes \overline{\w} \otimes \ee =: |\omega|^2$, as its sections yield natural objects that can be integrated over the \emph{entire} complex supermanifold $X$. Thus when comparing to the classical setting, sections of $\omega$ correspond to "holomorphic top-forms" or forms of type $(n,0)$ on a complex $n$-manifold, while sections of $|\omega|^2$ correspond to genuine top forms or forms of type $(n,n)$.

Here one can also consider super Dolbeault cohomology. Let $\Omega_X^{(p,q)}$ denote the sheaf of differential forms of type $(p,q)$, namely those forms $\sigma$ that can be written in local coordinates 

$$
\sigma = \sum_{I = J \cup K} f_{I}(X) dX_J \wedge d\overline{X}_K
$$
with $f_I$ smooth and where we use above the usual multi-index notation and let $X$ denote any of the local coordinantes $(x_1, \dots, x_m | \xi_1, \dots, \xi_n)$. One then has the usual Cauchy-Riemann operator $\bar{\partial}$ satisfying the usual Leibniz rule and squaring to zero when viewed as an odd derivation of the algebra $\Omega_X^{(\bullet,q)}$. More generally for each holomorphic super vector bundle $\mathcal{F}$ (locally free $\stsh_X$-module) we let $\Omega_X^{(p,q)}(\mathcal{F})$ denote the sheaf of smooth $(p,q)$ forms with values in $\mathcal{F}$ and then we obtain a super Dolbeault resolution 
 $$
0 \to \mathcal{F} \overset{\bar{\partial}}{\rightarrow} \Omega_X^{(0,1)}(\mathcal{F}) \overset{\bar{\partial}}{\rightarrow} \Omega_X^{(0,2)}(\mathcal{F}) \overset{\bar{\partial}}{\rightarrow} \dots
$$
which is an acyclic resolution just as in the classical case (as each $\Omega_X^{(p,q)}(\mathcal{F})$ is fine). Hence the cohomology computes the sheaf cohomology $H^q(X, \mathcal{F})$.

\subsection{Super GAGA}

We pause briefly to discuss the GAGA principle in the super setting. When convenient we often pass freely between the algebraic and holomorphic categories. This is done in an effort to increase clarity and sometimes be more inline with the related literature. 

As the goal of this work is largely applied: it is aimed at providing formulas for objects related to measures arising from string theory and not too much in developing foundations, it is natural to use whatever language or techniques are convenient and available. 

Super GAGA has been studied in \cite{super_gaga} and is mentioned in various places in the lectures of \cite{vids}. For example, in \cite{super_gaga} it is shown that for algebraic sub supervarieties $X \subset \mathbb{P}^{m|n}$ there is a natural analytification functor $X \mapsto X^h$, as in classical GAGA, and for which the natural maps
$$
H^q(X, \mathcal{F}) \to H^q(X^h, \mathcal{F}^h)
$$
are isomorphisms for $\mathcal{F}$ coherent.

For the majority of our discussions, we work in the holomorphic category.
\chapter{Super Riemann Surfaces and Other Preliminaries} \label{susy_prelim}
\label{SUSYprelim}

\section{Basic Notions}\mbox{}

\subsection{Definitions and Elementary Structure Theory}

We briefly review some basic definitions and notions to setup notation. \emph{Super Riemann surfaces} are a certain class of complex supermanifolds of dimension $1 | 1$, which carry an additional piece of structure. These play the role of superstring worldsheets and their theory very closely parallels that of classical Riemann surfaces.

We are interested in the moduli of these objects and thus have the following notion of a family.

\begin{defx}
A family of super Riemann surfaces is a family of complex supermanifolds $\pi: X \to S$ of relative dimension $1 | 1$ equipped with a maximally non-integrable distribution $\mathcal{D}$ of rank $0 | 1$, i.e. an odd subbundle of the relative tangent bundle $\mathcal{T}_{X/S}$ such that the Lie bracket induces the isomorphism
$$ [\cdot,\cdot] : \mathcal{D}^2 \xrightarrow{\sim} \mathcal{T}_{X/S}/\mathcal{D}.
$$
\end{defx}

The main complication in the study of families of super Riemann surfaces is the presence of odd moduli. Essentially these should be thought of "odd parameters" on which the super Riemann surface depends on. In fact, if only local properties of the moduli space is of interest, it was shown in \cite{locmod} that it suffices to study the slightly more general objects consisting of super Riemann surfaces with an enlarged structure sheaf. Specifically we now require that the structure sheaf is locally modeled on
$$ \stsh_{\mathbb{C}}[\xi] \otimes \Lambda(\tau_1, \dots, \tau_L) $$
for some $L$. Here the $\tau$'s are the additional odd parameters or "odd moduli". One could also add additional even moduli but it turns out not to change the analysis. It is custom (and admittedly somewhat confusing) to suppress this from explanation and work in the situation of a single super Riemann surface with ``odd parameters which it depends on".

Let us analyze some local structure.

\begin{lem}
Locally we can find relative coordinates $x|\xi$ such that the distribution $\mathcal{D}$ is generated by the odd vector field
$$ 
D_\xi = \frac{\partial}{\partial \xi} + \xi \frac{\partial}{\partial x}.
$$
Such coordinates are called superconformal.
\end{lem}
\begin{proof}
The distribution $\mathcal{D}$ is locally free of rank $0 | 1$ hence, around a point $p \in X$ we can trivialize $\dis$ in a coordinate chart $x | \xi$  so that it is generated by the single odd vector field
$$ V = a(x|\xi) \parx + b(x|\xi) \parxi $$
for $a$ odd and $b$ even. Expanding in powers of $\xi$ we write $a(x|\xi) = a_0(x) + a_1(x)\xi$, $b(x|\xi) = b_0(x) + b_1(x)\xi$ with $a_1, b_0$ even and $a_0, b_1$ odd. Since $V$ generates $\dis$, we can map it to a generator of $\dis_{\tred}$ , $V|_{X_{\tred}} = V $ mod $\J = b_0(x) \partial_{\xi}$, and thus we see that $b_0$ is non-zero in the local chart and we conclude that $b$ is invertible. A quick computation will show that (here primes denotes the $x$ derivative)
\begin{equation} \nonumber
\begin{split}
\frac{1}{2}[V,V] & = V^2 \\
& = \left (a(x|\xi)\parx + b(x|\xi) \parxi \right )\left (a(x|\xi)\parx + b(x|\xi) \parxi \right ) \\
& = \Big( a(x|\xi)a'(x|\xi) +a_1(x)b(x|\xi) \Big )\parx + \Big ( a(x|\xi)b'(x|\xi) + b(x|\xi)b_1(x|\xi) \Big ) \parxi. 
\end{split}
\end{equation}
The condition that $\dis^2 \cong \mathcal{T}_X/\dis$ via the Lie bracket implies that $V^2$ must generate this quotient. This quotient is locally free of rank $1 | 0$ and same is true for its reduction. Hence, as above, we conclude that $V^2_{\tred} = a_1(x)b_0(x) \partial_x$ mod $\J \dis^2$ cannot vanish which gives $a_1 \neq 0$. As $b$ is invertible we can define $f(x|\xi) = f_0(x) + f_1(x)\xi = b^{-1}a$ and assume $V$ is the generator $f(x|\xi)\partial_x + \partial_\xi $. Noting that $f_1(x)$ does not vanish, possibly shrinking the coordinate neighborhood, we can find a local holomorphic square root $h(x)$, and then the coordinate transformation
$$ x' = x, \,\,\, \xi' = \frac{1}{h(x)} f(x | \xi) $$
transforms $V$ to
\begin{equation}
\begin{split}
V & = f(x|\xi)\parx + \parxi \\
& = h(x)\xi' \left( \frac{\partial}{\partial x'} + \frac{\partial \xi'}{\partial x}\frac{\partial}{\partial \xi'} \right )+ h(x) \frac{\partial}{\partial \xi'} \\
& = h(x)\xi' \frac{\partial}{\partial x'} + \left ( h(x) \xi' \frac{\partial \xi'}{\partial x} + h(x) \right ) \frac{\partial}{\partial \xi'} \\
& = \left ( h(x) \xi' \frac{\partial \xi'}{\partial x} + h(x) \right ) D_{\xi'}
\end{split}
\end{equation}
showing that $\dis $ is generated by $D_{\xi'}$ in these new coordinates.

\end{proof}

We say a change of coordinates $y | \zeta$ is \emph{superconformal} if $\mathcal{D}_{\zeta}$ and $\mathcal{D}_{\theta}$ are $\stsh_X$-multiples of each other. 

Throughout this paper we will sometimes refer to a family of super Riemann surfaces as a \emph{family of SUSY curves} or simply by a \emph{SUSY family}. It is well known \cite{modSRS} that if the base $S$ is reduced, we essentially get a classical object, namely a family of spin curves.

\begin{prop} \label{jspin}
Let $\pi: X \to S$ be a family of super Riemann surfaces over a reduced base $S$. Let $\J \subset \stsh_X$ denote the sheaf of ideals generated by all odd elements. Then
\begin{enumerate}
    \item $\mathcal{J}$ is a locally free $\stsh_{\tred}$ module of rank $0 | 1$,
    \item $\J^* = \mathcal{H}om_{\stsh_{\tred}}(\J, \stsh_{\tred}) \cong \mathcal{D}_{\tred}$
    \item $\Pi \mathcal{J}$ becomes a relative spin structure on the family $X_{\tred} \to S$, i.e.
$$ (\Pi \mathcal{J}^{\otimes 2}) = \J^{\otimes 2} \cong \Omega_{X_{\tred}/S}^1, $$
\end{enumerate}
where $\Pi$ is the parity reversing functor.
\end{prop}
\begin{proof}
Locally the structure sheaf $\stsh_X$ is isomorphic to a restriction of $\stsh_{S} \otimes \stsh_{\mathbb{C}^{1|1}}$. Thus the decomposition (3) above implies that $\stsh_{\tred}$ is locally a restriction of $\stsh_{S \times \mathbb{C}}$ and $\J$ is locally a restriction of $\Pi \stsh_{S \times \mathbb{C}}$. This gives the first assertion.

Let $\J^* = \mathcal{H}om_{\stsh_{\tred}}(\J, \stsh_{\tred})$ denote the dual line bundle of $\J$ over $X_{\tred}$ and $\mathcal{T}_X$ and $\mathcal{T}_{X_{\tred}}$ denote the relative tangent sheaves over $X$ and $X_{\tred}$ respectively. By proposition 3.2 above we have the canonical decomposition
$$ (\mathcal{T}_X)_{\tred} = (\mathcal{T}_X)_{\tred,0} \oplus (\mathcal{T}_X)_{\tred, 1} = \mathcal{T}_{X_{\tred}} \oplus \J^* $$
of $\stsh_{\tred}$ modules.

On a super Riemann surface we have the exact sequence
$$ 0 \longrightarrow \mathcal{D} \longrightarrow \mathcal{T}_X \longrightarrow \mathcal{D}^{\otimes 2} \longrightarrow 0 $$
which can be reduced and yields an exact sequence of line bundles on $X_{\tred}$.
$$ 0 \longrightarrow \mathcal{D}_{\tred} \longrightarrow (\mathcal{T}_X)_{\tred} \longrightarrow \mathcal{D}^{\otimes 2}_{\tred} \longrightarrow 0. $$
Decomposing the above exact sequence according to the $\mathbb{Z}_2$ grading in fact gives another identification of the even and odd sub $\stsh_{\tred}$ modules of $(\mathcal{T}_X)_{\tred}$, i.e. $ \dis_{\tred} \cong \J^*$ and $\dis^{\otimes 2}_{\tred} \cong \mathcal{T}_{X_{\tred}}$. These two facts yield the desired result, $(\Pi \J)^{\otimes 2} \cong \Omega^1_{X_{\tred}/S}$.

\end{proof}

In other words, we have that the reduction of supermoduli space $\mathfrak{M}_g$ of super Riemann surfaces of genus $g$ is the moduli space $\mathcal{SM}_g$ of genus $g$ Riemann surfaces equipped with a spin structure. 

We pause first to introduce some notation to make the following Corollary more precise. We denote by $\underline{\textbf{SSch}}, \underline{\textbf{Sch}}^{\tred}$ and $\underline{\textbf{Set}}$ the categories of superschemes, reduced schemes and sets respectively. One can easily seen that We have a natural adjunction between the functors $\tred: \underline{\textbf{SSch}} \to \underline{\textbf{Sch}}^{\tred}$ and $\beta: \underline{\textbf{Sch}}^{\tred} \to \underline{\textbf{SSch}}$ where $\beta$ simply views an ordinary reduced scheme as a superscheme in a trivial way. We denote by
$$
\epsilon: \beta \circ \tred \to 1_{\underline{\textbf{SSch}}}
$$
the counit of this adjunction.

\begin{cor}
The reduction of supermoduli space $\mathfrak{M}_g$ of super Riemann surfaces of genus $g$ is the moduli space $\mathcal{SM}_g$ of genus $g$ Riemann surfaces with a spin structure.
\end{cor}
\begin{proof}
The supermoduli space $\mathfrak{M}_g$ we view as the geometric object (precisely a superstack) which represents the functor (which we also denote by $\mathfrak{M}_g$)
$$
\mathfrak{M}_g: \underline{\textbf{SSch}} \longrightarrow \underline{\textbf{Set}}
$$
$$
S \longmapsto \{ \text{isomorphism classes of families of SUSY curves } X \to S \}.
$$
The reduction or reduced space $(\mathfrak{M}_g)_{\tred}$ is by definition the geometric object which represents the functor $\mathfrak{M}_g \circ \epsilon$, i.e. isomorphism classes of SUSY curves over a reduced base. By Proposition \ref{jspin} this is exactly $\mathcal{SM}_g$.

\end{proof}

It turns out that the moduli space $\mathfrak{M}_g$ is not a supermanifold but rather a much more general object known as a super algebraic stack. The specifics will not concern us as it does not do us too much harm in working with $\mathfrak{M}_g$ thinking intuitively that it is a supermanifold. We can remark however that the stackyness of supermoduli space is forced on us immediately as every SUSY curve has a canonical automorphism given simply by $f \mapsto (-1)^{|f|}f$.

We will also use the notion of a \emph{family of supercurves}, by which we mean simply a family $\pi: X \to S$ of complex supermanifolds of relative dimension $1|1$. Then a family of SUSY curves is a family of supercurves with the extra data of the odd distribution $\mathcal{D}$.

\subsection{The Berezinian of SUSY Curves}~

Let $\pi: X \to S$ denote a family of SUSY curves. Of fundamental importance to the theory is the exact seqeunce
\begin{equation} \label{fund_SES}
    0 \longrightarrow \mathcal{D} \longrightarrow \mathcal{T}_{X/S} \longrightarrow \mathcal{D}^{\otimes 2} \longrightarrow 0.
\end{equation}
The map $\mathcal{T}_{X/S} \to \mathcal{D}^{\otimes 2}$ is the composition $$\mathcal{T}_{X/S} \to \mathcal{T}_{X/S} / \mathcal{D} \overset{[\cdot,\cdot]^{-1}}{\rightarrow} \mathcal{D}^{\otimes 2}$$ which in local relative superconformal coordinates $x | \xi$ is the map determined by
$$
\parx \to D_{\xi} \otimes D_{\xi}, \hspace{.5cm} D_{\xi} \to 0.  
$$
Dualizing (\ref{fund_SES}) gives
\begin{equation} \label{dual_fund_SES}
    0 \longrightarrow \mathcal{D}^{\otimes(-2)} \longrightarrow \Omega^1_{X/S} \longrightarrow \mathcal{D}^{-1} \longrightarrow 0.
\end{equation}
The distinguished subsheaf $\mathcal{D}^{(\otimes -2)} = \mathcal{D}^{-2}$ corresponds to the dual of the quotient $\mathcal{T}_{X/S} / \mathcal{D}$, i.e. those relative one-forms that vanish identically on $\mathcal{D}$. Since $\mathcal{D}$ is locally generated by the odd vector field $\parxi + \xi \parx$, a quick calculation will show that $\mathcal{D}^{-2}$ is locally generated by $dx - \xi d\xi$.

Taking the Berezinians of (\ref{dual_fund_SES}) gives a canonical isomorphism
$$ \omega_{X/S} := \ber \Omega^1_{X/S} = \ber \dis^{-1} \otimes \ber \dis^{-2} $$

As $\dis^{-1}$ is of rank $0|1$ and $\dis^{-2}$ is of rank $1|0$, their Berezinians are canonically $\dis^{1}$ and $\dis^{-2}$ respectively, hence
$$ \omega_{X/S} \cong \dis^{-1}. $$

Thus, in local coordiantes, the relative Berezinian of $X$ over $S$ can be also be thought of as relative one-forms modulo $dx - \xi d\xi$. We will frequently denote relative Berezinian sheaf by $\w_{X/S}$, or simply by $\w$.

\subsubsection{Connection Between the Berezinian and One-Forms}~

In \cite{RSV} an interesting and useful connection was made between one-forms and sections of the Berezinian on a super Riemann surface. Combining the map $\Omega^1_{X/S} \to \mathcal{D}^{-1}$ of (\ref{dual_fund_SES}) and the isomorphism $\mathcal{D}^{-1} \cong \w$, we get a natural map taking holomorphic one-forms to sections of the Berezinian $\Omega^1_{X/S} \to \w$. In local coordinates $z | \tth$ this is
$$
f(z|\tth) dz + g(z|\tth) d\tth \mapsto ( g(z|\tth) + f(z|\tth) \tth) [dz \, | \, d\tth].
$$

This map cannot be an isomorphism as $\Omega^1_X$ has rank $1|1$ while $\w$ is of rank $0|1$, however in \cite{RSV} it was noticed that upon restriction to $d$-\emph{closed} one-forms, we do get an isomorphism (here $d$ is the usual exterior derivative). The inverse map we denote by $\alpha: \w \to Z^1_X := \{ \text{closed holomorphic one-forms} \}$. It is given in coordinates as, for $\sigma = f(z|\tth) [dz \, | \, d\tth]$,
\begin{equation} \label{alpha}
\alpha(\sigma) : = d\tth f(z|\tth) + \varpi D_{\tth}f(z|\tth),
\end{equation}
where $\varpi := dz - \tth d\tth$ is the local generator of $\mathcal{D}^{-2}$ and $D_{\tth}$ is the usual local generator of the distribution. Note that above we have followed the convention in \cite{RSV} and have written the coefficient functions \emph{to the right} of the forms $d\tth$ and $\varpi$. 

One can check that the local coordinate definition (\ref{alpha}) is well-defined and gives a genuine map $\alpha: \w \to Z^1_X$. A coordinate invariant description of $\alpha$ is described in \cite{wit1}, it is related to the notion of \emph{picture number} and \emph{picture changing operators} in string theory. We will not need these notions here and so omit further discussion.

The natural map $\Omega_{X/S}^1 \to \w$ corresponds to what one might consider the ``\emph{super exterior derivative}'' $d : \oo \to \w$, very analogous to the classical situation. Locally this maps $f = f(z|\tth)$ to 
$$
f \mapsto df = D_{\tth}(f) [dz \, | \, d\tth].
$$

\bigskip

\bigskip

\subsection{The Sheaf of Superconformal Vector Fields}~

On a super Riemann surface one also has the \emph{sheaf of superconformal vector fields} $\mathcal{W}$. These are vector fields that preserve the supersymmetry in the sense that $[\mathcal{W},\mathcal{D}] \subset \mathcal{D}$. We remark that the sheaf $\mathcal{W}$ is {\bf not} an $\stsh_X$-module but only a sheaf of $\mathbb{C}$ vector spaces. Nevertheless $\mathcal{W}$'s utility will be in the fact that it generates automorphisms of the the super Riemann surface and thus will help us identify tangent spaces to the Moduli spaces of interest.

Locally in superconformal coordinates $x | \xi$, a vector field $V$
$$
V = f(x, \xi) \parx + g(x, \xi) D_{\xi} 
$$
is in $\mathcal{W}$ if and only if $[V,D_{\xi}] = 0$ mod $\mathcal{D}$. A quick computation will give
\begin{equation}
\begin{split}
[V,D_{\xi}] & \equiv ( 2g(x,\xi) - (-1)^{|V|}D_{\xi}f(x,\xi) ) \parx \,\, \text{mod} \,\, \mathcal{D}
\end{split}
\end{equation}
which yields that the local form on a section of $\mathcal{W}$ is
$$
V = f(x, \xi) \parx + (-1)^{|V|}\frac{D_{\xi}f(x,\xi)}{2} D_{\xi}.
$$
The above local form of $V$ implies that the natural map (of sheaves of super $\mathbb{C}$ vector spaces)
$$ \mathcal{W} \to \mathcal{T}_X/ \mathcal{D} $$
is an isomorphism. We emphasize that the above map is {\bf not} a map between two vector bundles, but nevertheless it will be useful to us in computation of cohomology.

\bigskip

\section{Residues and Serre Duality on SUSY Curves}

\subsection{Theory of Residues - Basic Definitions}

\begin{defx}
A SUSY-disk, is a non-compact super Riemann surface $\Delta \subset \Cx^{1|1}$ whose underlying topological space is a classical disk in $\Cx$
$$
|\Delta| = \{ z \in \Cx \, : \, |z| < \epsilon \},
$$
with a choice of global superconformal coordinates $z  | \tth$ on $\Delta$. That is, we take $\partial_{\tth} + \tth \partial_z$ as the generator for the odd distribution $\mathcal{D}$.

\end{defx}

\begin{defx} (\textbf{Residue - Absolute Case}) Let $\Delta$ be a SUSY-disk with coordinates $z|\tth$. Then given a meromorphic section $\eta$ of $\w = \ber \Omega^1_{\Delta}$ we write
$$
\eta = \left( \sum_{k \geq -N}^{\infty} (\alpha_k + \beta_k \tth) z^k \right ) [dz \, | \, d\tth]
$$
for $\alpha_k, \beta_k \in \Cx$. Then the residue of $\eta$ at zero is
$$
\emph{res}_0(\eta) := \beta_{-1}.
$$
\end{defx}

It is easy to check that this is independent of superconformal change of coordinates on $\Delta$.

\begin{defx}
Let $\pi: X \to S$ be a SUSY family and $q: S_{\tred} \to X_{\tred}$ a section of the reduced family. A SUSY-tubular neighborhood of $q$ is an open set $U \subset X$ containing $\text{Im} \, q$ with the restricted superconformal structure. We say a SUSY-tubular neighborhood is trivial if there exists a SUSY-disk $\Delta$ and an isomorphism of families of super Riemann surfaces 
$$
\alpha: U \cong S \times_{\Cx} \Delta
$$
such that the following diagram commutes
$$
\begin{tikzcd}
U_{\tred} \arrow[rd, "\pi_{\tred}"] \arrow[rr, "\alpha_{\tred}"] &                                                                    & (S \times_{\Cx} \Delta)_{\tred} \arrow[ld, "(\text{pr}_1)_{\tred}"'] \\
                                                                 & S_{\tred} \arrow[lu, "q", bend left] \arrow[ru, "q'"', bend right] &                                                                     
\end{tikzcd}
$$
where $q'(s) = (s,0)$.
\end{defx}

Each such trivial SUSY-tubular neighborhood then has a canonical set of relative superconformal coordinates coming from the pullback of the standard ones $z  |  \tth$ on the SUSY-disk via $\alpha$. We will frequently abuse notation and still denote these pullbacks by $z | \tth$.

\begin{defx} (\textbf{Residue - Relative Case}) \label{res_rel_def} Let $\pi:X \to S$ be a SUSY family, $q: S_{\tred} \to X_{\tred}$ a section and $U$ a trivial SUSY-tubular neighborhood of q. Let $j_q = j: U \setminus \emph{Im}\, q \to X$ denote the inclusion.

Suppose $\eta$ is a local section of $\pi_*j^*\w$. Write
$$
\eta = \left( \sum_{k \geq -N}^{\infty} (\alpha_k + \beta_k \tth) z^k \right ) [dz \, | \, d\tth],
$$
for $\alpha_k, \beta_k$ local functions on $S$. Then we define the residue of $\eta$ at $q$
$$
\emph{res}_q(\eta) := \beta_{-1}.
$$
Thus the residue yields a morphism
$$
\emph{res}_q : \pi_*j^*\w \to \oo_S
$$
\end{defx}

Of course, one must go ahead and prove that the above definition is independent of all the choices made but this is done very much in similar fashion as in the classical case. Here however, it is very important we work on SUSY curves as only superconformal coordinate transformations preserve the residue.

\subsection{The Relative \v{C}ech Complex}~ \label{dist_rel_cech}

Let $\pi: X \to S$ be a family of SUSY curves and assume we are given sections $q_1, \dots, q_r: S_{\tred} \to X_{\tred}$, for some $r \geq 1$, of the reduced family $\pi_{\tred}: X_{\tred} \to S_{\tred}$ such that the following is true: For each $k=1, \dots, r$ we can find a trivial SUSY-tubular neighborhood $U_k$ of $q_k$
        $$
        U_k \cong S \times_{\Cx} \Delta_k.
        $$
        so that the $U_k$'s are pairwise disjoint.

For each $k=1,\dots,r$ choose superconformal coordinates $z_k|\tth_k$ on the corresponding SUSY-disks $\Delta_k$ and identify them with relative local superconformal coordinates on $U_k$. Identify $q_k$ with its image and let $j_k = j_{q_k}: U_k \setminus q_k \to X$ denote the inclusion. 

Under these assumptions we define a canonical relative \v{C}ech complex for which we will use for cohomology computations and give an explicit description of Serre duality.

Define the open set $U^0 = X \setminus (q_1 \cup \dots \cup q_r)$ and let $\mathcal{F}$ be quasi-coherent and flat over $S$. For any open $U \subset X$ with natural inclusion $j: U \to X$ we denote the ``push-pull" sheaf in the usual way
$$
j_*j^*\mathcal{F} = \mathcal{F}{\big|}_{U}.
$$
For each such $U$ we have a natural morphism of sheaves
$$
\mathcal{F} \to \mathcal{F}{\big|}_{U}.
$$
\begin{defx} (Relative \v{C}ech Complex) \label{def_rel_cech} Let $\mathcal{F}$ be a quasi-coherent sheaf on $X$, flat over $S$. In the notation used above the relative \v{C}ech complex $\check{C}^{\bullet}(X,\mathcal{F})$ is the complex
\begin{equation} \label{cech_com2}
    0 \longrightarrow \mathcal{F}{\big|}_{U^0} \oplus  \bigoplus_{k=1}^r \mathcal{F}{\big|}_{U_k} \overset{d}{\longrightarrow} \bigoplus_{k=1}^r \mathcal{F}{\big|}_{U_k \setminus q_k} \longrightarrow 0.
\end{equation}
where
\begin{equation*}
    \left ( d(s_0, s_1, \dots, s_r) \right )_k = s_0{\big |}_{U_k \setminus q_k} - s_k{\big |}_{U_k \setminus q_k}.
\end{equation*}
\end{defx}

This complex will prove useful for computations below thanks to the following proposition.

\begin{prop} The complex
$$
0 \longrightarrow \mathcal{F} \longrightarrow \check{C}^{\bullet}(X,\mathcal{F})
$$
is a $\pi_*$-acyclic resolution of $\mathcal{F}$. Hence, in particular we have the natural identification 
$$
R^i\pi_*\mathcal{F} = H^i(\pi_*\check{C}^{\bullet}(X,\mathcal{F})).
$$
\end{prop}
\begin{proof}
Let $U \in \{U^0, U_1, \dots, U_r, U_1 \setminus q_k, \dots, U_r \setminus q_r \}$ be one of the open sets which appear in the definition of $\check{C}^{\bullet}(X, \mathcal{F})$ and $j: U \to X$ the inclusion. Let $s \in S$ be a point in the base. We claim that for $i \geq 1$
\begin{equation} \label{claim1_acyclic}
    H^i(X_s, (\mathcal{F}{\big|}_{U}){\big|}_{X_s}) = H^i(X_s, \mathcal{F}{\big|}_{U \cap X_s}) = 0.
\end{equation}
Indeed $\mathcal{F}{\big|}_{U \cap X_s}$ has support in the open set $U \cap X_s \subset X_s$ which is a Stein manifold. Hence all higher sheaf cohomology groups vanish.

By (\ref{claim1_acyclic}) the natural base change map
\begin{equation}
    \left( R^i\pi_*\mathcal{F}{\big |}_U \right )_s \otimes_{\oo_{S,s}} k(s) \longrightarrow H^i(X_s, \mathcal{F}{\big|}_{U \cap X_s})
\end{equation}
is trivially surjective, hence an isomorphism by the cohomology and base change theorm. Thus by Nakayama, $R^i\pi_*\mathcal{F}{\big |}_U = 0$ for $i \geq 1$.
\end{proof}

\subsection{Serre Duality and the Trace Map}~

If $\mathcal{F} = \w = \ber \Omega^1_{X/S}$ is the relative Berezinian sheaf then by definition (\ref{res_rel_def}) and direct sum we obtain a map on the pushforward of relative 1-cochains
$$
\sum_{k=1}^r \text{res}_{q_k}: \pi_*\check{C}^1(X, \w) \to \oo_S,
$$
which clearly descends to a map on the quotient
$$
\sum_{k=1}^r \text{res}_{q_k}: R^1\pi_*\w = \pi_*\check{C}^1(X, \w)/d\pi_*\check{C}^0(X, \w) \to \oo_S.
$$
This is the \emph{trace map} of Serre duality $\text{tr}: R^1\pi_*\w \to \oo_S$ and induces a perfect pairing 
$$
R^i\pi_*(\mathcal{F}) \otimes R^{1-i}\pi_*(\mathcal{F}^{\vee} \otimes \w) \to \oo_S
$$
for $\mathcal{F}$ locally free. A more general and complete account of Serre duality in the super case is given in \cite{elmsg}, but for our purposes the previous description suffices.

\subsection{An Induced Long Exact Sequence Computation.}~ \label{main_comp}

  Suppose we have a family of SUSY-curves $\pi:X \to S$ for which we are given a global section $t$ of an invertible sheaf $\mathcal{F}$, flat over $S$, such that $t_{\tred}$ has simple zeros $\{q_1, \dots, q_r\}$. From this data we get a canonical short exact sequence
$$
0 \longrightarrow \oo_X \overset{t}{\longrightarrow} \mathcal{F} \longrightarrow \mathcal{F}|_{T} \longrightarrow 0
$$
where $T = \{ t=0 \}$, and an induced long exact sequence of sheaves of higher direct images
$$
0 \longrightarrow \pi_*\oo_X \overset{t}{\longrightarrow} \pi_*\mathcal{F} \longrightarrow \pi_*\mathcal{F}{\big |}_T \overset{\delta}{\longrightarrow} R^1\pi_*\oo_X \overset{t}{\longrightarrow} R^1\pi_*\mathcal{F} \longrightarrow 0.
$$

Decompose $T = \sum T_k$ into pairwise disjoint prime divisors so that $T_{red} = \sum q_k$. Then we can view each $q_k$ as a section of the reduced family $q_k: S_{\tred} \to X_{\tred}$. Assume furthermore we can find trivial SUSY tubular neighborhoods as in the assumptions of Section \ref{dist_rel_cech}.

We wish to analyze this long exact sequence in more detail within the context of the relative \v{C}ech complex in Definition (\ref{def_rel_cech}). We have the following diagram
$$
\begin{tikzcd}
\pi_*\oo_X \arrow[r, "t"] \arrow[d]                                             & \pi_*\mathcal{F} \arrow[r] \arrow[d]                                                                            & \pi_*\mathcal{F}{\big |}_T \arrow[d]                                                                                \\
\pi_*\oo_{U^0} \oplus \bigoplus_{k=1}^r \oo_{U_k} \arrow[r, "t"] \arrow[d, "d"] & \pi_*\mathcal{F}{\big |}_{U^0} \oplus \bigoplus_{k=1}^r \pi_*\mathcal{F}{\big |}_{U_k} \arrow[r] \arrow[d, "d"] & \pi_*\mathcal{F}{\big |}_{T \cap U^0} \oplus \bigoplus_{k=1}^r \pi_*\mathcal{F}{\big |}_{T \cap U_k} \arrow[d, "d"] \\
\bigoplus_{k=1}^r \oo_{U_k \setminus q_k} \arrow[r, "t"] \arrow[d]              & \bigoplus_{k=1}^r \pi_*\mathcal{F}{\big |}_{U_k \setminus q_k} \arrow[r] \arrow[d]                              & \bigoplus_{k=1}^r \pi_*\mathcal{F}{\big |}_{T \cap (U_k \setminus q_k)} \arrow[d]                                   \\
R^1\pi_*\oo_X \arrow[r, "t"]                                                    & R^1\pi_*\mathcal{F} \arrow[r]                                                                                   & 0                                                                                                                  
\end{tikzcd}
$$

The connecting homomorphism $\delta: \pi_*\mathcal{F}{\big |}_T \to R^1\pi_*\oo_X$ is computed in the usual snake-lemma way, by follwing the zig-zag pattern on the above diagram starting from the upper right to the lower left. Specifically let $f \in \pi_*\mathcal{F}{\big |}_T$ be a local section and interpret $f$ as the element
$$
f = (f^0 = 0, f_1, \dots, f_r)
$$
where $f_k \in \pi_*\mathcal{F}{\big |}_{T \cap U_k}$. Lift $f$ to an element $\Tilde{f} = (0, \Tilde{f}_1, \dots, \Tilde{f}_r )$ and apply the \v{C}ech differential
$$
d \Tilde{f} = (\Tilde{f}_1|_{U_1 \setminus q_1}, \dots, \Tilde{f}_r|_{U_r \setminus q_r}).
$$
Lifting then under the multiplication by $t$ map (division by $t$ as we are away from $\{ t = 0 \}$ and looking at the image in the quotient, we conclude that the connecting homomorphism $\delta$ sends $f \in \pi_*\mathcal{F}{\big |}_T$ to the cohomology class
$$
\delta(f) = \left[ \frac{d \Tilde{f}}{t} \right ] = \left [\frac{\Tilde{f}_1}{t}{\bigg |}_{U_1 \setminus q_1}, \dots, \frac{\Tilde{f}_r}{t}{\bigg |}_{U_r \setminus q_r}  \right ] \in R^1\pi_*\oo_X.
$$
The composition of this with the Serre-dual map $R^1\pi_*\oo_X \to (\pi_*\w)^{\vee}$ we denote by $\delta'$. This is explicitly
\begin{equation} \label{delta'_def}
    \begin{split}
        & \delta': \pi_*\mathcal{F}{\big |}_T \to (\pi_*\w)^{\vee} \\
        f \mapsto \bigg ( \eta & \mapsto \sum_{k=1}^r \text{res}_{q_k} \bigg( \eta \big|_{U_k} \frac{\Tilde{f}_k}{t}{\bigg |}_{U_k \setminus q_k} \bigg) \bigg ).
    \end{split}
\end{equation}

We will use this explicit description of $\delta'$ multiple times in the main arguments presented below. We also consider a slightly more general argument than that just given; namely when $\oo_X$ is replace by some other invertible sheaf $\mathcal{F}'$. The changes are minimal and we do not pause to comment further.

\section{Punctures}~

Scattering amplitudes of superstring theory are written as integrals over moduli spaces of slightly more general objects than strictly super Riemann surfaces. These are \emph{punctured super Riemann surfaces} which we discuss now.

There are two types of punctures one can consider in the theory of SUSY curves, known as Neveu-Schwarz and Ramond punctures. Neveu-Schwarz punctures are more familiar, while Ramond punctures are a bit exotic. We focus on Ramond punctures first.

\medskip

\subsection{Ramond Punctures}~ \label{ramondpunctures}

Suppose $\pi: X \to S$ of $1 | 1$ is a family of supercurves along with an odd distribution $\mathcal{D} \subset \mathcal{T}_{X/S}$ such that the Lie bracket
$$
\mathcal{D}^{\otimes 2} \xrightarrow{[\cdot, \cdot]} \mathcal{T}_{X/S}/\mathcal{D}
$$
fails to be an isomorphism along an effective relative divisor $\mathcal{F}$, in the sense that instead $[\cdot, \cdot]$ induces an isomorphism
$$
\mathcal{D}^{\otimes 2} \xrightarrow{[\cdot, \cdot]} \mathcal{T}_{X/S}/\mathcal{D} \otimes \stsh_X(-\mathcal{F}).
$$
In this case the family $\pi: X \to S$ is called \emph{a family of super Riemann surfaces with Ramond punctures} or a \emph{family of SUSY curves with Ramond punctures}. The divisor $\mathcal{F}$ is called the \emph{Ramond divisor}. If we write $\mathcal{F}$ as a sum of minimal divisors (irreducible divisors)
$$
\mathcal{F} = \sum_{k=1}^{n_R} \mathcal{F}_k,
$$
then each $\mathcal{F}_k$ is called a \emph{Ramond puncture}. One can think of a Ramond puncture as a ``puncture" in the distribution $\mathcal{D}$ itself.

Locally near a Ramond puncture $\mathcal{F}_k$ we can find a coordinate chart $z | \zeta$ so that $\mathcal{F}_k$ is given by $z = 0$ and that $\mathcal{D}$ is locally generated by $D^*_{\zeta} = \partial_\zeta + z \zeta \partial_z$. Such coordinates are also called \emph{superconformal}. The usual exact sequence now becomes
$$ 0 \longrightarrow \mathcal{D} \longrightarrow \mathcal{T}_X \longrightarrow \mathcal{D}^{2}(\mathcal{F}) \longrightarrow 0. $$
Dualizing and taking Berezinians we conclude $\omega = \ber X/S \cong \mathcal{D}^{-1}(-\mathcal{F})$. In fact, in this case $\omega$ remains a relative dualizing sheaf.

The reduction of the Ramond divisor is a sum of points on the Riemann surface $X_{\tred}$,
$$
\mathcal{F}_{\tred} = \sum_{k=1}^{n_R} q_k
$$
thus reducing the above exact sequence gives
$$ 0 \longrightarrow \mathcal{D}_{\tred} \longrightarrow (\mathcal{T}_X)_{\tred} \longrightarrow \mathcal{D}_{\tred}^{\otimes 2}(\sum^{n_R}_{k=1} q_k) \longrightarrow 0. $$
Note that the rank $0 | 1$ line bundle $\mathcal{D}_{\tred}$ on $X_{\tred}$ is concentrated in odd degree, i.e. its even part $\mathcal{D}_{\tred,0}$ is zero. Similarly the odd part of $\mathcal{D}_{\tred}^{\otimes 2}(\sum^{n_R}_{k=1} q_k) = 0$ is zero and hence the above exact sequence actually splits canonically
$$
(\mathcal{T}_X)_{\tred} = (\mathcal{T}_X)_{\tred,0} \oplus (\mathcal{T}_X)_{\tred,1} \cong \mathcal{D}_{\tred}^{\otimes 2}(\sum^{n_R}_{k=1} q_k) \oplus \mathcal{D}_{\tred}.
$$
Now by the Proposition \ref{tangentsheaf} we always have the identification $(\mathcal{T}_X)_{\tred,0} = \mathcal{T}_{X_{\tred}}$, thus
$$
\mathcal{D}_{\tred}^{\otimes 2}(\sum^{n_R}_{k=1} q_k) \cong \mathcal{T}_{X_{\tred}}.
$$
In particular taking degrees gives
$$
2 \, \text{deg }\mathcal{D}_{\tred} + n_R = 2-2g,
$$
where $g$ is the genus of the surface $X$. One concludes that the number of Ramond punctures $n_R$ must be even. 

We denote the moduli space of super Riemann surfaces with $n_{R}$ Ramond punctures by $\mathfrak{M}_{g;n_{R}}$.

\bigskip

\medskip

\subsection{Neveu-Schwarz Punctures}~ \label{nspuctures_sec}

Suppose $\pi: X \to S$ is a SUSY family. A \emph{Neveu-Schwarz} (NS) puncture is simply a section $s: S \to X$ of the map $\pi$. Such a section is locally of the form $ U \to U \times \Cx^{1|1}$, for $U \subset S$ open, and hence equivalent to a choice of an even and odd function on $U$. Hence it is common to say that an NS puncture is given in local superconformal coordinates by $z = z_0, \zeta = \zeta_0$ for some choice of even and odd functions $z_0, \zeta_0$ on the base $S$. 

Given an NS puncture $s$ we have a natural associated divisor using the distribution $\mathcal{D}$. Namely, we use $s$ to pullback $\mathcal{D}$ and then take its total space $s^*\mathcal{D}^{\text{tot}}$.
$$
\begin{tikzcd}
(s^*\mathcal{D})^{\text{tot}} \arrow[r] \arrow[d]
& \mathcal{D}^{\text{tot}} \arrow[d, "p"] \\
S \arrow[r, "s"] & X
\end{tikzcd}
$$
This gives a subvariety $(s^*\mathcal{D})^{\text{tot}} \to X$, which is of relative dimension $0 | 1$ over $S$. We will denote this subvariety associated to $s$ by div$(s)$. Given such a family $\pi:X \to S$ with $n_{NS}$ NS punctures $s_1, \dots, s_{n_{NS}}$, we denote by $N = \sum_{j=1}^{n_{NS}} \text{div}(s_j)$ the \emph{Neveu-Schwarz divisor}.

We denote the moduli space of super Riemann surfaces with $n_{NS}$ Neveu-Schwarz punctures by $\mathfrak{M}_{g;n_{NS}}$.

\bigskip

\section{The Moduli Spaces $\mathfrak{M}_{g,n_{NS},n_R}$} ~

We now turn our attention to the various moduli spaces of interest. By $\mathfrak{M}_{g,n_{NS},n_R}$ we mean the moduli space (stack) of super Riemann surfaces of genus $g$ with $n_{NS}$ Neveu-Schwarz and $n_{R}$ Ramond punctures. If one or both of $n_{NS}$ or $n_R$ is zero we will simply write $\mathfrak{M}_{g}, \mathfrak{M}_{g,n_{NS}}$ or $\mathfrak{M}_{g,n_R}$. These stacks are fine moduli spaces in their appropriate categories.

\medskip

\subsection{The Tangent Sheaf of $\mathfrak{M}_g$}~

A closed point in $\mathfrak{M}_g$ corresponds to a super Riemann surface $X_0$, and the tangent space
$$
T_{X_0}\mathfrak{M}_g = (T_{X_0}\mathfrak{M}_g)_0 \oplus (T_{X_0}\mathfrak{M}_g)_1
$$
splits as usual. One has the following characterization of the even an odd parts of $T_{X_0}\mathfrak{M}_g$: an even tangent vector to $\mathfrak{M}_g$ at $X_0$ is a map
$$
\text{Spec }\mathbb{C}[\varepsilon]/(\varepsilon^2) \to \mathfrak{M}_g
$$
where $\varepsilon$ is an even parameter and mapping the unique $\mathbb{C}$ point of $\text{Spec }\mathbb{C}[\varepsilon]/(\varepsilon^2)$ to $X_0$. By definition of a fine moduli space this is equivalent to a family over the base $\text{Spec }\mathbb{C}[\varepsilon]/(\varepsilon^2)$ induced by the universal curve $\mathcal{C}_g$
$$
\begin{tikzcd}
\text{Spec }\mathbb{C}[\varepsilon]/(\varepsilon^2) \times _{\mathfrak{M}_g} \mathcal{C}_g \arrow[r] \arrow[d]
& \mathcal{C}_g \arrow[d, "p"] \\
\text{Spec }\mathbb{C}[\varepsilon]/(\varepsilon^2) \arrow[r, "s"] & \mathfrak{M}_g
\end{tikzcd}
$$
whose special fiber is $X_0$. In other words an element of $(T_{X_0}\mathfrak{M}_g)_0$ is family $X \to \text{Spec }\mathbb{C}[\varepsilon]/(\varepsilon^2)$ whose fiber over the point $(\varepsilon)$ is $X_0$. We will call such a family an \emph{even first order deformation} of the single super Riemann surface $X_0$.

The argument is similar for the analysis of $(T_{X_0}\mathfrak{M}_g)_1$, namely an odd tangent vector at $X_0$ is a family $X \to \text{Spec }\mathbb{C}[\eta]$ with $\eta$ an odd parameter and whose fiber over $(\eta)$ is $X_0$. We will call these families \emph{odd first order deformations} of $X_0$.

Now the key observation is that families over the two special bases specified above with special fiber $X_0$ are parametrized by the cohomology group $H^1(X_0, \mathcal{W}_0)$, where $\mathcal{W}_0$ is the sheaf of superconformal vector fields on $X_0$. This observation follows from, e.g., the \v{C}ech description of $H^1(X_0, \mathcal{W}_0)$ as locally automorphisms of $X_0$ are described by infinitesimal automorphisms, i.e. vector fields preserving the given structure. In conclusion 
$$
T_{X_0}\mathfrak{M}_g = H^1(X_0, \mathcal{W}_0),
$$
or in global terms in turns out \cite{wit1},
$$
\mathcal{T}_{\mathfrak{M}_g} = R^1\pi_*\mathcal{W}
$$
where $\mathcal{W}$ is the sheaf of relative superconformal vector fields and $\pi: \mathcal{C}_g \to \mathfrak{M}_g$ is the projection from the universal curve.

This observation allows us to compute the dimension of $\mathfrak{M}_g$. To simplify notation denote by $X_0$ simply $X$ and $\mathcal{W}_0$ by $\mathcal{W}$. As mentioned above the sheaf $\mathcal{W}$ is canonically isomorphic (as a sheaf of $\mathbb{C}$-vector spaces) with the quotient sheaf $\mathcal{T}_{X} / \mathcal{D}$ and hence, by the SUSY structure, to $\mathcal{D}^{\otimes 2}$. As $X$ is a single SUSY curve we have the natural splitting of $\stsh = \stsh_{\tred} \oplus \J$ and also every invertible sheaf on $X$ splits in a similar way, in particular
$$
\mathcal{W} \cong \mathcal{D}^{\otimes 2} = \mathcal{D}^{\otimes 2}_{\tred} \oplus ( \J \otimes \mathcal{D}^{\otimes 2}_{\tred}). 
$$
By the proof of Proposition \ref{tangentsheaf} we saw that $\mathcal{D}_{\tred} \cong \J^{-1}$ and that $\Pi \mathcal{D}_{\tred}^{-1} = \Omega^{\otimes 1/2}_{X_{\tred}}$ was a square root of the canonical bundle $\Omega_{X_{\tred}}$ of the reduced space. Thus by classical Serre duality on the reduced space, we conclude that the even and odd parts of $H^1(X, \mathcal{W})$ are,
$$
H^1(X, \mathcal{W})_0 \cong H^1(X, \mathcal{D}^{\otimes 2}_{\tred}) = H^0(X, \Omega^{\otimes 2}_{X_{\tred}})^*
$$
$$
H^1(X, \mathcal{W})_1 \cong H^1(X, \mathcal{D}_{\tred}) = \Pi H^0(X, \Omega^{\otimes 3/2}_{X_{\tred}})^*.
$$
The Riemann Roch theorem on $X_{\tred}$ then immediately gives
$$
\text{dim } \mathfrak{M}_g = 3g - 3 \, | \, 2g-2.
$$
Adjusting the above arguments slightly allows one to compute the dimensions of the moduli spaces with punctures. One must now consider a subsheaf $\mathcal{W}'\subset \mathcal{W}$ of the sheaf of superconformal vector fields that yield infinitesimal automorphisms of punctured SUSY curves. The sheaf $\mathcal{W}'$ still has the important property that
$$ R^1\pi_*\mathcal{W}' = T_{\mathfrak{M}_{g, n_{NS}, n_{R}}}. $$
We analyze the sheaf $\mathcal{W}'$ separately for the cases of Neveu-Schwarz and Ramond punctures.

\medskip

\subsection{The Tangent Sheaf of $\mathfrak{M}_{g, n_{NS}}$}~

For a family of SUSY curves $\pi: X \to S$ with $n_{NS}$ punctures $s_1, \dots, s_{n_{NS}}$, sections of $\mathcal{W}'$ are defined to be vector fields preserving the distribution $\mathcal{D}$ in the sense that $[\mathcal{W}', \mathcal{D}] \subset \mathcal{D}$ which also vanish on the divisor $N := \sum_k \text{div}(s_k)$. In local coordinates one can verify the analogous isomorphism (of sheaves of $\mathbb{C}$-vector spaces)
$$
\mathcal{W}' \cong \mathcal{T}_{X/S} / \mathcal{D} \big ( - N \big ).
$$
In the presence of $NS$ punctures $\mathcal{D}$ is still maximally non-integrable. Using the other relationships derived for SUSY curves above, letting $N_{\tred} = \sum_k s_k$
$$
\mathcal{W}' \cong \mathcal{D}^{\otimes 2}(-N) = \mathcal{D}^{\otimes 2}_{\tred}\big (- N_{\tred} \big ) \oplus \mathcal{D}_{\tred} \big ( - N_{\tred} \big )
$$
thus,
$$
H^1(X, \mathcal{W}')_0 \cong H^0(X, \Omega^{\otimes 2}_{X_{\tred}}( N_{\tred} ) )^*
$$
$$
H^1(X, \mathcal{W}')_1 \cong \Pi H^0(X, \Omega^{\otimes 3/2}_{X_{\tred}}( N_{\tred} ) )^*.
$$
Then as $\text{deg}\, N_{\tred} = n_{NS}$ we have
$$
\text{dim } \mathfrak{M}_{g, n_{NS}} = 3g - 3 + n_{NS} \, | \, 2g-2 + n_{NS}.
$$
\medskip

\subsection{The Tangent Sheaf of $\mathfrak{M}_{g, n_{R}}$}~

The analysis in the presence of Ramond punctures is slightly more subtle. Suppose we have a single super Riemann surface with Ramond punctures, i.e. a family $\pi: X \to \text{Spec} (\mathbb{C})$ over a point of super Riemann surfaces with Ramond punctures. The maximally non-integrable condition is replaced by the isomorphism
$$
\mathcal{D}^{\otimes 2} \xrightarrow{[\cdot, \cdot]} \mathcal{T}_{X}/\mathcal{D} \otimes \stsh_X(-\mathcal{F})
$$
for an effective divisor $\mathcal{F}$ called the \emph{Ramond divisor}. We write as above $\mathcal{F} = \sum_k \mathcal{F}_k$, with each $\mathcal{F}_k$ a minimal/irreducible divisor, and $\mathcal{F}_{\tred} = \sum_k q_k$. For this specific case when dealing with a single curve the structure sheaf $\stsh = \stsh_{\tred} \oplus \J$ splits as usual and as in the proof of Proposition \ref{tangentsheaf} above we have that the odd part of the reduced tangent sheaf is the dual $\J^* = \mathcal{H}om_{\stsh_{\tred}}(\J, \stsh_{\tred})$,
$$
(\mathcal{T}_X)_{\tred, 1} \cong \J^*.
$$
On the other hand we know from above we know the even and odd parts of $(\mathcal{T}_X)_{\tred}$ thus
$$
\mathcal{D}_{\tred} \cong \J^*,  \hspace{.5cm} \mathcal{D}^{\otimes 2} _{\tred}(\sum_k q_k) \cong \mathcal{T}_{X_{\tred}}.
$$
The sheaf $\mathcal{W}'$ of infinitesimal automorphisms is simply $\mathcal{W}$, the sheaf superconformal vector fields defined by $[\mathcal{W}, \mathcal{D}] \subset \mathcal{D}$. A local coordinate computation will show that 
$$ 
\mathcal{W} \cong \mathcal{D}^{\otimes 2}
$$ 
again as sheaves of $\mathbb{C}$-vector spaces. Thus the splitting $\mathcal{D}^{\otimes 2} = \mathcal{D}_{\tred}^{\otimes 2} \oplus \mathcal{D}_{\tred}$ gives
$$
H^1(X, \mathcal{W})_0 \cong H^1(X, \mathcal{D}^{\otimes 2}_{\tred}) = H^0(X, \Omega^{\otimes 2}_{X_{\tred}}(-\sum_k q_k))^*
$$
$$
H^1(X, \mathcal{W})_1 \cong H^1(X, \mathcal{D}_{\tred}) = \Pi H^0(X, \Omega_{X_{\tred}} \otimes \mathcal{D}^{-1}_{\tred}  )^*.
$$
In view of our argument showing the the number of Ramond punctures $n_R$ must be even, we get that $\deg \mathcal{D}_{\tred} = 1 - g - n_{R}/2$ and hence Riemann Roch yields
$$
\text{dim } \mathfrak{M}_{g, n_{R}} = 3g - 3 + n_{R} \, | \, 2g-2 + n_{R}/2.
$$
In fact when both NS and Ramond punctures are considered we get
$$
\text{dim } \mathfrak{M}_{g, n_{NS}, n_R} = 3g - 3 + n_{NS} + n_R \, | \, 2g-2 + n_{NS} + n_{R}/2.
$$

\bigskip

\section{The Riemann Roch Theorem on $ 1 | 1$ Supercurves}~

The Riemann Roch Theorem has a nice generalization to the setting of $1 | 1$ supercurves, we discuss it now. If $\Ll$ is an invertible sheaf on a supercurve $X$, the cohomology groups $H^i(X, \Ll)$ are naturally super $\mathbb{C}$-vector spaces. We let $h^i(X,\Ll) = \text{dim }H^i(X,\Ll)$ denote the dimension. If it is clear from the context we will frequently not mention $X$ and simply write $H^i(\Ll)$ for $H^i(X,\Ll)$. Recall that for a super vector space of dimension $m | n$ the \emph{superdimension} $\text{sdim}\,V$ is defined to be $m-n$ and thus we use the notation $sh^i(\Ll)$ to denote $\text{sdim}\,H^i(\Ll)$.

\begin{defx}
Suppose $X$ is a $1 | 1$ complex supercurve. Let $\Ll$ be an invertible sheaf on $X$. We define the Euler Characteristic of $\Ll$ to be
$$
\chi(\Ll) = h^0(\Ll) - h^1(\Ll)
$$
and the super Euler Characteristic to be single integer
$$
s\chi(\Ll) = sh^0(\Ll) - sh^1(\Ll).
$$
\end{defx}

A single supercurve $X$ is split and hence every invertible sheaf $\Ll$ on $X$ splits as the direct sum as two line bundles on the reduced space,
$$
\Ll = \Ll_{\tred} \oplus \left ( \Ll_{\tred} \otimes \J \right )
$$
with $\stsh_X = \stsh_{X_{\tred}} \oplus \J$ as usual. This observation allows us to formulate the Riemann Roch Theorem on $X$ as simply two applications of the classical Riemann Roch Theorem. 

\begin{thm} \label{RieRoch1}
(Riemann-Roch for Supercurves) Suppose $X$ is a $1 | 1$ supercurve and $\Ll$ an invertible sheaf locally free of rank $1|0$ on $X$, then
$$
\chi(\Ll) = ( \deg \Ll - g + 1 \, | \, \deg \Ll + \deg \J - g + 1).
$$
Hence in particular the super Euler Characteristic
$$
s\chi(\Ll) = -\deg \J
$$
in independent of $\Ll$.
\end{thm}
\begin{proof}
Immediate based on the observed decomposition of $\J$ and the classical Riemann Roch formula.
\end{proof}

In the special case $X$ is a SUSY curve we have

\begin{cor} \label{RieRoch2}
(Riemann-Roch for SUSY Curves) Suppose $X$ is a SUSY curve and $\Ll$ and invertible sheaf on it locally free of rank $1|0$, then
$$
\chi(\Ll) = ( \deg \Ll - g + 1 \, | \, \deg \Ll ).
$$
and
$$
s\chi(\Ll) = 1-g
$$
\end{cor}
\begin{proof}
By Proposition \ref{tangentsheaf} $\J$ is a spin structure on $X$ and hence has degree $g-1$.
\end{proof}

We remark that if we are working with a family of $1 | 1$ supercurves $\pi: X \to S$ the above two results still can be of use. Namely if $\Ll$ is an invertible sheaf on $X$ such that both $\pi_*\Ll$ and $R^1\pi_*\Ll$ are locally free over $S$ then the fiber of $R^i\pi_*\Ll$ at any point $s$ of $S$ is $H^i(X_s, \Ll|_{X_s})$. By the local freeness assumptions, the dimensions of these super vector spaces will not change as we vary $s$ and hence Theorem \ref{RieRoch1} and Corollary \ref{RieRoch2} compute the ranks of the vector bundles $R^i\pi_*\Ll$.

\bigskip

\section{The Super Mumford Isomorphism}~

We follow \cite{vormum} closely. Let $\pi: X \to S$ denote a family of $1 | 1$ supercurves. For any locally free sheaf $\mathcal{F}$ on $X$ we can consider the invertible sheaf $B(\mathcal{F})$ on $S$, called the \emph{Berezinian of cohomology of} $\mathcal{F}$. If each $R^i\pi_* \mathcal{F}$ is locally free, then $B(\mathcal{F})$ is given by
$$
B(\mathcal{F}) = \otimes_{i} \, (\ber R^i\pi_*\mathcal{F})^{(-1)^i}.
$$
Moreover, for every short exact sequence of locally free sheaves on $X$
$$
0 \longrightarrow \mathcal{F}' \longrightarrow \mathcal{F} \longrightarrow \mathcal{F}'' \longrightarrow 0,
$$
we get a canonical isomorphism
$$
B(\mathcal{F}') \otimes B(\mathcal{F}'') \cong B(\mathcal{F}).
$$
Hence, in particular any isomorphism $f: \mathcal{F} \to \mathcal{G}$ induces an isomorphism $B(f): B(\mathcal{F}) \to B(\mathcal{G})$. 

For $\w = \ber \Omega^1_{X/S}$ the relative Berezinian, we set for each $j$,
$$ \lambda_{j/2} = B(\w^{\otimes j}). $$
Serre duality gives the canonical identifications $\lambda_{j/2} \cong \lambda_{(1-j)/2}$. The super Mumford isomorphism(s) are the following canonical isomorphisms amongst the $\lambda_{j/2}$.

\begin{prop} \label{smumiso} (The Super Mumford Isomorphism)
For any family of $1 | 1$ supercurves $\pi: X \to S$ we have canonical isomorphisms
$$
\lambda_{j/2} \cong \lambda_{1/2}^{(-1)^{j-1}(2j-1)}.
$$
In particular,
$$
\lambda_{3/2} \cong \lambda_{1/2}^5.
$$
\end{prop}

Proofs of the super Mumford isomorphisms can be found in \cite{vormum} and \cite{vids}. We have also included a proof in Appendix \ref{smum_proof}. We will denote by $\mu$ the trivializing section of $\lambda_{3/2} \lambda_{1/2}^{-5}$. Such an object is called the \emph{super Mumford form}. 

In the following we will consider two separate situations:
\begin{enumerate}
    \item $\pi: X \to S$ is a family of super Riemann surfaces of genus $g \geq 2$ with $\nr$ Ramond punctures such that:
\begin{enumerate}
    \item The sheaves $R^i\pi_*\w^j$ are locally free for $i=0,1$, $j=-2, -1, 0, 1$.
    \item $\nr > 6g-6$.
\end{enumerate}
    \item $\pi: X \to S$ is a family of super Riemann surfaces of genus $g \geq 2$ with $n_{\text{NS}}$ Neveu-Schwarz punctures such that:
\begin{enumerate}
    \item The sheaves $R^i\pi_*\w^j$ are locally free for $i=0,1$, $j=0, 1, 2, 3$.
    \item $\pi_*\w$ has rank $g | 1$.
\end{enumerate}
\end{enumerate}
In each case we produce a concrete proof of the corresponding super Mumford isomorphism $\lambda_{3/2} = \lambda_{1/2}^5$ and use it to produce the main results.

We make heavy use of the higher direct image sheaves of the relative Berezinain $R^i\pi_*\w^j$ and emphasize that in both situations we work under the \emph{assumption} that these sheaves are locally free. We pause to discuss this assumption more in depth at the end of Section \ref{somerierochcalc}. 

The conditions listed for the Neveu-Schwarz case speak to the fact that we work over the component of the moduli space $\mathfrak{M}_{g;n_{NS}}$ corresponding to an odd spin structure. Hence, fiberwise $\Pi \w_{\tred}$ gives an odd nondegenerate theta characteristic.

We begin with the Ramond case.
\chapter{The Ramond Puncture Case} \label{r_case}
\label{ramond}


Here we derive our first main result: an explicit formula for the super Mumford form $\mu$ on the moduli space of super Riemann surfaces with $\nr$ Ramond Punctures $\msp$. This can then be used to create a measure on $\msp$ whose integral computes scattering amplitudes of superstring theory. The following arguments were heavily inspired by the work done in \cite{belman}, \cite{vormum} and \cite{RSVphy}.

Throughout this section we let $\pi: X \to S$ denote a family of genus $g \geq 2$ SUSY curves with $\nr$ Ramond punctures with $R^i\pi_*\w^j$ locally free and $\nr > 6g-6$. We denote the Ramond divisor by $\mathcal{F}$.

\bigskip

\section{Some Riemann-Roch Calculations}~ \label{somerierochcalc}

In the special case that $S$ is a point the structure sheaf $\stsh_X$ admits the global decomposition $\stsh_X = \stsh_{X_{\tred}} \oplus \J$. This allows one to decompose any super holomorphic line bundle $\Ll$ into the direct sum of two ordinary holomorphic line bundles over $X_{\tred}$ as
\begin{equation} \label{decomp1}
\Ll = \Ll_{\text{red}} \oplus \left( \Ll_{\text{red}} \otimes \J \right ).
\end{equation}
Here $\J \subset \stsh_X$ again denotes the sheaf of ideals generated by the odd elements. The summands above are exactly the even and odd parts of $\Ll$. We are most interested in the case $\Ll = \w^{\otimes j} = \w^j$ with our goal being to identify the ranks of these bundles. By the local freeness assumption (and the cohomology and base change Theorem), to compute these ranks it suffices to assume that $S$ is a point. Thus by the decomposition (\ref{decomp1})
\begin{equation} \label{rank1}
\text{rank} \, R^i\pi_*\Ll = h^i(\Ll_{\text{red}}) \, | \,  h^i(\Ll_{\text{red}} \otimes \J)
\end{equation}
for $\Ll$ of rank $1 | 0 $ (and vice versa for $\Ll$ or rank $0 | 1 $).

In Section \ref{ramondpunctures} we saw that the Berezinian sheaf $\w$ was identified with $\mathcal{D}^{-1}(-\F)$. Furthermore, one can easily see \cite{wit1} that if $S$ is a point, $\mathcal{T}_{X_\tred} = \mathcal{D}^2(\mathcal{F})_\tred$ and thus the distribution $\mathcal{D}$ had degree $1-g-\nr/2$. Arguing in this fashion, i.e. utilizing the classical Riemann-Roch theorem on $X_{\tred}$, a slightly tweaked Proposition \ref{jspin} and the assumption that $\nr > 6g-6$, allows one to complete the tables below of the various ranks of the sheaves $R^i\pi_*\w^j$. The calculations involved are somewhat tedious and we omit them here.
$$
\begin{array}{c|c|c|c}
j & \text{rank } \pi_*\wre^j & \text{rank } \pi_*(\wre^j \otimes \J) & \text{rank } \pi_*\omega^j \\ \hline
-2 &\nr + 3 - 3g & 3\nr/2 + 2 - 2g & \nr + 3 - 3g \, | \, 3\nr/2 + 2 - 2g \\
-1 &\nr/2 + 2 - 2g & \nr + 1 - g & \nr + 1 - g \,|\, \nr/2 + 2 - 2g   \\
0 & 1 & \nr/2 & 1 \, | \, \nr/2 \\
1 & 0 & g & g \, | \, 0 \\
\end{array}
$$
$$
\begin{array}{c|c|c|c}
j & \text{rank } R^1\pi_*\wre^j & \text{rank } R^1\pi_*(\wre^j \otimes \J) & \text{rank } R^1\pi_*\omega^j \\ \hline
-2 & 0 & 0 & 0 \, | \, 0 \\ 
-1 & 0 & 0 & 0 \, | \, 0   \\ 
0 & g & 0 & g \, | \, 0 \\ 
1 & \nr/2 & 1 & 1 \, | \, \nr/2 \\ 
\end{array}
$$

\medskip

Let us now address the issue of local freeness of the sheaves $R^i\pi_*\omega^j$ on the base $S$. Specifically for our purposes in the Ramond case we are interested in the local freeness of those sheaves with $i=0, 1$ and $j=-2, -1, 0, 1$. Common cohomology and base change arguments give immediately that the sheaves $R^i\pi_*\omega^j$ with $i=0,1$, $j=-1, -2$ are indeed locally free. Unfortunately for the others, it seems that there is no elementary argument to guarantee their local freeness. In fact, similar issues have been discussed in the literature before. A result in \cite{bruzzo2} essentially shows that there is no super version of Grauert's classical theorem of algebraic geometry, which would yield the desired result. The author would like to thank E. Witten for his helpful and stimulating comments regarding this issue.

\bigskip

\section{The Super Mumford Isomorphism $\lambda_{3/2}\lambda_{1/2}^{-5} \cong \stsh_S$ Explicitly }~

Our goal now is to prove explicitly the super Mumford isomorphism $\lambda_{3/2}\lambda_{1/2}^{-5} \cong \mathcal{O}_S$, following every step carefully to identify the trivializing section $\mu$ corresponding to the image of $1_S$ under the above isomorphism. We will express $\mu$ in terms of bases chosen for the locally free sheaves $R^i\pi_*\omega^j$.

The isomorphism will follow from 3 short exact sequences of sheaves on $X$

\begin{equation}
\label{firstSES} 0 \longrightarrow \Pi\omega \overset{t}{\longrightarrow} \oo \longrightarrow \oo |_{T} \longrightarrow 0,
\end{equation}

\begin{equation}
\label{secondSES} 0 \longrightarrow \oo \overset{t}{\longrightarrow} \Pi\w^{-1} \longrightarrow (\Pi\w^{-1}) |_{T} \longrightarrow 0,
\end{equation}

\begin{equation}
\label{thirdSES} 0 \longrightarrow \Pi\omega^{-1} \overset{t}{\longrightarrow} \w^{-2} \longrightarrow \w^{-2} |_{T} \longrightarrow 0,
\end{equation}
where $t = \Pi t'$ for $t'$ is an odd global section of $\w^{-1}$, and $T$ is the divisor $\{ t = 0 \}$. From these one concludes utilizing Serre duality (noting $B(\Pi \mathcal{F}) = B^{-1}(\mathcal{F})$)

\begin{equation} \label{3eqns}
    B(\oo|_T)\cong \lambda^2_{1/2}, \hspace{.25in} B((\Pi \w^{-1})|_T) \cong \lambda^{-1}_1 \lambda^{-1}_{1/2}, \hspace{.25in} \text{and} \hspace{.25in} B(\w^{-2}|_T) \cong \lambda_{3/2} \lambda_1.
\end{equation}

An important lemma is shown in \cite{vormum} and \cite{vids}, a proof of which is also given in Proposition \ref{can_iso_1_new} of the Appendix, stating that given any invertible sheaves $\mathcal{L}$ and $\mathcal{K}$ of on $X$ and any effective relative divisor $D$ of dimension $0|1$ over the base, we have a canonical isomorphism $B(\mathcal{L}|_D) \cong B(\mathcal{K}|_D)$. Using this result we get that the left hand sides of the three equations in (\ref{3eqns}) are all canonically identified and thus,
$$ 
\lambda^2_{1/2} \cong \lambda^{-1}_1 \lambda^{-1}_{1/2}, \hspace{.1in} \lambda_{3/2} \lambda_1 \cong \lambda^2_{1/2} \,\,\, \Longrightarrow \,\,\, \lambda_{3/2} \cong \lambda^5_{1/2}, 
$$
giving the super Mumford isomorphism.

We will follow the above argument in detail to identify $\mu$ in terms of specified bases for the sheaves $R^i\pi_*\w^j$.

\bigskip

\section{Various Bases}~ \label{variousbases}

To simplify notation we let $r := \nr/2 - g + 1$. We choose a distinguished odd global section $t'$ of $\w^{-1}$ such that $t'_{\tred}$ vanishes to first order at points $q_1, \dots, q_r$. Set $\Pi t' = t$ and near each point $q_k$ choose local superconformal coordinates $ z_k | \tth_k $ centered at $q_k$ so that $t$ expands in these coordinates as
$$ t' \sim z_kf_k(z_k | \tth_k)\bpk. $$
We denote by $T$ the divisor $\{ t' = 0 \} = \{t = 0 \}$ and assume it is disjoint from the Ramond divisor $\mathcal{F}$ (which is an open condition).

\bigskip 

\section*{Local Basis for $\pi_*\w$:}~

The rank of $\pi_*\w$ is $g | 0$, thus we choose an (even) basis 
$$ 
\mathcal{B}_{\w} := \{ \varphi_1, \dots, \varphi_g \}.
$$
Near each $q_k$ we expand
$$ \varphi_j \sim (\varphi_j^{k,-} + \varphi_j^{k,+}\tth_k + O(z_k)) [dz_k \, | \, d\tth_k], $$
where $\varphi_j^{k,\pm}$ are even/odd functions from the base. 

\section*{Local Basis for $\pi_*\stsh$:}~

We take 
$$ 
\mathcal{B}_{\stsh} := \{1 \, | \, t'\varphi_1, \dots, t'\varphi_g, \xi_1, \dots \xi_{r-1} \},
$$
to be a basis in $\pi_*\oo$ where
$$ \xi_j \sim (\xi_j^{k,-} + \xi_j^{k,+}\tth_k + O(z_k)) $$
near $q_k$.

\subsection*{Local Basis for $\pi_*\w^{-1}$:}~

Let 
$$
\mathcal{B}_{\w^{-1}} := \{ t'^2\varphi_1, \dots, t'^2\varphi_g, t'\xi_1, \dots, t'\xi_{r-1}, \sigma_1, \dots \sigma_r \, | \, t', \tau_1 \dots, \tau_{r-g} \},
$$
be a basis for $\pi_*\w^{-1}$. In local coordinates near each $q_k$ expand
$$ \sigma_j \sim (\sigma_j^{k,-} + \sigma_j^{k,+}\tth_k + O(z_k)) \bpk, $$
$$ \tau_j \sim (\tau_j^{k,+} + \tau_j^{k,-}\tth_k + O(z_k)) \bpk. $$

\subsection*{Local Basis for $\pi_*\w^{-2}$:}~

Let 
\begin{align*}
\mathcal{B}_{\w^{-2}} : = \{t'^2, t'\tau_1&, \dots, t'\tau_{r-g}, \eta_1, \dots, \eta_r \, |
\\
& \, t'^3\varphi_1, \dots, t'^3\varphi_g, t'^2\xi_1, \dots, t'^2\xi_{r-1}, t'\sigma_1, \dots t'\sigma_r, \psi_1, \dots, \psi_r \}, 
\end{align*} 
be a basis for $\pi_*\w^{-2}$ and as above expand near each $q_k$
$$ \eta_j \sim (\eta_j^{k,+} + \eta_j^{k,-}\tth_k + O(z_k))\bpk^2, $$
$$ \psi_j \sim (\psi_j^{k,-} + \psi_j^{k,+}\tth_k + O(z_k))\bpk^2. $$

\subsection*{Local Basis for $\pi_*(\w^j|_T)$:}~

We have singled out a specific odd global section $t'$ of $\w^{-1}$, for which we defined a divisor $T = \{t' = 0\}$. We assume this is disjoint from the Ramond Divisor $\mathcal{F}$ and then take
$$
\mathcal{B}_{\w^j|_T} := \{ [\partial_{z_1}\,|\,\partial_{\tth_1}]^j, \dots, [\partial_{z_r}\,|\,\partial_{\tth_r}]^j \, | \, \theta_1[\partial_{z_1}\,|\,\partial_{\tth_1}]^j, \dots, \theta_r[\partial_{z_r}\,|\,\partial_{\tth_r}]^j \} $$
to be a basis in $\w^j|_T$ (for $j=0$ we denote by $1_k = \bpk^0$ the element which is the function $1$ near each $q_k$, similarly $1_k\tth_k$ will sometimes be shortened to $\tth_k$). To shorten notation we will sometimes use $\varpi_k$ for $\bpk$.

Now by Serre duality we identify $R^1\pi_*\w^j \cong (\pi_*\w^{1-j})^*$, and hence by taking dual bases we get local bases $\mathcal{B}_{\stsh}^*$ and $\mathcal{B}_{\w}^*$  for $R^1\pi_*\w$ and $R^1\pi_*\stsh$ respectively (recall from Section \ref{somerierochcalc} that both $R^1\pi_*\w^{-1}$ and $R^1\pi_*\w^{-2}$ vanish). 

The five (ordered) bases above give rise to generating elements of their respective Berezinian of cohomology, 
\begin{equation} \label{gens_1}
\begin{split}
d_{1/2} & := \ber \mathcal{B}_{\w} \otimes \ber \mathcal{B}^*_{\stsh} \in \lambda_{1/2}\\
d_0 & := \ber \mathcal{B}_{\stsh} \otimes \ber \mathcal{B}^*_{\w} \in \lambda_{0} \\
d_{-1/2} & := \ber \mathcal{B}_{\w^{-1}} \in \lambda_{-1/2} \\
d_{-1} & := \ber \mathcal{B}_{\w^{-2}} \in \lambda_{-1} \\
\delta_{j/2} & := \ber \mathcal{B}_{\w^j|_T} \in B(\w^j|_T).
\end{split}
\end{equation}

\smallskip

\section{The First Short Exact Sequence}~ \label{firstSESsec}

We now study in detail the first short exact sequence shown in \eqref{firstSES}. Considering the induced long exact sequence in cohomology and utilizing Serre duality we obtain (Note: $R^1\pi_*(\stsh|_T)$ will vanish as $T$ has relative dimension $0|1$),
\begin{equation} \label{ses1}
0 \longrightarrow \Pi(\pi_*\w) \overset{t}{\longrightarrow} \pi_*\oo \longrightarrow \pi_*(\oo|_{T}) \overset{\delta'}{\longrightarrow} \Pi(\pi_*(\oo))^* \overset{t^*}{\longrightarrow} (\pi_*(\w))^* \longrightarrow 0.
\end{equation}

We emphasize that the morphism $\delta'$ is that as was described in  Section \ref{main_comp} and by equation (\ref{delta'_def}). We will sometimes write $\delta'(s) = \langle \cdot , \delta(s) \rangle$ where $\delta: \pi_*(\oo|_T) \to \Pi R^1\pi_*\w$ is the usual connecting homomorphsim and $\langle \cdot, \cdot \rangle$ is the Serre duality pairing. 

Under the first map the (odd) basis $\Pi\{\varphi_1, \dots \varphi_g\}$ of $\Pi(\pi_*\w)$ maps to $\{t\varphi_1, \dots t\varphi_g\}$ (note the presence of $t$ and not $t'$) which is then completed to the chosen basis $\mathcal{B}_{\stsh} = \{ 1 \, | \, t\varphi_1, \dots, t\varphi_g, \xi_1, \dots, \xi_{r-1} \}$ of $\pi_*\oo$. The restriction map then sends the $t\varphi_j$'s to zero and the remaining basis elements to $\{ 1|_T \, | \, \xi_1|_T, \dots, \xi_{r-1}|_T \}$. In terms of our chosen (ordered) basis $\mathcal{B}_{\stsh|_T}$ for $\pi_*(\oo|_T)$, these are in components
\begin{align} \label{res1}
\begin{split}
1|_T = & \sum_{k=1}^{r} 1_k, \\
\xi_j|_T = & \sum_{k=1}^{r} \xi_j^{k,-}1_k + \sum_{k=1}^{r} \xi_j^{k,+}\tth_k1_k.
\end{split}
\end{align}

We will complete $\{ 1|_T \, | \, \xi_1|_T, \dots, \xi_{r-1}|_T \}$ to a basis by taking certain lifts under $\delta'$ of elements of $\Pi(\pi_*(\oo))^*$ and compare this with $\mathcal{B}_{\stsh|_T}$.

As described in (\ref{delta'_def})
\begin{equation} \label{res3}
\delta'(s)(h) = \sum_{k=1}^r \text{res}_{q_k} \left ( \frac{hs}{t} \right )
\end{equation}
for $h \in \pi_*\stsh$. We remark that the meaning of taking the residue at $q_k$ of the expression $hs/t$ means to take the local function defining $s$ near $q_k$ and compute the resulting residue.

Now, in view of the exact sequence (\ref{ses1}) we see that the image of the map $\pi_*(\stsh|_{T}) \to \Pi ( \pi_*(\stsh))^*$ sending $s \mapsto \delta'(s) = \langle \cdot, \delta(s) \rangle$ is the kernel of $t^*: \Pi(\pi_*(\oo))^* \to(\pi_*(\w))^*$ which is $\text{span}\, \Pi \{1^* \, | \, \xi_1^*, \dots, \xi_{r-1}^*\}$. Thus $\langle \cdot, \delta(s) \rangle$ expands as
\begin{equation} \label{exp1}
\langle \cdot, \delta(s) \rangle = \langle 1, \delta(s) \rangle 1^* + \sum_{k=1}^{r-1}\langle \xi_k, \delta(s) \rangle \xi_k^*.
\end{equation}

On the other hand, the kernel of $t^*$ is spanned by the set
$$
\{\langle \cdot, \delta(1_1) \rangle , \dots, \langle \cdot, \delta(1_r) \rangle \, | \, \langle \cdot, \delta(\tth_1) \rangle , \dots, \langle \cdot, \delta(\tth_r) \rangle \},
$$
and hence equation (\ref{exp1}) applied to each member of the (ordered) set above yield the various expressions
\begin{align*}
\langle 1, \delta(1_k) \rangle = \text{res}_{q_k} \left ( \frac{1}{t} \right ), \hspace{.5cm} \langle 1, \delta(\tth_k) \rangle = \text{res}_{q_k} \left ( \frac{\tth_k}{t} \right ),
\end{align*}
\begin{align*}
\langle \xi_j, \delta(1_k) \rangle = \text{res}_{q_k} \left ( \frac{\xi_j}{t} \right ), \hspace{.5cm} \langle \xi_j, \delta(\tth_k) \rangle = \text{res}_{q_k} \left ( \frac{\xi_j \tth_k}{t} \right ),
\end{align*}
which in view of the local expansions of Section \ref{variousbases}, can be encoded in the $(2r \times r)$ matrix
$$ 
A' = 
\begin{bmatrix}
  \text{res}_{q_1} 1/t & \text{res}_{q_1} \xi_1/t  & \dots & \text{res}_{q_1} \xi_{r-1}/t \\
  \vdots & \vdots & &  \vdots \\
  \text{res}_{q_r} 1/t & \text{res}_{q_r} \xi_1/t & \dots & \text{res}_{q_r} \xi_{r-1}/t \\
  \text{res}_{q_1} \tth_1/t & \text{res}_{q_1} \xi_1\tth_1/t  & \dots & \text{res}_{q_1} \xi_{r-1}\tth_1/t \\
  \vdots & \vdots & &  \vdots \\
  \text{res}_{q_r} \tth_r/t & \text{res}_{q_r} \xi_1\tth_r/t & \dots & \text{res}_{q_r} \xi_{r-1}\tth_r/t \\
\end{bmatrix}.
$$

Letting $A = (a_{ij})$ denote any $(r \times 2r)$ left inverse of $A'$ gives us lifts to $\pi_*(\stsh|_{T})$ of the elements $\{ 1^* \, | \, \xi_1^*, \dots, \xi_{r-1}^* \}$,
\begin{equation} \label{lift1}
\begin{split}
\widetilde{1^*} = & \sum_{k=1}^r a_{1,j}1_k + \sum_{k=1}^r a_{1,j+r}\tth_k1_k, \\
\widetilde{\xi_{j-1}^*} = & \sum_{k=1}^r a_{j,k}1_k + \sum_{k=1}^r a_{j,k+r}\tth_k1_k.
\end{split}
\end{equation}
Therefore combining (\ref{res1}) and (\ref{lift1}) we have that the bases 
$$
\{1|_T, \widetilde{\xi^*_1},\dots,\widetilde{\xi^*_{r-1}} \, | \, \xi_1|_T, \dots, \xi_{r-1}|_T, \widetilde{1^*} \}
$$
and
$$ \{ 1_1,\dots, 1_r \, | \, \tth_1, \dots, \tth_r \} $$
of $\pi_*\stsh|_T$ are related by the matrix 
\begin{equation} \label{m0_r}
M_0=
\begin{bmatrix}
    1       & a_{2,1} & \dots &  a_{r,1} & \xi_1^{1,-} & \dots & \xi_{r-1}^{1,-} & a_{1,1}\\
    1       & a_{2,2} & \dots & a_{r,2} & \xi_1^{2,-} &  \dots & \xi_{r-1}^{2,-} & a_{1,2} \\
    1       & a_{2,3} & \dots & a_{r,3} & \xi_1^{3,-} &  \dots & \xi_{r-1}^{3,-} & a_{1,3} \\
    \vdots & \vdots && \vdots & \vdots & & \vdots & \vdots \\
    1       & a_{2,r} & \dots & a_{r,r} & \xi_1^{r,-} &  \dots & \xi_{r-1}^{r,-} & a_{1,r} \\
    0 & a_{2,r+1} & \dots & a_{r,r+1} & \xi_1^{1,+} &  \dots & \xi_{r-1}^{1,+} & a_{1,r+1} \\
    0 & a_{2,r+2} & \dots & a_{r,r+2} & \xi_1^{2,+} &  \dots & \xi_{r-1}^{2,+} & a_{1,r+2} \\
    \vdots & \vdots & \dots & \vdots & \vdots & \dots & \vdots & \vdots \\
	0 & a_{2,2r-1} & \dots & a_{r,2r-1} & \xi_{1}^{r-1,+} &  \dots & \xi_{r-1}^{r-1,+} & a_{1,2r-1} \\
    0 & a_{2,2r} & \dots &  a_{r,2r} & \xi_1^{r,+} & \dots & \xi_{r-1}^{r,+} & a_{1,2r} &
\end{bmatrix}.
\end{equation}
That is, we get
$$ \ber \{ 1|_T, \tilde{\xi_k^*} \, | \, \xi_k|_T, \tilde{1^*} \} = \ber M_0 \, \delta_{0} $$
in $B(\stsh|_T) = \ber \pi_*(\stsh|_T)$. Therefore under the canonical isomorphism
$$ B(\oo|_T) \cong \lambda_{1/2} \lambda_{0} $$
we have the identification
\begin{equation}
\ber \{ 1|_T, \tilde{\xi_k^*} \, | \, \xi_k|_T, \tilde{1^*} \} = \ber M_0 \, \delta_{0} = d_{1/2}d_0. 
\end{equation}

\bigskip

\section{The Second Short Exact Sequence}~ \label{secondSESsec}

We move on to analyze the second short exact sequence \eqref{secondSES}. Our specified bases for the sheaves listed here are compatible with this short exact sequence, except for the third term. The induced long exact sequence after Serre duality reads

$$ 0 \longrightarrow \pi_*\oo \overset{t}{\longrightarrow} \pi_*(\Pi \w^{-1}) \longrightarrow \pi_*(\Pi \w^{-1}|_{T}) \overset{\delta'}{\longrightarrow} (\pi_*\w)^* \longrightarrow 0. $$

We aim to replicate the argument given in Section \ref{firstSESsec} to relate the two bases we have for $\pi_*(\Pi \w^{-1}|_{T})$. The key again is understanding the connecting homomorphism $\delta: \pi_*(\Pi \w^{-1}|_T) \to R^1\pi_*\stsh$ and its composition with the Serre duality map $\delta': \pi_*(\Pi \w^{-1}|_T) \overset{\delta}{\to} R^1\pi_*\stsh \to (\pi_*\w)^*$.

Following the same argument as given in Section \ref{firstSESsec} and slightly generalized arguments of Section \ref{main_comp} we again have

\begin{equation} \label{sd2}
\langle \cdot, \delta(s) \rangle = \delta'(s) = \sum_{k=1}^r \text{res}_{q_k} \left( (-) \frac{s}{t} \right ) .
\end{equation}
In terms of the basis $\mathcal{B}_{\w}^*$, each $\langle \cdot, \delta(s) \rangle $ expands as
\begin{equation} \label{exp2}
\langle \cdot, \delta(s) \rangle = \sum_{j=1}^g \langle \varphi_j, \delta(s) \rangle \varphi^*_j.
\end{equation}
Thus, we apply (\ref{exp2}) to each member of the basis $\mathcal{B}_{\w^{-1}|_T}$ and encode it in a $(2r \times g)$ matrix $B'$. To simplify notation let $\varpi_k = \bpk $,

$$ 
B'  = 
\begin{bmatrix}
    \text{res}_{q_1} \varphi_1 \varpi_1/t & \dots & \text{res}_{q_1} \varphi_g \varpi_1/t \\
    \vdots & & \vdots \\
    \text{res}_{q_r} \varphi_1 \varpi_r/t & \dots & \text{res}_{q_r} \varphi_g \varpi_r/t \\
    \text{res}_{q_1} \varphi_1 \tth_1\varpi_1/t & \dots & \text{res}_{q_1} \varphi_g \tth_1\varpi_1/t \\
    \vdots & & \vdots \\
    \text{res}_{q_r} \varphi_1 \tth_r\varpi_r/t & \dots & \text{res}_{q_r} \varphi_g \tth_r\varpi_r/t \\
\end{bmatrix}.
$$
Hence we can invert (non-uniquely) the relationships (\ref{exp2}) encoded in $B'$ by finding any left inverse to $B'$, call it $B = (b_{ij})$ which yields lifts to $\pi_*(\Pi \w^{-1}|_T)$ of the elements $\{ \varphi_1^*, \dots, \varphi_g^* \}$
\begin{equation}
\widetilde{\varphi_k^*} = \sum_{j=1}^r b_{k,j}\varpi_j + \sum_{j=1}^r b_{k,r+j}\tth_j \varpi_j.
\end{equation}

Now in $\pi_*(\Pi \w^{-1}|_T)$ the two bases $\Pi \mathcal{B}_{\w^{-1}|_T}$ and 
$$
\{ \tau_1|_T, \dots, \tau_{r-g}|_T, \widetilde{\varphi_1^*}, \dots, \widetilde{\varphi_g^*} \, | \, \sigma_1|_T, \dots, \sigma_r|_T \},
$$
are related by the matrix $M_{-1/2}$,
\begin{equation} \label{m-1/2_r}
M_{-1/2} = 
\begin{bmatrix}
    \tau_1^{1,-} & \dots & \tau_1^{r,-} & \tau_1^{1,+} & \dots & \tau_1^{r,+} \\
    \vdots && \vdots &\vdots && \vdots \\
    \tau_r^{1,-} & \dots & \tau_r^{r,-} & \tau_r^{1,+} & \dots & \tau_r^{r,+} \\
    b_{1,r+1} & \dots & b_{1,2r} & b_{1,1} & \dots & b_{1,r} \\
    \vdots && \vdots &\vdots && \vdots \\
    b_{g,r+1} & \dots & b_{g,2r} & b_{g,1} & \dots & b_{g,r} \\
    \sigma_1^{1,+} & \dots & \sigma_1^{r,+} & \sigma_1^{1,-} & \dots & \sigma_1^{r,-} \\
    \vdots && \vdots &\vdots && \vdots \\
    \sigma_r^{1,+} & \dots & \sigma_r^{r,+} & \sigma_r^{1,-} & \dots & \sigma_r^{r,-} \\
\end{bmatrix}.
\end{equation}
Therefore the relationship
$$ \ber \{ \tau_1|_T, \dots, \tau_{r-g}|_T, \widetilde{\varphi_1^*}, \dots, \widetilde{\varphi_g^*} \, | \, \sigma_1|_T, \dots, \sigma_r|_T \} = \ber M_{-1/2} \, \delta_{-1/2} $$
along with the canonical isomorphism $\lambda_0 \otimes B(\Pi \w^{-1}|_T) \cong \lambda^{-1}_{-1/2} $, gives 
\begin{equation}
d^{-1}_0 d^{-1}_{-1/2} = \ber M_{-1/2} \delta_{-1/2}.
\end{equation}

\bigskip

\section{The Third Short Exact Sequence}~ \label{thirdSESsec}

We analyze the final short exact sequence \eqref{thirdSES}. In this case as $R^1\pi_*(\w^{-1}) = R^1\pi_*(\w^{-2}) = 0$ and the induced long exact sequence is actually the short exact sequence

$$ 0 \longrightarrow \pi_*(\Pi \w^{-1}) \overset{t}{\longrightarrow} \pi_*\w^{-2} \longrightarrow \pi_*(\w^{-2}|_{T}) \longrightarrow  0. $$
This allows us to quickly identify a basis of $\pi_*(\w^{-2}|_{T})$ coming from the chosen basis $\mathcal{B}_{\w^{-2}}$, namely $\{ \eta_1|_T , \dots, \eta_r|_T \, | \, \psi_1|_T, \dots, \psi_r|_T \}$. This is related to the basis $\{\varpi_k^2 \, | \, \tth_k\varpi_k^2 \}$ by the matrix
\begin{equation} \label{m-1_r}
M_{-1} = 
\begin{bmatrix}
    \eta_1^{1,+} & \dots & \eta_1^{r,+} & \eta_1^{1,-}  & \dots &  \eta_1^{r,-} \\
    \vdots && \vdots & \vdots && \vdots \\
    \eta_r^{1,+} & \dots & \eta_r^{r,+} & \eta_r^{1,-}  & \dots &  \eta_r^{r,-} \\
    \psi_1^{1,-} & \dots & \psi_1^{r,-} & \psi_1^{1,+}  & \dots &  \psi_1^{r,+} \\
    \vdots && \vdots & \vdots && \vdots \\
    \psi_r^{1,-} & \dots & \psi_r^{r,-} & \psi_r^{1,+}  & \dots &  \psi_r^{r,+} \\
\end{bmatrix}.
\end{equation}
Therefore in $B(\w^{-2}|_T)$ we have
$$ \ber \{ \eta_1|_T , \dots, \eta_r|_T \, | \, \psi_1|_T, \dots, \psi_r|_T \} = \ber M_{-1} \, \delta_{-1}$$
and under the identification $\lambda_{-1/2}^{-1} \otimes B(\w^{-2}|_T) \cong \lambda_{-1}$, we get
\begin{equation}
d_{-1}d_{-1/2} = \ber M_{-1} \delta_{-1}. 
\end{equation}

\bigskip

\section{An Expression for $\mu$}~  \label{mu_expression}

The calculations done in Sections \ref{firstSESsec}, \ref{secondSESsec} and \ref{thirdSESsec} gave
\begin{equation*}
    \begin{split}
        d_{1/2}d_0 & = \ber M_0 \, \delta_{0}, \\
        d^{-1}_0 d^{-1}_{-1/2} & = \ber M_{-1/2} \, \delta_{-1/2}, \\
        d_{-1}d_{-1/2} & = \ber M_{-1} \, \delta_{-1}.
    \end{split}
\end{equation*} 

Serre duality yields $d_0 = d_{1/2}$, and by the argument in \cite{vormum} and \cite{vids} one identifies for each $j$,  $\delta_{j/2} = \delta_0$. These facts give by elementary algebra

$$
d_{-1} = \frac{\ber M_{-1}  \,\ber M_{-1/2}}{(\ber M_0)^2} \, d_{1/2}^{5}.
$$
Thus, we obtain an explicit expression for the trivializing section $\mu$:
\begin{thm} \label{mainthm1}
Suppose $\pi: X \to S$ is a family super Riemann surfaces of genus $g \geq 2$ with $\nr$ Ramond punctures such that:
\begin{enumerate}
    \item $\nr > 6g-6$.
    \item The sheaves $R^i\pi_*\w^j$ are locally free for $i=0,1$, $j=-2, -1, 0, 1$.
\end{enumerate}
Then the super Mumford form $\mu$ may be expressed via the sections chosen in (\ref{gens_1}) as
$$ \mu = d_{-1}d_{1/2}^{-5} \frac{(\ber M_0)^2}{\ber M_{-1}  \,\ber M_{-1/2}} \in \lambda_{-1} \lambda_{1/2}^{-5} \cong \oo_S, $$
where $M_0, M_{-1/2}$ and $M_{-1}$ are given by (\ref{m0_r}), (\ref{m-1/2_r}) and (\ref{m-1_r}) respectively.
\end{thm}

\bigskip

\section{A Measure on $\msp$}~ \label{measure_mspr}

Here we follow an idea of E. Witten in \cite{witmeasure}. Thus far we have an explicit formula for the super Mumford form $\mu$, trivializing the line bundle $\lambda_{-1} \lambda_{1/2}^{-5}$ on the moduli space $\msp$. The significance of such a section is that the line bundle $\lambda_{-1}\lambda_{1/2}^{-5}$ is related to the Berezinian of $\msp$.

As discussed in Section \ref{ramondpunctures}, the tangent sheaf to $\msp$ is $R^1\pi_*\mathcal{W}$ where $\mathcal{W}$ is the sheaf of infinitesimal automorphisms, which is seen to be isomorphic (as sheaves of $\mathbb{C}$-vector spaces) to $\mathcal{D}^2$. Hence, by Serre duality and the isomorphism $\w^{-2}(-2\mathcal{F}) \cong \mathcal{D}^2$ one sees that
\begin{equation} \nonumber
\ber \msp = \ber \Omega^1_{\msp} = \ber \pi_*(\w^3(2\mathcal{F})).
\end{equation}
Noting that $R^1\pi_*\w^3(2\mathcal{F}) = 0$, we will write this as
\begin{equation}
\ber \msp = B(\w^3(2\mathcal{F})).
\end{equation}

Trivially we have the short exact sequence
$$
0 \longrightarrow \w^3 \longrightarrow \w^3(2\mathcal{F}) \longrightarrow \w^3(2\mathcal{F})/\w^3 \longrightarrow 0.
$$
By Corollary \ref{trivL} of the Appendix we have canonically the identification
\begin{equation}
B(\w^3(2\mathcal{F})) = B(\w^3) \cong \lambda_{3/2}.
\end{equation}

Now Serre duality identifies $\lambda_{3/2}$ with $\lambda_{-1}$, and thus the super Mumford form $\mu$ can in fact be thought of as a section of $\ber \msp$ valued in a certain line bundle
$$
\mu \in \ber \msp \otimes \lambda_{1/2}^{-5} \cong \stsh_{\msp}.
$$

In bosonic string theory (without punctures), the analogous argument would yield a form similar to $\mu$ such that its modulus squared could genuinely be regarded (in the sense that one had a natural pairing between the analogous factor $\lambda_{1/2}^{-5}$ and its conjugate) as a section of the smooth Berezinian (or simply the determinant in this case) of the moduli space of Riemann surfaces. The celebrated result of Belavin and Knizhnik \cite{belkniz} states that this procedure indeed yields the integrand of the bosonic string partition function, the so-called Polyakov measure.

In superstring theory the super Mumford form $\mu$ plays a similar role to its bosonic counterpart, in that it can be paired with something analogous to its complex conjugate to yield a genuine measure. However, the story is a bit more complicated. The interested reader can learn more in E. Witten's notes \cite{wit4}.
\chapter{The Neveu-Schwarz Puncture Case} \label{ns_case}
\label{chapter_template}


Suppose now we have a family $\pi: X \to S$ of SUSY curves of genus $g \geq 2$ with $n_{NS}$ Neveu-Schwarz punctures. We will reproduce the arguments of \cite{vormum} and \cite{RSVphy} to write down the explicit formula for the associated super Mumford form $\mu$. We then discuss how this form can be used to create a genuine measure on $\mathfrak{M}_{g;n_{NS}}$.

As in the Ramond case, we make local freeness assumptions on the higher direct images $R^i\pi_*\w^j$ and describe $\mu$ in terms of chosen local bases for these sheaves. Here specifically we work with $R^i\pi_*\w^j$ for $i=0,1, \, j=0, 1, 2, 3$. We then relate this to a section of $\ber \mathfrak{M}_{g;n_{NS}}$. As the first part of the following argument is identical to the one found in \cite{vormum} and \cite{RSVphy}, we quickly review the procedure to establish notation but omit some details.

We also assume that we are working over the component of $\mathfrak{M}_{g;n_{NS}}$ corresponding to an odd spin structure. That is, for the relative Berezinian sheaf $\w$ we assume that $\pi_*\w$ has rank $g | 1$. Thus on each fiber the reduction $\Pi \w_{\tred}$ gives an odd nondegenerate theta characteristic.

We choose an odd global section $\nu' \in \w$ and consider the short exact sequence

$$ 0 \longrightarrow \stsh_X \overset{\nu}{\longrightarrow} \Pi\w \longrightarrow (\Pi\w)|_{D} \longrightarrow 0 $$
where $\nu = \Pi \nu'$ and $D = \{ \nu = 0 \} = \{ \nu' = 0\}$. This short exact sequence and the two others obtained by twisting by $\w$ and $\w^2$ is what we focus on. Similar to the argument given in the Ramond case, using these three short exact sequences we can produce the super Mumford isomorphism $\lambda_{3/2} \cong \lambda_{1/2}^5$.

As $\nu'$ is an odd global section of $\w$ its divisor $D$ has the property that its reduction $D_{\tred}$ is a finite sum of $g-1$ points (which we assume to be distinct),
$$
D_{\tred} = \sum_{j=1}^{g-1} p_j.
$$
For each $j$ we choose local superconformal coordinates $z_j \, | \, \zeta_j$ centered at $p_j$.

\bigskip

\section{Bases}~

We choose specific local bases here of the sheaves $R^i\pi_*\w^j$ and $\pi_*\w|_D$ and analyze their compatibility to the above mentioned short exact sequences. The ranks of these various sheaves are easily computable using the same techniques as used in Section \ref{somerierochcalc} of the Ramond case.

\bigskip

\subsection*{Basis in $\pi_*\stsh_X$:}~

The rank of $\pi_*\stsh_X$ is $1 | 1$ and thus we take the local basis 
$$\mathcal{B}_{\stsh_X} := \{ 1 \, | \, \xi \},$$
where $\xi$ expands near each $p_k$ as
$$
\xi \sim ( \xi^{k, -} + \xi^{k, +}\zeta_k + O(z_j) )
$$
where $\xi^{k,\pm}$ are some even/odd functions from the base $S$.

\subsection*{Basis in $\pi_*\w$:}~

The rank of $\pi_*\w$ is $g  | 1$ and we take the local basis 
$$
\mathcal{B}_{\w} := \{ \varphi_1, \cdots, \varphi_{g-1}, \nu' \xi \, | \, \nu' \},
$$
expanding each $\varphi_j$ near $p_k$ as
$$
\varphi_j \sim ( \varphi_j^{k, +} + \varphi_j^{k, -}\zeta_k + O(z_j) ) [dz_j \, | \, d\zeta_j].
$$

\subsection*{Basis in $\pi_*\w^2$:}~

The rank of $\pi_*\w^2$ is $g  | 2g-2$. We take the local basis 
$$
\mathcal{B}_{\w^2} := \{ \nu'^2, \chi_1, \dots, \chi_{g-1} \, | \, \nu' \varphi_1, \cdots, \nu'\varphi_{g-1}, \nu'^2 \xi, \psi_1, \cdots, \psi_{g-2} \},
$$
expanding each $\chi_j$ and $\psi_j$ near $p_k$ as
$$
\chi_j \sim ( \chi_j^{k, +} + \chi_j^{k, -}\zeta_k + O(z_j) ) [dz_j \, | \, d\zeta_j]^2,
$$
$$
\psi_j \sim ( \psi_j^{k, -} + \psi_j^{k, +}\zeta_k + O(z_j) ) [dz_j \, | \, d\zeta_j]^2.
$$

\subsection*{Basis in $\pi_*\w^3$:}~

The rank of $\pi_*\w^3$ is $3g-3  | 2g-2 $. We take the local basis 
\begin{align} \nonumber
\mathcal{B}_{\w^3} := \{ \nu'^2 \varphi_1, \cdots, \nu'^2\varphi_{g-1}, \nu'^3\xi, &\nu'\psi_1, \cdots, \nu'\psi_{g-1}, \sigma_1, \cdots, \sigma_{g-1} \, \\ & | \, \nu'^3, \nu' \chi_1, \cdots, \nu' \chi_{g-1}, \rho_1, \cdots, \rho_{g-2}  \}, \nonumber
\end{align}
expanding each $\sigma_j$ and $\rho_j$ near $p_k$ as
$$
\sigma_j \sim ( \sigma_j^{k, +} + \sigma_j^{k, -}\zeta_k + O(z_j) ) [dz_j \, | \, d\zeta_j]^3,
$$
$$
\rho_j \sim ( \rho_j^{k, -} + \rho_j^{k, +}\zeta_k + O(z_j) ) [dz_j \, | \, d\zeta_j]^3.
$$

\subsection*{Basis in $\pi_*(\w^{-1})$:}~

The rank of $\pi_*(\w^{-1})$ is $1  |  0$ and we take the local basis 
$$
\mathcal{B}_{\w^{-1}} := \{ \xi / \nu'\}.
$$

\subsection*{Basis in $\pi_*\w^j|_{D}$:}~

With the aid of the specific chosen local coordinates $z_j \, | \, \zeta_j$ we take the local basis (for $j \geq 0$)
$$
\mathcal{B}_{\w^j|_D} := \{ [dz_1|d\zeta_1]^j, \cdots, [dz_{g-1}|d\zeta_{g-1}]^j \, | \, \zeta_1[dz_1|d\zeta_1]^j, \cdots, \zeta_{g-1}[dz_{g-1}|d\zeta_{g-1}]^j \}.
$$

Finally we utilize Serre duality, the canonical isomorphisms $R^1\pi_*\w^j \cong (\pi_*\w^{1-j})^*$, to construct local bases for $R^1\pi_*\stsh_X$, $R^1\pi_*\w$, $R^1\pi_*\w^2$, and $R^1\pi_*\w^3$ by taking the image of the corresponding dual bases of those already specified. We denote these bases by $\mathcal{B}_{\w^j}^*$.

We set
$$
\lambda_{j/2} := B(\w^j).
$$
Using these bases, we consider the following local generators of the various $\lambda_{j/2}$,
\begin{equation} \label{gens_2}
\begin{split}
d_0 & := \ber \mathcal{B}_{\stsh_X} \otimes \ber \mathcal{B}^*_{\w} \in \lambda_0 \\
d_{1/2} & : = d_0 \in \lambda_{1/2} \\
d_1 & := \ber \mathcal{B}_{\w^2} \otimes \ber \mathcal{B}^*_{\w^{-1}} \in \lambda_1 \\
d_{3/2} & := \ber \mathcal{B}_{\w^3} \in \lambda_{3/2} \\
\delta_{j/2} & := \ber \mathcal{B}_{\w^j|_D} \in B(\w^j|_D). \\
\end{split}
\end{equation}

\bigskip

\section{Relating the Chosen Bases}~ \label{relate_bases_ns}

The first short exact sequence is
$$ 0 \longrightarrow \stsh_X \overset{\nu}{\longrightarrow} \Pi\w \longrightarrow (\Pi\w)|_{D} \longrightarrow 0 $$
which gives the long exact sequence on cohomology (after using Serre duality)
$$
0 \longrightarrow \pi_*\stsh_X \overset{\nu}{\longrightarrow} \Pi\pi_*\w \longrightarrow \Pi\pi_*(\w)|_{D} \longrightarrow (\pi_*\w)^* \longrightarrow \Pi (\pi_*\stsh_X)^* \longrightarrow 0. 
$$

Following the work of \cite{vormum} or \cite{RSVphy} we conclude under the canonical isomorphism $\lambda_{1/2}^{-1} \cong \lambda_0 \otimes B(\w|_{D})^{-1}$ that
\begin{equation} \label{id1}
d_{1/2}^{-1} = \ber M_1 \, d_0 \delta_{1/2}^{-1}
\end{equation}
where $M_1$ is the block matrix
\begin{equation} \label{m1_ns}
M_1 = 
\begin{pmatrix}
   A_1 \\
   B^t_1
\end{pmatrix}
\end{equation}
where $A_1$ is the $(g-1) \times (2g-2)$ matrix
$$
A_1 = 
\left (
\begin{array}{ccc | ccc}
\varphi_1^{1,+} & \dots & \varphi_1^{g-1,+} & \varphi_1^{1,-} & \dots & \varphi_1^{g-1,-} \\
\vdots & & \vdots & \vdots & & \vdots \\
\varphi_{g-1}^{1,+} & \dots & \varphi_{g-1}^{g-1,+} & \varphi_{g-1}^{1,-} & \dots & \varphi_{g-1}^{g-1,-}
\end{array}
\right )
$$
and $B_1$ is any left inverse of $A_1$.

Similarly, the exact sequence
$$
0 \longrightarrow \Pi\w \overset{\nu}{\longrightarrow} \w^2 \longrightarrow (\w^2)|_{D} \longrightarrow 0,
$$
yields
$$
0 \longrightarrow \Pi \pi_*\w \overset{\nu}{\longrightarrow} \pi_*\w^2 \longrightarrow \pi_*(\w^2)|_{D} \longrightarrow \Pi (\pi_*\stsh_X)^* \longrightarrow (\pi_*\w^{-1})^* \longrightarrow 0.
$$
Arguing again as in \cite{vormum} and \cite{RSVphy} we conclude
\begin{equation} \label{id2}
    d_1 = \ber M_2 \, d_{1/2}^{-1} \delta_{1},
\end{equation}
where $M_2$ is the $(2g-2) \times (2g-2)$ square matrix

\begin{equation} \label{m2_ns}
M_2 = 
\left (
\begin{array}{ccc | ccc}
\chi_1^{1,+} & \dots & \chi_1^{g-1,+} & \chi_1^{1,-} & \dots & \chi_1^{g-1,-} \\
\vdots & & \vdots & \vdots & & \vdots \\
\chi_{g-1}^{1,+} & \dots & \chi_{g-1}^{g-1,+} & \chi_{g-1}^{1,-} & \dots & \chi_{g-1}^{g-1,-} \\
\phantom{asdf} \\
\hline \\
\psi_{1}^{1,-} & \dots & \psi_{1}^{g-1,-} & \psi_{1}^{1,+} & \dots & \psi_{1}^{g-1,+} \\
\vdots & & \vdots & \vdots & & \vdots \\
\psi_{g-2}^{1,-} & \dots & \psi_{g-2}^{g-1,-} & \psi_{g-2}^{1,+} & \dots & \psi_{g-2}^{g-1,+} \\
0 & \dots & 0 & 0 & \dots & 1
\end{array}
\right ).
\end{equation}
\medskip

Lastly, the final short exact sequence
$$
0 \longrightarrow \w^2 \overset{\nu}{\longrightarrow} \Pi\w^3 \longrightarrow (\Pi\w^3)|_{D} \longrightarrow 0
$$
gives
$$
0 \longrightarrow \pi_*\w^2 \overset{\nu}{\longrightarrow} \Pi\pi_*\w^3 \longrightarrow \Pi\pi_*(\w^3)|_{D} \longrightarrow (\pi_*\w^{-1})^* \longrightarrow 0.
$$
Thus, under $\lambda_{3/2}^{-1} \cong \lambda_1 B(\w^3|_D)^{-1}$ we conclude
\begin{equation} \label{id3}
    d_{3/2}^{-1} = \ber M_3 \, d_1 \delta_{3/2}^{-1},
\end{equation}
where $M_3$ is the $(2g-2) \times (2g-2)$ square matrix

\begin{equation} \label{m3_ns}
M_3 = 
\left (
\begin{array}{ccc | ccc}
\rho_1^{1,+} & \dots & \rho_1^{g-1,+} & \rho_1^{1,-} & \dots & \rho_1^{g-1,-} \\
\vdots & & \vdots & \vdots & & \vdots \\
\rho_{g-2}^{1,+} & \dots & \rho_{g-2}^{g-1,+} & \rho_{g-2}^{1,-} & \dots & \rho_{g-2}^{g-1,-} \\
\xi_1^{-1} & \dots & 0 & 0 & \dots & 0 \\
\hline \\
\sigma_{1}^{1,-} & \dots & \sigma_{1}^{g-1,-} & \sigma_{1}^{1,+} & \dots & \sigma_{1}^{g-1,+} \\
\vdots & & \vdots & \vdots & & \vdots \\
\sigma_{g-1}^{1,-} & \dots & \sigma_{g-1}^{g-1,-} & \sigma_{g-1}^{1,+} & \dots & \sigma_{g-1}^{g-1,+}
\end{array}
\right ).
\end{equation}

\medskip 

Under the canonical isomorphism guaranteed by Proposition \ref{can_iso_1_new} of the Appendix, we have that 
$$
\delta_{j/2}^{(-1)^{j-1}}=\delta_{1/2}.
$$
Combining this with the identifications (\ref{id1}), (\ref{id2}) and (\ref{id3}), we get the desired formula for the super Mumford form $\mu$ as in \cite{vormum} and \cite{RSVphy}.
\begin{thm} \label{mainthm2} \emph{(}\cite{vormum}, \cite{RSVphy}\emph{)}
Suppose $\pi:X \to S$ is a family of super Riemann surfaces of genus $g \geq 2$ such that:
\begin{enumerate}
    \item The sheaves $R^i\pi_*\w^j$ are locally free for $i=0,1$, $j=0,1,2, 3$.
    \item $\pi_*\w$ has rank $g|1$.
\end{enumerate}
Then the super Mumford form $\mu$ may be expressed via the sections chosen in (\ref{gens_2}) as
$$
\mu = d_{3/2} d_{1/2}^{-5} \frac{\ber M_3 \, \ber M_2}{(\ber M_1)^2},
$$
where $M_1, M_2$ and $M_3$ are given by (\ref{m1_ns}), (\ref{m2_ns}) and (\ref{m3_ns}) respectively.
\end{thm}

\bigskip

\section{Relation to $\ber \mathfrak{M}_{g;n_{NS}}$}~

In the situation without any punctures, such a formula for $\mu$ is of immediate interest as the Berezinian of $\mathfrak{M}_g$ \emph{is} simply $\lambda_{3/2}$, hence $\mu$ is interpreted as a section of the Berezinian of supermoduli space valued in a certain line bundle. However, when one considers punctures the story is a bit different, as it is no longer true that $\lambda_{3/2}$ is $\ber \mathfrak{M}_{g;n_{NS}}$. In Section \ref{nspuctures_sec} we saw that $\mathcal{T}_{\mathfrak{M}_{g;n_{NS}}} \cong R^1\pi_*\mathcal{W} \cong R^1\pi_*\mathcal{D}^2(-N)$ for $N = \sum_{k=1}^{n_{NS}} \text{div}(s_k)$ the Neveu-Schwarz divisor. Serre duality then gives that instead we have $\ber \mathfrak{M}_{g;n_{NS}} \cong B(\w^3(N)) = \ber \pi_*\w^3(N)$.

Let $N_{\tred} = \sum q_k$ be the reduction. Then each $q_k$ is a divisor in $X_{\tred}$ that is a single point in each fiber of $\pi$. We choose an even global section of $\w^3(N)$, call it $\tau$, that vanishes to exactly first order on each $q_k$. For each $k$ we choose local coordinates $x_k \, | \, \theta_k$ such that $\tau$ near each $q_k$ is
$$
\tau \sim x_k(a_k + b_k\tth_k + O(x_k)) [dx_k \, | \, d\tth_k].
$$

$\tau$ then induces a short exact sequence
$$
0 \longrightarrow \w^3 \overset{\tau}{\longrightarrow}\w^3(N) \longrightarrow \w^3(N)|_N \longrightarrow 0,
$$
which in fact gives the short exact sequence on cohomology
$$
0 \longrightarrow \pi_*\w^3 \overset{\tau}{\longrightarrow} \pi_*\w^3(N) \longrightarrow \pi_*\w^3(N)|_N \longrightarrow 0.
$$

The rank of $\pi_*\w^3(N)$ is $3g-3 + n_{NS} \, | \, 2g-2 + n_{NS}$, thus we construct a local basis for $\pi_*\w^3(N)$ in the following way: first we consider the image of $\mathcal{B}_{\w^3}$ under $\tau$ and complete it to a basis. Namely we construct
$$
\mathcal{B}_{\w^3(N)} := \tau \mathcal{B}_{\w^3} \cup \mathcal{B}'
$$
where $\mathcal{B}'$ is
$$
\mathcal{B}' = \{ \alpha_1, \cdots, \alpha_{n_{NS}} \, | \, \beta_1, \cdots, \beta_{n_{NS}} \}.
$$
We expand each $\alpha_j$ and $\beta_j$ near $q_k$ as
$$
\alpha_j \sim ( \alpha_j^{k, +} + \alpha_j^{k, -}\tth_k + O(x_j) ) [dx_j \, | \, d\tth_j]^3,
$$
$$
\beta_j \sim ( \beta_j^{k, -} + \beta_j^{k, +}\tth_k + O(x_j) ) [dx_j \, | \, d\tth_j]^3,
$$
and let
\begin{align} \nonumber
\mathcal{B}_{\w^3(N)|_N} := & \\ \{ [dx_1|d\tth_1]^3, \cdots, [dx_{n_{NS}}|d\tth_{n_{NS}}]^3 \, | \, &\tth_1[dx_1|d\tth_1]^3, \cdots, \tth_{n_{NS}}[dx_{n_{NS}}|d\tth_{n_{NS}}]^3 \}. \nonumber
\end{align}
Putting
\begin{equation} \label{gens_3}
    \begin{split}
        \delta^N_{3/2} & := \ber \mathcal{B}_{\w^3(N)|_N}, \\ 
        d_{3/2}^N & := \ber \mathcal{B}_{\w^3(N)},
    \end{split}
\end{equation}
we easily see that in the canonical identification 
$$ 
B(\w^3(N)) \cong \lambda_{3/2} \, B(\w^3(N)|_N)
$$
we have
$$
d_{3/2}^N = \ber M' \, d_{3/2} \delta^N_{3/2},
$$
where

\begin{equation} \label{m'}
M' = 
\left (
\begin{array}{ccc | ccc}
\alpha_1^{1,+} & \dots & \alpha_1^{n_{NS},+} & \alpha_1^{1,-} & \dots & \alpha_1^{n_{NS},-} \\
\vdots & & \vdots & \vdots & & \vdots \\
\alpha_{n_{NS}}^{1,+} & \dots & \alpha_{n_{NS}}^{n_{NS},+} & \alpha_{n_{NS}}^{1,-} & \dots & \alpha_{n_{NS}}^{n_{NS},-} \\
\phantom{asdf} \\
\hline \\
\beta_{1}^{1,-} & \dots & \beta_{1}^{n_{NS},-} & \beta_{1}^{1,+} & \dots & \beta_{1}^{n_{NS},+} \\
\vdots & & \vdots & \vdots & & \vdots \\
\beta_{n_{NS}}^{1,-} & \dots & \beta_{n_{NS}}^{n_{NS},-} & \beta_{n_{NS}}^{1,+} & \dots & \beta_{n_{NS}}^{n_{NS},+}
\end{array}
\right ).
\end{equation}
Combining this discussion with that of the Section \ref{relate_bases_ns} we obtain the following corollary.
\begin{cor} \label{maincor} Suppose $\pi:X \to S$ is a family of super Riemann surfaces of genus $g \geq 2$ with $n_{NS}$ Neveu-Schwarz punctures such that:
\begin{enumerate}
    \item The sheaves $R^i\pi_*\w^j$ are locally free for $i=0,1$, $j=0,1,2, 3$.
    \item $\pi_*\w$ has rank $g|1$.
\end{enumerate}
Then via the sections defined in (\ref{gens_2}), (\ref{gens_3}) and matrices in (\ref{m1_ns}), (\ref{m2_ns}), (\ref{m3_ns}), (\ref{m'}), the expression
$$
\mu^N := d^N_{3/2} (\delta^N_{3/2})^{-1} d_{1/2}^{-5} \frac{\ber M_3 \, \ber M_2}{(\ber M_1)^2 \, \ber M'} 
$$
gives a trivializing section of the line bundle
$$
 \ber \mathfrak{M}_{g;n_{NS}} \otimes B(\w^3(N)|_N)^{-1} \otimes \lambda_{1/2}^{-5}
$$
on $\mathfrak{M}_{g; n_{NS}}$.
\end{cor} 

Thus the constructed object $\mu^N$ can be viewed as a section of the Berezinian of the moduli space $\mathfrak{M}_{g:n_{NS}}$ with values in a particular line bundle.

The utility of such a formula for $\mu^N$ is that it indeed can be used to construct a measure on $\mathfrak{M}_{g; n_{NS}}$. The process of constructing this measure depends on the particular type of superstring theory one is working in, heterotic or Type II for example. In \cite{witmeasure} such a process is described, however it assumes the object one starts with is a section of $\ber \mathfrak{M}_{g;n_{NS}} \otimes \lambda_{1/2}^{-5}$ rather than what is given in Corollary \ref{maincor}, in a section of $\ber \mathfrak{M}_{g;n_{NS}} \otimes B(\w^3(N)|_N)^{-1} \otimes \lambda_{1/2}^{-5}$. In calculating scattering amplitudes, one inserts so called vertex operators at each puncture. The collection of them can be thought of as sections of $B(\w^3(N)|_N)$. Hence, after multiplying with the form $\mu^N$ we indeed arrive at a section of $\ber \mathfrak{M}_{g;n_{NS}} \otimes \lambda_{1/2}^{-5}$. The details of this discussion can be found in \cite{wit4} and \cite{witmeasure}.

\include{chapters/conclusion}

\bibliography{thesis}

\appendix
\chapter{A Few Technical Results} \label{appendices}
\label{appendix}

\section{More on the Local Structure Near a Ramond Puncture}~

Here we develop a few technical statements used in the main arguments of the paper. These were motivated by \cite{witmeasure}.

\bigskip

Suppose we have a family of SUSY curves with $n_R$ Ramond punctures $\pi: X \to S$. Denote by $\w = \ber X/S$, the relative Berezinian sheaf. Let $\mathcal{F}$ denote the Ramond divisor and decompose it $\mathcal{F} = \sum_{k}^{n_R} \mathcal{F}_k$ into its $n_R$ minimal components. Recall that near a Ramond divisor $\mathcal{F}_k$ there are coordinates $x \, | \, \tth$ such that the divisor $\mathcal{F}_k$ is given by $\{ x = 0 \}$ and distribution $\mathcal{D}$ is generated by
$$
D^*_\tth := \frac{\partial}{\partial \tth} + x \tth \frac{\partial}{\partial x}.
$$
Here the distinguished subbundle of $\Omega^1_{X/S}$ is $\mathcal{D}^{-2}(-\mathcal{F})$ and admits a generator
$$
\varpi^*_\tth := dx - x\tth d\tth.
$$

Coordinates near a Ramond puncture for which $D^*_\tth$ (or $\varpi^*_\tth$) generate $\mathcal{D}$ (resp. $\mathcal{D}^{-2}(-\mathcal{F})$) are called superconformal. A \emph{superconformal} change of coordinates near a Ramond puncture is a change of coordinates $z \, | \, \zeta$ such that one still has $\mathcal{F}_k = \{z = 0 \}$ and $D^*_{\zeta}$ is a $\stsh_X$-multiple of $D^*_{\tth}$. One can also phrase this condition equivalently as the form $\varpi^*_{\zeta}$ is a multiple of $\varpi_{\tth}^*$.

It turns out that the possible choices of superconformal coordinates near a Ramond puncture is restricted. In fact, this is heavily exploited by E. Witten in \cite{wit3} to define the notion of odd periods of closed holomorphic one-forms on such a family of Ramond punctured SUSY curves. Witten phrases this constraint on coordinates as ``The odd coordinate $\tth$ is defined up to sign and a shift by an odd constant."
\begin{lem} \label{superconformallemma}
Let $x \, | \, \tth$ denote superconformal coordinates near a Ramond puncture $\mathcal{F}_k$. Any superconformal change of coordinates $z \, | \, \zeta$ can be expressed as
$$
z = f(x) + \lambda(x)\tth
$$
$$
\zeta = \psi(x) + g(x)\tth
$$
for even $f,g$ and odd $\psi, \lambda$. We then have 
\begin{enumerate}
    \item \label{g_square_1} $g(0)^2 = 1$, and
    \item \label{lambda_prime_psi_0} $\lambda'(0)\psi(0) = 0$.
\end{enumerate}
\end{lem}
\begin{proof}
After some tedious calculations one finds that the condition $\varpi^*_{\zeta}$ is proportional to $\varpi^*_{\tth}$ is
$$
-\left( \frac{\partial z}{\partial x} - z \zeta \frac{\partial \zeta}{\partial x} \right ) x \tth = \left ( \frac{\partial z}{\partial \tth} - z \zeta \frac{\partial \zeta}{\partial \tth} \right ).
$$
In terms of the functions $f,g, \lambda$ and $\psi$ this condition is the pair of conditions
$$
\lambda(x) - f(x)g(x)\psi(x) = 0,
$$
and
$$
f(x)g(x)^2 + \lambda(x)\psi(x)g(x) = xf'(x) - xf(x)\psi(x) \psi'(x).
$$
The first of these two conditions says that $\lambda$ and $\psi$ are proportional, hence their product vanishes. Using this and dividing by $f(x)$ in the second equation gives (note $f \neq 0$ away from $x = 0$)
\begin{equation} \label{app1}
    g(x)^2 = \frac{x}{f(x)} f'(x) - x \psi(x) \psi'(x).
\end{equation}
As the change of coordinates was superconformal, the divisor $\mathcal{F}_k$ was given as both the zero locus of $x$ and $z$, hence in particular $f(0) = 0$. This implies that the ratio $x/f(x) \to 1/f'(0)$ as $x \to 0$. Hence, taking $x \to 0$ in (\ref{app1}) yields $g(0)^2 = 1$, giving (\ref{g_square_1}). Assertion (\ref{lambda_prime_psi_0}) follows immediately from $\lambda(x) = f(x)g(x)\psi(x)$, recalling that $f(0)=0$ and $\psi(x)^2=0$. 
\end{proof}

Lemma \ref{superconformallemma}, will allow us to trivialize $\ber \pi_*(\oo/\oo(-2\mathcal{F}))$ canonically over the base $S$. Combined with Proposition \ref{can_iso_1_new}, this will give a natural trivialization of $\ber \pi_*(\w^3(2\mathcal{F})/\w^3)$. This result proved significant in Section \ref{measure_mspr} as it allowed us to connect the Mumford form constructed in Theorem \ref{mainthm1} with sections of $\ber \msp$.

\begin{lem} \label{ber_trivial_2f}
The Berezinian of the locally free $\stsh_S$-module $\pi_*(\oo/\oo(-2\mathcal{F}))$ is canonically trivial,
$$ \ber \pi_*(\oo/\oo(-2\mathcal{F})) \cong \stsh_S. $$
\end{lem}
\begin{proof}
We work locally on $S$. First, decompose $\mathcal{F} = \sum_k^{n_R} \mathcal{F}_k$ into its irreducible components. Then $\ber \pi_*(\oo/\oo(-2\mathcal{F})) = \otimes_k \ber \pi_*(\oo/\oo(-2\mathcal{F}_k))$ and so it suffices to show the result for each $\mathcal{F}_k$. To simplify notation, for the remainder of the proof write $\mathcal{F}$ for some $\mathcal{F}_k$.

Choose superconformal coordinates $x \, | \, \tth$ near $\mathcal{F} = \{ x = 0 \}$. With these coordinates one can trivialize $\ber \pi_*(\oo/\oo(-2\mathcal{F}))$ by the element
\begin{equation} \label{nat_elt_ber}
    \sigma_{x|\tth} = [1, x \, | \, \tth, x \tth]
\end{equation}
where $1, x, \tth, x\tth$ in (\ref{nat_elt_ber}) are to be understood as their images in $\oo/\oo(-2\mathcal{F})$. We claim that the element $\sigma_{x|\tth}$ is in fact canonical, in the sense that if $z \, | \, \zeta$ is another choice of superconformal coordinates we have $\sigma_{x|\tth} = \sigma_{z|\zeta}$. Indeed, for such a change of coordinates, write as in Lemma \ref{superconformallemma}
$$
z = f(x) + \lambda(x)\tth,
$$
$$
\zeta = \psi(x) + g(x)\tth.
$$
Looking at their images in the quotient $\oo/\oo(-2\mathcal{F})$, we see that modulo $\oo(-2\mathcal{F})$
\begin{equation}
    \begin{split}
        z & = f'(0)x + \lambda'(0)x\tth, \\
        \zeta & = \psi(0) + \psi'(0)x + g(0)\tth + g'(0)x\tth, \\
        z \zeta & = f'(0)\psi(0)x + f'(0)g(0) x\tth.
    \end{split}
\end{equation}

Hence, in $\pi_*(\oo/\oo(-2\mathcal{F}))$, the change of basis matrix $A$ from $\{1, x \, | \, \tth, x\tth \}$ to $\{ 1, z, \, | \, \zeta, z \zeta \}$ is given by
$$
A = 
\begin{pmatrix}
1 & 0 & \psi(0) & 0 \\
0 & f'(0) & \psi'(0) & f'(0)\psi(0) \\
0 & 0 & g(0) & 0 \\
0 & \lambda'(0) & g'(0) & f'(0)g(0) 
\end{pmatrix}.
$$
Recalling from Lemma \ref{superconformallemma} that $g(0)^2=1$ and $\lambda'(0)\psi(0) = 0$, a quick calculation will show that $\ber A = 1$. Thus the element $\sigma = \sigma_{x|\tth} = \sigma_{z|\zeta}$ is independent of the choice of superconformal coordinates. 

This local argument glues to a global canonical isomorphism $\ber \pi_*(\oo/\oo(-2\mathcal{F})) \cong \stsh_S$.

\end{proof}

Now, with the aid of Proposition \ref{can_iso_1_new} we obtain

\begin{cor} \label{trivL} There is a canonical isomorphism
$$
\left( \ber \pi_*(\oo|_{\mathcal{F}}) \right)^{\otimes 2} \cong \oo_S.
$$
Hence, in particular for $\w$ the relative Berezinian sheaf, we get a natural identification
$$
\ber \pi_*(\w^3(2\mathcal{F})/\w^3) \cong \oo_S.
$$
\end{cor}
\begin{proof}
By Lemma \ref{ber_trivial_2f} $\ber \pi_*(\oo/\oo(-2\mathcal{F})) \cong \stsh_S$ is naturally trivial. On the other hand, by Proposition \ref{can_iso_1_new}
\begin{equation}
    \begin{split}
        \ber \pi_*(\oo/\oo(-2\mathcal{F})) & \cong \ber(\oo|_{\mathcal{F}}) \otimes \ber(\oo(-\mathcal{F})|_{\mathcal{F}}) \\ 
        & \cong \left( \ber \pi_*(\oo|_{\mathcal{F}}) \right)^{\otimes 2}.
    \end{split}
\end{equation}
From here, it follows that
\begin{equation}
    \begin{split}
        \ber \pi_*(\w^3(2\mathcal{F})/\w^3) & \cong \ber(\w^3(2\mathcal{F})|_{\mathcal{F}}) \otimes \ber(\w^3(\mathcal{F})|_{\mathcal{F}}) \\ 
        & \cong \left( \ber \pi_*(\oo|_{\mathcal{F}}) \right)^{\otimes (-2)} \\
        & \cong \oo_S.
    \end{split}
\end{equation}
\end{proof}

\section{A Proof of The Super Mumford Isomorphism} \label{smum_proof}

Here we attempt to explain in detail the canonical super Mumford isomorphism in the spirit of A. Voronov \cite{vormum} and P. Deligne \cite{vids}. We work in the algebro-geometric setting where the fundamental object of interest is a morphism $f: X \to S$ of complex superschemes which is proper and smooth of relative dimension $1|1$, i.e. a family of supercurves. Of course, an interesting case is one of that of a SUSY family.

Let $f:X \to S$ be any morphism of complex superschemes. If $\F$ is locally free on $X$, flat over $S$ then one considers $h(\F)$, the Berezinian of cohomology of $\F$ (above we denoted this by $B(\F)$), which if all higher direct images $R^if_*\F$ are locally free on $S$ is given by
$$
h(\F) = \bigotimes_i (R^if_*\F)^{(-1)^i}.
$$
Let $\w = \ber \Omega^1_{X/S}$ denote the relative Berezinian line bundle, we denote by $\lambda_{j/2} = h(\w^{\otimes j})$. When $f:X \to S$ is a family of supercurves, the super Mumford isomorphism is a canonical isomorphism
$$
\lambda_{3/2} \cong \lambda_{1/2}^{\otimes 5}.
$$
This will come from a study on the nature of the functor $h$.

\begin{prop} \label{can_iso_1_new}
Suppose $f: D \to S$ is a smooth proper morphism of complex superschemes of relative dimension $0|1$. For any line bundle $\Kk$ of rank $1|0$, flat over $S$, we have the canonical isomorphism
$$
h(\Kk) \otimes h(\oo_D)^{-1} \cong \oo_S.
$$
\end{prop}
\begin{proof}
As the morphism $f$ is smooth and proper of relative dimension $0|1$, it follows that $f$ is a finite morphism of degree say $d$, and that $f_*\oo_D$ is locally free on $S$ of rank $d|d$. As such there is a canonical norm map, the morphism of sheaves of groups
\begin{equation} \label{ber_norm}
    \nn_{D/S} : f_*\oo_{D,0}^* \longrightarrow \oo_{S,0}^*
\end{equation}
defined completely analogously as in the classical case. That is, for an open set $U \subset S$ on which $f_*\oo_D$ is trivial, we have for each element $g \in f_*\oo_{D,0}^*$ the $\oo_S$-automorphism $m_g: f_*\oo_D \to f_*\oo_D$ given by multiplication by $g$. The map in (\ref{ber_norm}) is then $g \mapsto \ber m_g$. This definition glues nicely by the properties of the Berezinian.

The morphism (\ref{ber_norm}) induces a morphism
\begin{equation} \label{ber_norm_h1}
    \nn_{D/S} : H^1(S, f_*\oo_{D,0}^*) \longrightarrow H^1(S, \oo_{S,0}^*)
\end{equation}
which should be thought of as a group homomorphism between the group of invertible $f_*\oo_D$-modules to invertible $\oo_S$-modules.

Now as the map $f$ is finite, it follows that the natural map
\begin{equation} \label{nat_1}
     H^1(S, f_*\oo_{D,0}^*) \longrightarrow H^1(D, \oo_{D,0}^*)
\end{equation}
is an isomorphism. In other words, $f_*\Kk$ is invertible as an $f_*\oo_D$-module if and only if $\Kk$ is an invertible $\oo_D$-module. 

Hence, composing the inverse of (\ref{nat_1}) with (\ref{ber_norm_h1}) we get a group homomorphism, which we still denote by $\nn_{D/S}$
\begin{equation} \label{pic_map}
    \nn_{D/S}: \text{Pic}\, D \longrightarrow \text{Pic}\, S.
\end{equation}

Let us elaborate on the map (\ref{pic_map}). To compute $\nn_{D/S}(\Kk)$ one finds an open cover and trivializations $\{ U_j \subset S, \beta_j: f_*\Kk |_{U_j} \to f_*\oo_D|_{U_j} \}$ and then considers the line bundle on $S$ defined by the cocycle $\{ U_i \cap U_j,  \nn_{D/S}(\beta_i \circ \beta_j^{-1}) \}$.

The rest of the proof will go as follows: we will show that there is a natural isomorphism
\begin{equation} \label{norm_iso}
    h(\Kk) \cong h(\oo_D) \otimes \nn_{D/S}(\Kk)
\end{equation}
and then show that in fact, the map of (\ref{pic_map}) is trivial. To see (\ref{norm_iso}) we find and open cover $\{ U_j \}$ of $D$ so that simultaneously $f_*\Kk$ and $f_*\oo_D$ are trivialized as $f_*\oo_D$ and $\oo_S$ modules respectively. Denote these trivializations by
$$
\beta_j: f_*\Kk |_{U_j} \to f_*\oo_D|_{U_j}
$$
$$
g_j: f_*\oo_D |_{U_j} \to \oo_S^{d|d}|_{U_j}.
$$
Then $\{ U_i \cap U_j , \, g_i \circ (\beta_i \circ \beta_j^{-1}) \circ g_j^{-1} \}$ is a cocycle representing $\Kk$ as a locally free $\oo_S$-module. Thus, $\{ \ber \left( g_i \circ (\beta_i \circ \beta_j^{-1}) \circ g_j^{-1} \right ) \} $ are the transition functions for $h(\Kk)$ with respect to this cover. The maps $\beta_i \circ \beta_j^{-1}$ are each automorphisms of $f_*\oo_D|_{U_i \cap U_j}$, and as such they are multiplication by some element $\beta_{i,j}$. Thus the transition functions for $h(\Kk)$ can be written as
\begin{equation}
    \begin{split}
        \{ \ber \left( g_i \circ m_{\beta_{i,j}} \circ g_j^{-1} \right ) \} & = \{ \ber \left( g_i \circ g_j^{-1} \right ) \ber(m_{\beta_{i,j}})  \} \\
        & = \{ \ber \left( g_i \circ g_j^{-1} \right ) \nn_{D/S} (\beta_{i,j}) \} 
    \end{split}
\end{equation}
which is visibly a set of transition functions for the bundle $h(\oo_D) \otimes \nn_{D/S}(\Kk)$. The triviality of the bundle $\nn_{D/S}(\Kk)$ will be shown in a lemma given below.
\end{proof}

\begin{lem}
Suppose $f: D \to S$ is proper and smooth of relative dimension $0|1$. Then the natural norm map
$$
\nn_{D/S} : f_*\oo_{D,0}^* \longrightarrow \oo_{S,0}^*
$$
is the trivial morphism.
\end{lem}
\begin{proof}
 For sufficiently small $V \subset S$, the preimage $\pi^{-1}(V) \subset D$ is isomorphic to a finite product of copies of $V \times \Cx^{0|1} =: V[\alpha]$, $\alpha$ odd (more precisely this is true in the \'etale topology).
$$
\pi^{-1}(V) \cong \prod_{k=1}^n V \times \Cx^{0|1} = \prod_{k=1}^n V[\alpha_k].
$$
Hence, locally any $g \in f_*\oo_{D,0}^*$ is a direct sum $g = \oplus g_k$ of even invertible functions 
$$
g_k = g^0_k + g^1_k \in \oo_V[\alpha_k] = \oo_V \oplus \oo_V\alpha_k
$$
and so for $\ber m_g$ we see that
$$
\ber m_g = \prod_{k=1}^n \ber m_{g_k}.
$$
The matrix of the endomorphism $m_{g_k}$ with respect to the basis $1, \alpha_k$ is 
$$
m_{g_k} = 
\begin{pmatrix}
    g^0_k & 0 \\
    g^1_k & g^0_k
\end{pmatrix},
$$
with $g^0_k$ invertible. Hence, in fact $\ber m_{g_k} = 1_S$ and the claim follows.
\end{proof}

Before we can move on to our main result, we need a few preliminary facts regarding Berezinians.

\begin{lem} \label{basic_props_ber}
Let $f:X \to S$ be a morphism of superschemes, $\F, \mathcal{G}$ locally free $\oo_X$-modules of ranks $m|n$ and $r|s$ respectively, and $\Kk$ an invertible $\oo_S$-module of rank $1|0$. Then
\begin{enumerate}
    \item $\ber_{\oo_X}(\F \otimes \mathcal{G}) \cong \ber_{\oo_X}(\F)^{r-s} \otimes \ber_{\oo_X}(\mathcal{G})^{m-n}$
    \item \label{euler_pullback} $h(\F \otimes f^*\Kk) \cong h(\mathcal{F}) \otimes \Kk^{s\chi(\F)}$, where $s\chi(\F)$ is the super euler characteristic of $\F$.
\end{enumerate}
\end{lem}
\begin{proof}
Property (1) is an easy extension from the classical formula
$$
\text{det}\,_{\oo_X}(\mathcal{V} \otimes \mathcal{W}) \cong \text{det}\,_{\oo_X}(\mathcal{V})^{\text{rk}\,\mathcal{W}} \otimes \text{det}\,_{\oo_X}(\mathcal{W})^{\text{rk}\, \mathcal{V}}.
$$
Then (2) follows from (1) and the projection formula $R^if_*(\F \otimes f^*\Kk) \cong R^if_*\F \otimes \Kk$.
\end{proof}

\begin{prop} \label{ber_coho_linear_prop}
Let $f:X \to S$ be a family of supercurves. Given line bundles $\M$ and $\Ll$ of rank $1|0$, flat over $S$, there is a canonical isomorphism
$$
h(\M \otimes \Ll) \cong h(\M) \otimes h(\Ll) \otimes h(\oo_X)^{-1}.
$$
\end{prop}
\begin{proof}
We carry out this proof in the special case that $f_*\M$ is locally free on $S$, and $\M$ admits global sections on $X$. 

Consider the projective bundle $\Pp_S(f_*\M) = \underline{\text{Proj}}_S(\text{Sym}^{\bullet}((f_*\M)^{\vee}))$. Recall that a $T/S$-point of $\Pp_S(f_*\M)$, where $g: T \to S$ is an $S$-scheme
$$
\begin{tikzcd}
                                  & \Pp_S(f_*\M) \arrow[d, "\pi"] \\
T \arrow[r, "g"] \arrow[ru, "g'"] & S                            
\end{tikzcd}
$$
is precisely the data of an invertible sheaf $\Kk^{\vee}$ on $T$ of rank $1|0$ along with a surjection 
$$
g^*(f_*\M)^{\vee} \longrightarrow \Kk^{\vee} \longrightarrow 0.
$$
Equivalently, it is the data of rank $1|0$ invertible sheaf $\Kk$ along with a short exact sequence of vector bundles
$$
0 \longrightarrow \Kk \longrightarrow g^*f_*\M \longrightarrow \mathcal{Q} \longrightarrow 0,
$$
this is the viewpoint we adopt. If $g: T \to S$ is an $S$-scheme, denote by $X', f', \M', \Ll'$ the corresponding objects pulled back to the family $f': X \times_S T \to T$. 

We will show that given a $T/S$-point $g:T \to S$ of $\Pp_S(f_*\M)$, given by the injection $\Kk \to g^*f_*\M$, one can produce an isomorphism
$$
h(\M' \otimes \Ll') \otimes h(\Ll')^{-1} \cong h(\M') \otimes h(\oo_{X'})^{-1}
$$
We will denote this corresponding isomorphism by $\alpha_{T}(\Kk \to g^*f_*\M)$.

The desired isomorphism for the lemma will come by choosing any $S$-point. We will then show that the resulting isomorphism was in fact independent of this choice.  The argument is as follows: the base change diagram is
$$
\begin{tikzcd}
X' \arrow[d, "f'"] \arrow[r, "g'"] & X \arrow[d, "f"] \\
T \arrow[r, "g"]                   & S               
\end{tikzcd}
$$
and we have a natural isomorphism $g^*f_*\M \cong f'_*(g')^*\M = f'_*\M'$ and furthermore the map $\Kk \to f'_*\M'$ corresponds to a map $(f')^*\Kk \to \M$. To avoid unpleasant notation, let us simply drop the prime and write $f = f', X = X'$ etc. Let $D$ be the divisor of this section of $\M \otimes f^*\Kk^{-1}$, then it is an effective relative Cartier divisor and we obtain a short exact sequence
$$
0 \longrightarrow f^*\mathcal{\Kk} \longrightarrow \M \longrightarrow \M|_D \longrightarrow 0 
$$
and a similar one by tensoring with $\Ll$
$$
0 \longrightarrow \Ll \otimes f^*\mathcal{\Kk} \longrightarrow \Ll \otimes \M \longrightarrow \Ll \otimes \M|_D \longrightarrow 0.
$$
These imply 
$$
h(\M) \otimes h(f^*\Kk)^{-1} \cong h(\M|_D),
$$
$$
h(\Ll \otimes \M) \otimes h(\Ll \otimes f^*\Kk)^{-1} \cong h(\Ll \otimes \M|_D).
$$
Now by Proposition \ref{can_iso_1_new}, in fact the right hand sides are canonically identified and thus by property \ref{euler_pullback} of Lemma \ref{basic_props_ber} we conclude
$$
h(\M) \otimes h(\oo_X)^{-1} \otimes \Kk^{-s\chi(\oo_X)} \cong h(\Ll \otimes \M) \otimes h(\Ll)^{-1} \otimes \Kk^{-s\chi(\Ll_X)}.
$$
Finally we recall that the super euler characteristic is in fact constant for a family of supercurves (Theorem \ref{RieRoch1} above), so $s\chi(\oo_X) = s\chi(\Ll_X)$ and we obtain the isomorphism $\alpha_{T}(\Kk \to g^*f_*\M)$.

The construction outline above constructs a map from the $T/S$-points of $\Pp_S(f_*\M)$ to the set
$$
\text{Isom}_{\oo_T}(h(\Ll' \otimes \M') \otimes h(\Ll')^{-1} , h(\M') \otimes h(\oo_{X'})^{-1}).
$$

By assumption, we can find a global section $t \in \Gamma(X, \M)$, which we view as a nowhere zero morphism $t: \oo_S \to f_*\M$. For any base change $g: T \to S$, this gives rise to natural section $t' : \oo_{T} \to f'_*\M' = g^*f_*\M$ and hence a natural $T/S$ point of $\Pp_S(f_*\M)$
$$
0 \longrightarrow \oo_T \overset{t'}{\longrightarrow} g^*f_*\M \longrightarrow \mathcal{Q} \longrightarrow 0.
$$
Thus for each $S$-scheme $g:T \to S$, the section $t$ gives a distinguished element $\alpha_T^0 : = \alpha_{T}(\oo_T \overset{t'}{\to} g^*f_*\M)$. As any two isomorphisms of line bundles differ by a global invertible function, we can view our construction as a map from the $T/S$-points of $\Pp_S(f_*\M)$ to the set $\Gamma(T, \oo_T^*)$, associating the $T/S$-point $g:T \to S$, $\Kk \to g^*f_*\M$ to the unique global invertible function $\lambda$ on $T$ so that
$$
\alpha_{T}(\Kk \to g^*f_*\M) = \lambda \alpha_T^0.
$$
We will denote this $\lambda$ by $\lambda_{T}(\Kk \to g^*f_*\M)$. One can easily check that this assignment
$$
\{ g:T \to S, \,\, \Kk \to g^*f_*\M \} \longmapsto \lambda_{T}(\Kk \to g^*f_*\M)
$$
is functorial in $T$, and hence by Yoneda's lemma, this is equivalent to a $S$-morphism
$$
\Lambda: \Pp_S(f_*\M) \longrightarrow \mathbb{G}_{m,S}
$$
which in turn is equivalent to a choice of global section $\gamma \in \Gamma(\Pp_S(f_*\M), \oo_{\Pp_S(f_*\M)})$. We claim $\gamma = 1$. Indeed, let $\sigma$ be the section associated to the distinguished $S/S$-point given by $t:\oo_S \to f_*\M$. We then have the commutative diagram
$$
\begin{tikzcd}
                                                & \Pp_S(f_*\M) \arrow[d, "\pi"] \arrow[r, "\Lambda"] & {\mathbb{G}_{m,S}} \arrow[ld] \\
S \arrow[r, "\text{id}_S"] \arrow[ru, "\sigma"] & S                                                  &                              
\end{tikzcd}
$$
which on global sections reads
$$
\begin{tikzcd}
                                     & {\Gamma(\Pp_S(f_*\M), \oo_{\Pp_S(f_*\M)}))} \arrow[ld, "\sigma^{\#}"'] & {\Gamma(S,\oo_S^*)[x,x^{-1}]} \arrow[l, "\Lambda^{\#}"] \\
{\Gamma(S,\oo_S^*)}  & {\Gamma(S,\oo_S^*)} \arrow[l, "1_S"] \arrow[ru] \arrow[u, "\pi^{\#}"]                   &                                                        .
\end{tikzcd}
$$
Note that the morphism $\pi^{\#}: \Gamma(S, \oo_S^*) \to \Gamma(\Pp_S(f_*\M), \oo_{\Pp_S(f_*\M)}))$ is an isomorphism by basic properties of projective space, and hence $\sigma^{\#}$ is its inverse. 

Hence $\gamma = \Lambda^{\#}(x) = 1$ if and only if $\sigma^{\#}(\gamma) = 1_S$, but this is trivial since by construction and functoriality we have that
\begin{equation}
    \begin{split}
        \sigma^{\#}(\gamma) & = (\sigma^{\#} \circ \Lambda^{\#}) (x) \\
        & = \lambda_{S}(\oo_S \overset{t}{\to} f_*\M) \\
        & = 1_S.
    \end{split}
\end{equation}

The knowledge that $\gamma = 1$, then implies the desired isomorphism is independent of the $S$-point chosen. Indeed for any $S/S$-point $\text{id}: S \to S$ of $\Pp_S(f_*\M)$, given by $\Kk \to f_*\M$, we get a corresponding section $\sigma'$ of $\pi$. By construction we then have $(\sigma')^{\#}(\gamma) = \lambda_{S}(\Kk \to f_*\M) = 1_S$. This completes the proof.

\end{proof}

This immediately gives a proof the super Mumford isomorphism (Proposition \ref{smumiso}).

\begin{proof}(Proof of Proposition \ref{smumiso}) We apply Proposition \ref{ber_coho_linear_prop} to $\Ll = \Pi \w$ and $\mathcal{M} = \Pi^{j-1} \w^{j-1}$ and obtain utilizing Serre duality
$$
\lambda_{j/2}^{(-1)^j} = \lambda_{1/2}^{-2}\lambda_{(j-1)/2}^{(-1)^{j-1}}.
$$
Inducting on $j$ yields the result. In particular, $\lambda_{3/2} \cong \lambda_{1/2}^5$.

\end{proof}
\include{chapters/app_glossary}

\end{document}